\documentclass{article}
\usepackage[T1]{fontenc}
\usepackage[english]{babel}
\usepackage[margin=1.in]{geometry}
\usepackage[utf8]{inputenc}
\usepackage{lmodern}
\usepackage{microtype}

\usepackage{amsmath,amssymb,amsfonts,amsthm}
\usepackage{mathtools}
\usepackage{mathrsfs} 
\usepackage[square,comma,numbers]{natbib}
\usepackage[colorlinks=true, allcolors=blue]{hyperref}
\usepackage{graphicx}
\usepackage{enumerate}
\usepackage[font=small]{caption}
\usepackage{subcaption}
\usepackage{xcolor}
\usepackage{tikz}
\usepackage{float}

\numberwithin{equation}{section}

\def\cE{{\mathcal E}}

\def\cN{{\mathcal N}}
\def\cO{{\mathcal O}}

\def\cS{{\mathcal S}}

\def\cZ{{\mathcal Z}}

\def\E{\mathbb{E}}
\def\F{\mathbb{F}}

\def\N{\mathbb{N}}

\def\P{\mathbb{P}}
\def\R{\mathbb{R}}

\def\sB{\mathscr{B}}
\def\sC{\mathscr{C}}
\def\sF{\mathscr{F}}
\def\sH{\mathscr{H}}
\def\sL{\mathscr{L}}
\def\sP{\mathscr{P}}

\def\sX{\mathscr{X}}

\def\ba{\boldsymbol{a}}
\def\bc{\boldsymbol{c}}

\def\lba{\underline{a}}
\def\lbz{\underline{z}}

\renewcommand{\limsup}{\overline{\mathrm{lim}}}

\renewcommand{\liminf}{\underline{\mathrm{lim}}}

\def\fB{\mathfrak{B}}

\def\fs{\mathfrak{s}}

\renewcommand{\d}{\mathrm{d}}
\newcommand{\da}{\mathrm{d}a}
\newcommand{\ds}{\mathrm{d}s}
\newcommand{\dt}{\mathrm{d}t}
\newcommand{\dx}{\mathrm{d}x}
\newcommand{\dz}{\mathrm{d}z}

\newtheorem*{Def*}{Definition}
\newtheorem*{Thm*}{Theorem}
\newtheorem*{Cor*}{Corollary}
\newtheorem*{Rmk*}{Remark}
\newtheorem*{Lem*}{Lemma}
\newtheorem*{Prop*}{Proposition}
\newtheorem*{Asm*}{Assumption}

\newtheorem{Def}{Definition}[section]
\newtheorem{Thm}[Def]{Theorem}
\newtheorem{Cor}[Def]{Corollary}
\newtheorem{Rmk}[Def]{Remark}
\newtheorem{Lem}[Def]{Lemma}
\newtheorem{Prop}[Def]{Proposition}
\newtheorem{Asm}[Def]{Assumption}

\def\be{\begin{equation}}
\def\ee{\end{equation}}

\usepackage{fancyhdr}
\title{Stationary Heterogeneous-Agent Models in Continuous Time\footnote{The author thanks Markus Brunnermeier and Sebastian Merkel for their invaluable guidance and support. The author is grateful to Yuki Shigeta for helpful comments.
Research partially supported 
by the National Science Foundation grant DMS 2406762.}}

\author{Felix H\"ofer\footnote{Department of Operations Research and Financial
Engineering, Princeton University, Princeton, NJ, 08540, USA, email: 
{\tt fhoefer@princeton.edu}}}
\date{}

\addtocounter{footnote}{-1} 

\begin{document}
\maketitle

\begin{abstract}
\noindent
We study a classical Bewley-Huggett-Aiyagari model in continuous time in which ex-post heterogeneity arises due to idiosyncratic, uninsurable income shocks. Our framework is rooted in the Fiscal Theory of the Price Level (FTPL), and we investigate the existence and multiplicity of stationary equilibria in models with and without capital. We establish the existence of an arbitrary even number of equilibria in which the government runs small constant deficits, which in turn implies the multiplicity of price levels. 
\end{abstract}

\section{Introduction}
We study a heterogeneous-agent\footnote{Households are heterogeneous because of income shocks that are idiosyncratic to each household, not due to different utility functions or dynamics. This should not be confused with heterogeneity modeled through graphon games.} economy with monetary and fiscal authorities who issue nominal government debt and levy taxes, respectively. We follow Kaplan, Nikolakoudis, and Violante \cite{kaplan2023price} to introduce household heterogeneity in the tradition of Bewley-Huggett-Aiyagari \cite{bewley1986stationary,huggett1993risk,aiyagari1994uninsured} in a model of Fiscal Theory of the Price Level (FTPL). Such heterogeneous-agent models have become central to modern macroeconomics, see for example Cherrier, Duarte, and Sa\"{i}di \cite{CHERRIER2023104497}. As emphasized in \cite{kaplan2023price}, while representative agent models only allow stationary equilibria in which governments run primary surpluses, the models we consider have equilibria with primary deficits---a prominent feature of empirical data from the United States since 1970 and Japan. In such equilibria, the real interest rate on government debt $r$ is less than the growth rate of the economy $g$. 

Household heterogeneity is driven by idiosyncratic endowment shocks, which, together with the price-taking assumption of households, naturally lends itself to a mean-field game (MFG) interpretation. The MFG literature, initiated independently by Caines, Malhamé and Huang~\cite{HMC1,HMC2,HMC3,HMC4}, and Lasry and Lions~\cite{LL1,LL2,LL3}, motivates our two-step equilibrium construction: we first establish the well-posedness of the household problem for given prices and policies and then impose market clearing. This mirrors the fixed-point structure of MFGs (see Carmona and Delarue \cite{carmona2018probabilistic})---a connection we make explicit in Section \ref{ssec:MFG-connection}.

We build on Achdou, Han, Lasry, Lions and Moll \cite{achdou2022income} to provide a mathematical analysis of stationary heterogeneous-agent equilibria, and our work complements the analysis of \cite{kaplan2023price}, which studies a Huggett model. We use techniques from the theory of viscosity solutions to establish regularity properties of households' optimal decisions and of the stationary cross-sectional distribution. We then apply these results to study the existence and uniqueness of stationary equilibria in models with capital (Aiyagari) and without (Huggett). In particular, 
\begin{itemize}
\item For a given strictly positive primary surplus, there exists a unique equilibrium in both Huggett and Aiyagari models, see Theorem \ref{thm:fixed-tau-surplus}.

\item For sufficiently small deficits, there exists an even number of equilibria. In fact, for any even number $d\in 2\N$, we construct Huggett economies with $d$ many equilibria, see Theorem \ref{thm:fixed-tau-deficit} \eqref{thm:fixed-tau-deficit-multiple-equilibria}. 

\item Finally, we prove that Aiyagari equilibria converge to Huggett equilibria as the capital elasticity of output tends to zero, see Theorem \ref{thm:huggett-to-aiyagari}.
\end{itemize}
These results rely on properties of aggregate savings $A(r)$ as a function of the interest rate, which are of interest in their own right:
\begin{itemize}
\item We provide an explicit interest rate criterion below which the economy is supported at the borrowing constraint, see Theorem \ref{thm:properties-G} \eqref{thm:properties-G-hand-to-mouth}.

\item We provide a direct proof that aggregate savings $A(r)$ diverge as the interest rate $r$ approaches the households' discount rate $\rho$, $\lim_{r\uparrow\rho}A(r)=\infty$, see Theorem \ref{thm:properties-G} \eqref{thm:properties-G-explosion-at-rho}.

\item In the case of CRRA utility, $A(r)$ is increasing in the interest rate $r$ if the CRRA constant $\gamma$ satisfies $\gamma\leq1$, see Theorem \ref{thm:properties-G}\eqref{thm:properties-A-increasing-in-r}. Conversely, for any $\gamma>1$, there exists an economy for which $A(r)$ is not increasing, see Theorem \ref{thm:A-can-decrease-gamma-greater-one}. 
\end{itemize}

\noindent 
We continue to introduce the models in Subsection \ref{ssec:models}, state our assumptions in Subsection \ref{ssec:assumptions}, and present our results on the household problem, the stationary cross-sectional distribution, and equilibria in Subsections \ref{ssec:main-results-hh}, \ref{ssec:main-results-invariant-distribution}, and \ref{ssec:main-results-equilibria}, respectively. We discuss the related literature in Subsection \ref{ssec:literature} and the MFG connection in Subsection \ref{ssec:MFG-connection}.
\vspace{1em}

{\bf Notation.} The set of non-negative real numbers is denoted by $\R_+:=[0,\infty)$.  On any topological space $X$, the Borel $\sigma$-algebra is $\sB(X)$ and the set of probability measures is $\sP(X)$, equipped with the topology of weak convergence. For any function $\varphi:X\times Y\mapsto U$, with a slight abuse of notation, we define the function $\varphi(x):Y\mapsto U$ by $\varphi(x)(y):=\varphi(x,y)$. The positive and negative part of a real number $x$ are denoted by $x^+=x\lor0$ and $x^-=(-x)\lor0$, respectively. The derivative of a function $\varphi(x,y)$ with respect to $x$ is denoted by $\varphi_x(x,y)$. The terms ``increasing'' and ``decreasing'' are used weakly to mean ``non-decreasing'' and ``non-increasing'', respectively.

\subsection{Huggett and Aiyagari models} \label{ssec:models}

Both the Huggett and Aiyagari economies are set in continuous time, and we begin by describing the Huggett economy. 
Aggregate output consists of a single consumption good. Without loss of generality, we assume that aggregate output remains constant over time—that is, its growth rate $g$ is normalized to $0$.
There is a continuum of households, each receiving a random idiosyncratic endowment share. These endowments are independent and identically distributed across households, and we use the term \emph{income} interchangeably to refer to them. We let $0<z_1<\ldots<z_d$ be fixed endowment levels. Idiosyncratic endowments of a typical household are modeled by a stationary Markov chain $(z_t)_{t\ge0}$ taking values in $\cZ:=\{z_1,\ldots,z_d\}$. Let $Z>0$ denote the mean of the stationary process $(z_t)$.

The government issues bonds $(\fB_t)_{t\ge0}$ and sets a constant nominal interest rate $i$. The nominal value of bonds evolves according to
$$
\d \fB_t = \mu^\fB \fB_t\,\d t ,\qquad t\geq0,
$$
where $\mu^\fB$ is the growth rate set by the government. 
Bonds are the only assets in this economy, and we use them as the num\'eraire. 
At time $t\geq0$, the price of the consumption good in terms of bonds is denoted by $P_t$. It satisfies the Fisher equation,
$$
\d  \left( \!\frac{1}{P_t} \!\right) = (r-i) \frac{1}{P_t}\, \dt,\qquad t\ge0.
$$
Upon receiving asset income $ia_t$ and an idiosyncratic endowment share $z_t$, a typical household pays an income tax $P_t\tau(z_t)$ (or receives a transfer if $\tau(z_t) < 0$), consumes an amount $c_t > 0$, and invests the remaining nominal savings
$$
ia_t +P_t[z_t-\tau(z_t)] - P_t c_t,
$$
in bonds. Households maximize discounted lifetime utility over consumption subject to a no-borrowing constraint. To close the economy, the government faces the nominal budget constraint
$$
\check\mu^\fB \fB_t + P_tT=0,\qquad t\geq0,
$$
where $\check \mu^\fB:=\mu^\fB-i$ and $T$ are the aggregate taxes (subsidies if negative), which we will refer to as \emph{primary surpluses} (deficits if negative) below.

For our analysis, it is convenient to directly work with \emph{real} quantities (in units of the good) instead of nominal ones (in units of the bond), and we directly define equilibria in the real formulation. 
If we let $B_t:=\fB_t/P_t$, $t\ge0$, denote the real government debt, then the (nominal) government's budget constraint implies that $\d B_t=0$, so that $B_t\equiv B$ for some $B$. Using the Fisher equation, 
$$
0= \d B_t = (\check \mu^{\fB}+r) B_t\,\dt \qquad \Longrightarrow \qquad r=-\check \mu^{\fB},
$$
where we recall $\check\mu^{\fB}= \mu^{\fB}-i$.
In real terms, the government's budget constraint becomes $rB=T$. 
\vspace{1em}

We continue with a precise definition of (real) stationary Huggett equilibria. Fix a filtered probability space $(\Omega,\sF,(\sF_t)_{t\ge0},\P)$ which is rich enough to support the stationary Markov chain $(z_t)_{t\ge0}$. The transition rate from state $z$ to $y\neq z$ is denoted by $\lambda(z,y)\geq0$. We define the \emph{state space} $\sX$ as the wealth-endowment space
$$
\sX := \R_+\times\cZ.
$$
We furthermore let $u:(0,\infty)\mapsto\R$ be a utility function. Precise assumptions and definitions on the income process $(z_t)$, the utility function and admissible consumption rules are stated below.\\

{\bf (Huggett).} A \emph{(real) stationary Huggett equilibrium} is a tuple
$$
\Xi^*=(\tau^*(\cdot), B^*, r^*, c^*(\cdot,\cdot), G^*(\d a, \d z))
$$
consisting of a \emph{tax-and-transfer function} $\tau^*(\cdot)$, a \emph{real value of bonds} $B^*\in[0,\infty)$, a \emph{real interest rate }$r^*\in\R$, a \emph{consumption rule} $c^*(a,z)$, and a \emph{cross-sectional distribution} $G^*\in\sP(\sX)$, such that the following hold:

\begin{enumerate}[(1)]
\item \label{def:Hug-I}  \emph{Household optimality.} Given $(r^*,\tau^*(\cdot))$, the map $(a,z)\mapsto c^*(a,z)$ is an optimal feedback control of the following stochastic control problem
$$
\sup_{(c_t)} \quad \E\left[\int_0^\infty \!\!e^{-\rho t} u(c_t)\,\d t\right] \quad \text{subject to} \quad  \d a_t = [r^* a_t + z_t - \tau^*(z_t)-c_t]\,\d t.
$$
Here, the supremum is taken over consumption streams $(c_t)$ that respect the no-borrowing constraint $a_t\ge0$ for all $t\ge0$. 

\item \label{def:Hug-II} \emph{Consistency.} $G^*(\d a, \d z)\in\sP(\sX)$ is a stationary distribution of the Markov process $(a^*_t,z_t)_{t\ge0}$ where
$$
a^*_t = a^*_0+\int_0^t[r^*a^*_s+z_s-\tau^*(z_s)-c^*(a^*_s,z_s)]\,\d s,\qquad t\ge0.
$$

\item \label{def:Hug-III}\emph{Government's budget constraint.} Aggregate real debt $B^*$ satisfies
$$
r^* B^* = \int_\sX\tau^*(z) \, G^*(\d a, \d z).
$$

\item \label{def:Hug-IV} \emph{Market clearing.} Both the \emph{asset market} and \emph{goods market }clear:
$$
A^*:= \int_\sX a\, G^*(\d a, \d z) = B^*\quad\text{and}\quad  C^* := \int_\sX c^*(a,z)\, G^*(\d a, \d z) = Z,
$$
where we recall that $Z$ denotes the mean of $(z_t)$.  We impose that $A^*$ and $C^*$ are well-defined in $[0,\infty)$.
\end{enumerate}

As explained in \cite{kaplan2023price}, real monetary equilibria $\Xi^*=(\tau^*(\cdot), B^*, r^*, c^*(\cdot,\cdot), G^*(\d a, \d z))$ with $B^*>0$ uniquely determine the price level. Indeed, let  $\fB_0>0$ be an initial value for nominal debt and $i^*$ be a value for the nominal interest rate. Then, the equilibrium price level $(P^*_t)_{t\ge0}$ is given as the unique solution to
$$
P^*_0 := \frac{\fB_0}{B^*},\quad \d P^*_t = (i^*-r^*) P^*_t\,\d t,\qquad t\ge0.
$$
We hence focus our treatment on real equilibria and now turn to an Aiyagari model with production. In addition to the setting of the Huggett model, there is a representative firm that rents capital $K$ at rate $r$ and pays a wage $w$ for $L$ amount of labor. Firms use a concave technology $F\in C^1(\R_+\times\R_+;\R_+)$ with constant returns to scale. Under the standard assumptions on $F$ such as $\partial_K F(K,1)\to\infty$ as $K\downarrow0$, in any competitive equilibrium, profit maximization implies 
$$
r=\partial_K F(K,L)-\delta,\quad w=\partial_L F(K,L),\qquad K,L\ge0,
$$
where $\delta>0$ is a constant capital depreciation rate. Note that this in particular forces $r>-\delta$ and $w>0$, see \cite[Lemma 2]{acikgoz2018existence}.
\vspace{1em}

\textbf{(Aiyagari)}. A \emph{(real) stationary Aiyagari equilibrium} is a tuple
$$
\Xi^* = (\tau^*(\cdot), B^*, K^*, r^*, w^*, c^*(\cdot,\cdot), G^*(\d a, \d z)),
$$
consisting of a \emph{tax-and-transfer function} $\tau^*(\cdot)$, a \emph{real value of bonds} $B^*\in[0,\infty)$, a \emph{real value of capital} $K^*\in(0,\infty)$, a \emph{real interest rate }$r^*\in(-\delta,\infty)$, a \emph{wage rate} $w^*\in(0,\infty)$, a \emph{consumption rule} $c^*(a,z)$ and a \emph{cross-sectional distribution} $G^*\in\sP(\sX)$, such that the following hold:

\begin{enumerate}[(1)]
\item \emph{Household optimality.} \label{def:aiyagari-1} Given $(r^*,w^*,\tau^*(\cdot))$, the function $(a,z)\mapsto c^*(a,z)$ provides an optimal feedback consumption rule of the household's problem,
$$
\sup_{(c_t)} \quad \E\left[\int_0^\infty \! e^{-\rho t} u(c_t)\,\d t\right] \quad \text{subject to} \quad  \d a_t = (r^* a_t + w^*[z_t - \tau^*(z_t)]-c_t)\,\d t.
$$
Again, the maximization is taken over admissible consumption controls satisfying the no-borrowing constraint.

\item \emph{Consistency.} \label{def:aiyagari-2} $G^*(\d a, \d z)$ is a stationary distribution of the optimal state process $(a^*_t,z_t)_{t\ge0}$ where
$$
a^*_t = a^*_0+\int_0^t(r^*a^*_s+w^*[z_s-\tau^*(z_s)]-c^*(a^*_s,z_s))\,\d s,\qquad t\ge0.
$$

\item \label{def:aiyagari-3}\emph{Government's budget constraint.} Aggregate real bonds $B^*$ satisfy $B^*\geq0$ and
$$
r^* B^* =  w^*\!\int_\sX\tau^*(z) \, G^*(\d a, \d z).
$$

\item \emph{Market clearing.} \label{def:aiyagari-4} Both the \emph{asset} and \emph{goods market} clear:
$$
A^* := \int_\sX a\, G^*(\d a, \d z) = K^* + B^*\quad\text{and}\quad C^* +I^* = F(K^*,Z),
$$
where aggregate consumption is $C^* := \int_\sX c^*(a,z)\, G^*(\d a, \d z)$, and \emph{investments} $I$ are defined by $I^* := \delta K^*$. We again impose that $A^*$ and $C^*$ are well-defined in $[0,\infty)$.

\item \emph{Competitive equilibrium between firms.} \label{def:aiyagari-5} The real interest rate $r^*>-\delta$ and wage $w^*>0$ satisfy
\begin{equation}\label{eq:aiy-market-prices-eq}
r^* = \partial_K F(K^*,Z)-\delta\quad \text{and}\quad w^* = \partial_L F(K^*,Z).
\end{equation}
\end{enumerate}

\begin{Rmk}
{\rm
It is commonly assumed in the economic literature that a \emph{continuum of households} is given, often on the unit interval $[0,1]$, and that each household $i\in[0,1]$ carries their own independent idiosyncratic random process $z^i_t$, see, for example, Bewley \cite{bewley1986stationary}. Aggregation is then taken with respect to the Lebesgue measure on $i\in[0,1]$. This approach is mathematically intricate and subtle measurability issues arise. While  this approach can be made precise using Fubini extensions, we avoid these issues entirely by treating the problem as a \emph{mean-field game}, and instead average over $\omega\in\Omega$ randomness. Under appropriate assumptions, both approaches yield the same solution concept. In particular, identities such as
$$
\int_\Omega z_t(\omega)\,\P(\d\omega) = \int_{[0,1]}z^i_t\,\d i
$$
can be proven using the exact law of large numbers, see for example Sun \cite{sun2006exact}.
}
\end{Rmk}

\noindent 
We emphasize that, in both the Huggett and Aiyagari model, households take prices $(r^*,w^*)$ and a tax-and-transfer function $\tau^*(\cdot)$ as given and then optimize over their consumption. In equilibrium, prices have to be consistent with the cross-sectional distribution of the optimally controlled state process $(a^*_t,z_t)_{t\geq0}$ via the market clearing conditions. This philosophy is precisely the idea underpinning the theory of mean-field games, and we discuss this connection in Section \ref{ssec:MFG-connection}.

\subsection{Assumptions} \label{ssec:assumptions}

\begin{Asm}[Standing]\label{asm:standing}
\phantom{.}
{\rm
\begin{enumerate}[(i)]
\item $(z_t)_{t\geq0}$ follows an irreducible stationary Markov chain on $\cZ=\{z_1,\ldots,z_d\}$,  for some $d\ge2$, and we normalize its mean to $Z=\E[z_t]=1$. 
\item In the Huggett model, the wage is $w=1$. In the Aiyagari model, the wage $w$ is constrained to be strictly positive: $w>0$.
\item We assume a linear tax-and-transfer function, $\tau(z)=\tau z$ for some $\tau\in(-\infty,1)$.
\item In the Aiyagari model, the production function is of Cobb-Douglas-type, $F(K,L) = K^\alpha L^{1-\alpha}$ for some $\alpha\in(0,1)$.
\end{enumerate}
}
\end{Asm}

\noindent
Note that under the normalization $Z=1$ and the linear tax-and-transfer function $\tau(z)=\tau z$, the parameter $\tau<1$ coincides with the primary surplus $\int \tau z \,G(\da,\dz)$ for any stationary distribution $G$ in the Huggett model. In the Aiyagari model, aggregate primary surpluses are $w\tau$.

\begin{Asm}[Borrowing limit]\label{asm:borrowing-constraint}
\phantom{.}
{\rm
The borrowing limit $\underline{a}\leq 0$ satisfies $\underline{a}>-(1-\tau)w\lbz/r$ if $r>0$.  Here, $\lbz>0$ denotes the lowest income level.
}
\end{Asm}

\noindent
If $\lba<0$, this imposes an upper bound on interest rates, and we define the set of admissible interest rates by
\be\label{eq:admissible-interest-rates-general}
R(w,\tau) := \{r< \rho : r\lba + w(1-\tau)\lbz>0\}.
\ee
Under the \emph{no-borrowing limit}, i.e.\ $\lba=0$, $R(w,\tau)=(-\infty,\rho)$, which is independent of $(w,\tau)$.

\begin{Asm}[Utility function]\label{asm:utility-function}
{\rm
\phantom{.}
The utility function $u:(0,\infty)\mapsto\R$ is $C^2$, increasing, strictly concave, and satisfies $\lim_{c\to\infty}u'(c)=0$ and $\lim_{c\downarrow0}u'(c)=\infty.$
}
\end{Asm}

\begin{Asm}[Tail condition]\label{asm:utility-tail}
{\rm
The utility function $u$ satisfies $\limsup_{c\to\infty} u'(c+y)/u'(c+z)\le 1$ for all $y,z>0$. 
}
\end{Asm}

\noindent
Finally, we say that the utility function $u(\cdot)$ is of \emph{CRRA type} if there is $\gamma>0$ such that, for all $c>0$,
\be \label{eq:CRRA}
u(c)=\frac{c^{1-\gamma}}{1-\gamma}\quad\text{if}\ \gamma \neq 1\quad \text{and} \quad u(c)=\log(c)\quad \text{if}\ \gamma=1.
\ee
The tail condition in Assumption \ref{asm:utility-tail} is used to establish the existence of a stationary cross-sectional distribution, see in particular Lemma \ref{lem:properties-saving-rules}. A similar condition can be found in Shigeta \cite[Assumption 5]{shigeta-equilibrium}. For example, it is implied by the condition that the absolute risk aversion vanishes, i.e.\ $\lim_{c\to\infty} -u''(c)/u'(c)=0$. In particular, it is true for all CRRA utilities. It does not hold for exponential (CARA) utilities. In fact, 
it is possible to construct models satisfying Assumptions \ref{asm:standing}--\ref{asm:utility-function} with
$r<\rho$, such that, for every initial condition, the optimal wealth process satisfies $a_t^*\to\infty$ almost surely for $t\to\infty$. In particular, there cannot exist a stationary distribution. Since this fact is not of central relevance to this paper, we chose not to include a proof.
We now present the main results of this paper.

\subsection{Results on the household problem}\label{ssec:main-results-hh}
We first analyze the problem of a typical household and then study the aggregate quantities that must satisfy the market-clearing conditions. For a borrowing limit $\lba$, let us set 
$$
\sX:=[\lba,\infty)\times\cZ.
$$
Given values for $(r,w,\tau)$, we define the \emph{value function} of a typical household, for $(a,z)\in\sX$, by
\begin{equation}\label{eq:value-function}
v^*(a,z):=\sup_{(c_t)_{t\geq0}}\E_{a,z}\left[\int_0^\infty \! e^{-\rho t}u(c_t)\,\dt\right]\quad \text{subject to}\quad \d a_t= [ra_t+w(1-\tau)z_t-c_t]\,\d t.
\end{equation}
Here, the notation $\E_{a,z}$ indicates that the wealth-income process is started from an initial value $(a_0,z_0)=(a,z)$, and the maximization is carried out over adapted processes $(c_t)_{t\ge0}$ such that the corresponding wealth process $a_t$ stays almost surely above $\lba$ for all $t\ge0$, see Definition \ref{def:adm-controls} below.
Our first result characterizes the value function as the unique classical solution of a dynamic programming equation with a boundary condition on the derivative at the borrowing limit $\{a=\lba\}$. To state this equation, we define a \emph{Hamiltonian} and the infinitesimal generator $\sL$ of the Markov chain $(z_t)_{t\geq0}$ by 
$$
H(p):=\sup_{c>0}\,(u(c)-cp)\in \R\cup\{\infty\},\qquad (\sL\varphi)(z):=\sum_{y\neq z}\lambda(z,y)\, [\varphi(y)-\varphi(z)],\qquad (p,z)\in\R\times\cZ,
$$
respectively, for any $\varphi:\cZ\mapsto\R$. Under Assumption \ref{asm:utility-function}, the Hamiltonian $H(\cdot)$ is strictly convex and decreasing on $(0,\infty)$. 
Then, the corresponding Hamilton-Jacobi-Bellman (HJB) equation is
\begin{equation}\label{eq:HJB}
\left\{
\begin{aligned}
\rho v(a,z) &= H(v_a(a,z)) + v_a(a,z)[ra+w(1-\tau)z]+\sL v(a,\cdot)(z),\\
v_a(\lba,z) &\geq u'(r\lba+w(1-\tau)z),
\end{aligned}
\right.
\end{equation}
for $(a,z)\in\sX$. We call a continuous function $v:\sX\mapsto\R$ a \emph{classical solution} of \eqref{eq:HJB} if, for every $z\in\cZ$, the function $v(\cdot,z)$ extends to a continuously differentiable function on $(\lba-\varepsilon,\infty)$ for some $\varepsilon>0$ and satisfies \eqref{eq:HJB} pointwise. Recall the set of admissible interest rates $R(w,\tau)$ from \eqref{eq:admissible-interest-rates-general}.

\begin{Thm}\label{thm:household}
Let Assumptions \ref{asm:standing}--\ref{asm:utility-function} be in force. 
\begin{enumerate}[(i)]
\item \label{thm:household-value} (Value function.) The value function $v^*$ defined in \eqref{eq:value-function} is a classical solution to \eqref{eq:HJB} on $\sX$. If $r\in R(w,\tau)$, it is unique in the class of continuous functions that satisfy \eqref{eq:HJB} in the constrained viscosity sense and which are of at most linear growth and bounded from below. If $r=\rho$, the same conclusion holds true when restricting the class of functions to those of sublinear growth.
\item \label{thm:household-control} (Optimal control.) For $r\in R(w,\tau)$, the feedback consumption rule 
$$
c^*(a,z) = (u')^{-1}(v^*_a(a,z)),\qquad (a,z)\in\sX,
$$
is optimal for the problem \eqref{eq:value-function}. It is unique in the class of admissible feedback controls stated in Definition \ref{def:adm-controls}.
\end{enumerate}
\end{Thm}

\noindent The notion of a \emph{constrained viscosity solution}, introduced by Soner \cite{soner1986optimal-I, soner1986optimal-II}, is recalled in Section \ref{ssec:viscosity}. 
We say that a function $f:\sX\mapsto\R$ is of \emph{at most linear growth} and \emph{bounded from below} if, respectively,
\begin{equation}\label{eq:growth-condition}
\underset{a\to\infty}{\limsup}\, \frac{f(a,z)}{a}<\infty \ \  \text{for all}  \ z\in\cZ \qquad \text{and} \qquad \inf \, f>-\infty.
\end{equation}
We say that $f$ is of \emph{sublinear growth} if $\lim_{a\to\infty} f(a,z)/a=0$, for all $z\in\cZ$. 

\subsection{Results on the stationary distribution} \label{ssec:main-results-invariant-distribution}

For $r\in R(w,\tau)$, the optimally controlled state process admits a unique invariant distribution which has compact support on $\sX$. In fact, it is exponentially ergodic as the next result shows.

\begin{Thm}\label{thm:invariant-distribution}
Under Assumptions
\ref{asm:standing}--\ref{asm:utility-tail}, for $r\in R(w,\tau)$, the optimal state process $(a^*_t,z_t)_{t\geq0}$ is an exponentially ergodic Markov process admitting a unique invariant measure $G^*(\d a,\d z)$. $G^*$ is compactly supported.
\end{Thm}

\noindent
To make the dependence on the parameters such as interest rate, wage and tax rate visible, we sometimes write $c^*(a,z)=c^*(a,z;r,w,\tau)$ and $G^*(\da,\dz)=G^*(\da,\dz;r,w,\tau)$ for the optimal feedback consumption and stationary measure, respectively, and omit certain parameters when they are fixed. Next, for given $(w,\tau)\in (0,\infty)\times(-\infty,1)$ and $r\in R(w,\tau)$, let 
$$
A(r,w,\tau) := \int_{\sX} a\, G^*(\da,\dz;r,w,\tau) \quad\text{and}\quad  C(r,w,\tau) := \int_{\sX} c^*(a,z;r,w,\tau)\,G^*(\da,\dz;r,w,\tau)
$$
denote the stationary aggregate asset demand and aggregate consumption of households, respectively.

\begin{Thm}\label{thm:properties-G}
Let Assumptions \ref{asm:standing}--\ref{asm:utility-tail} be in force and fix $(w,\tau)\in(0,\infty)\times(-\infty,1)$.
\begin{enumerate}[(i)]
\item \label{thm:properties-G-continuity} (Continuity.) The mapping $R(w,\tau)\ni r \mapsto G^*(\d a,\d z; r,w,\tau)$ is continuous in the weak topology. In particular, the map $r\mapsto A(r,w,\tau)$ is continuous.  
\item \label{thm:properties-G-hand-to-mouth} (Hand-to-mouth.) For $r\in R(w,\tau)$, set
\be\label{eq:hand-to-mouth}
\Theta(r):= \max_{z\in\cZ}\,  \sum_{y\neq z} \lambda(z,y)\left(\frac{u'(r\lba+w(1-\tau)y)}{u'(r\lba+w(1-\tau)z)} -1\right).
\ee
Then, the stationary measure $G^*=G^*(\d a,\d z;r)$ of the optimal state process is supported at the borrowing constraint $\{a=\lba\}\times\cZ$ if and only if the interest rate $r$ satisfies $\rho-r\geq \Theta(r)$.
\item \label{thm:properties-G-explosion-at-rho} (Explosion at $r=\rho$.) In addition, assume no-borrowing limit. Then, $\lim_{r\uparrow\rho}A(r,w,\tau)=\infty$.
\item \label{thm:properties-G-linearity-of-A-in-tau} (Linearity.) In addition, assume no-borrowing limit and CRRA utility. Then, for $r\in (-\infty,\rho)$,
$$
A(r,w,\tau)=w(1-\tau)A(r,1,0)\quad \text{and}\quad C(r,w,\tau)=w(1-\tau)C(r,1,0).
$$
\item \label{thm:properties-A-increasing-in-r} (Monotonicity.) In addition, assume no-borrowing limit and the CRRA coefficient satisfies $\gamma\le1$. Then,
$$
(-\infty,\rho)\ni r\longmapsto A(r,w,\tau) \quad \text{is an increasing function.}
$$
\end{enumerate}

\end{Thm}

We continue with some remarks on Theorem \ref{thm:properties-G}. 
First, Theorem \ref{thm:properties-G} \eqref{thm:properties-G-hand-to-mouth} provides a sharp condition on the interest rate which characterizes the unique invariant measure being supported at the borrowing constraint. In general, the function $\Theta(r)$ in Theorem \ref{thm:properties-G} \eqref{thm:properties-G-hand-to-mouth} is not necessarily monotone. However, for CRRA utility, it is a direct computation to show that $\Theta$ is increasing, which implies the next corollary.

\begin{Cor}
In addition to the assumptions of Theorem \ref{thm:properties-G} \eqref{thm:properties-G-hand-to-mouth} assume CRRA utility. Let $\underline{r}$ be the unique solution to $\underline{r}=\rho-\Theta(\underline{r})$. Then, $G^*$ is supported at the borrowing constraint if and only if $r\leq \underline{r}$.
\end{Cor}

In a two-state model $\cZ=\{z_\ell,z_h\}$ with a low and a high income state $0<z_\ell<z_h<\infty$ and CRRA utility \eqref{eq:CRRA}, equation \eqref{eq:hand-to-mouth} reduces to the criterion
$$
r\leq \rho - \lambda(z_h,z_\ell) \left(\left[\frac{r\lba+w(1-\tau)z_h}{r\lba+w(1-\tau)z_\ell}\right]^\gamma-1\right) .
$$
We note that the interest rate bound on the right-hand side can have either sign. Moreover, the explicit expression \eqref{eq:hand-to-mouth} shows that under the no-borrowing constraint $\lba=0$, the lower interest bound $\underline{r}$ is independent of the wage $w$ and tax-and-transfer rate $\tau$.

Theorem \ref{thm:properties-G} \eqref{thm:properties-A-increasing-in-r} recovers a result due to Achdou, Han, Lasry, Lions, and Moll \cite[Proposition 5]{achdou2022income}. We present a direct self-contained proof in Section \ref{ssec:A-increasing} using techniques from the theory of viscosity solutions. \cite{achdou2022income} suggest that conditions more general than $\gamma \leq 1$ should be examined in order to establish uniqueness of equilibria. In their setting, the equilibrium condition takes the form 
$$
A(r)=B
$$
where $B\ge0$ is fixed and $A(r)=A(r,w,\tau)$. Thus, uniqueness for arbitrary $B\geq 0$ is equivalent to strict monotonicity of the map $r\mapsto A(r)$. The next proposition shows that, under CRRA utility, $A(r)$ need not be increasing when $\gamma>1$. In this sense, the condition $\gamma\leq 1$ is sharp.

\begin{Thm}[Failure of monotonicity for $\gamma>1$]\label{thm:A-can-decrease-gamma-greater-one}
Assume the no-borrowing constraint and CRRA utility with coefficient $\gamma>1$. Then, there exists a two-state income process satisfying Assumption \ref{asm:standing} and $r_1<r_2<\rho$ such that $A(r_1,1,0)>A(r_2,1,0)$. If $\gamma>(1+\sqrt5)/2$, the two interest rates may be chosen strictly positive.
\end{Thm}

The proof of Theorem \ref{thm:A-can-decrease-gamma-greater-one} is constructive and given in Subsection \ref{ssec:A-fail-increasing}. It constructs a ``rare-disaster'' economy where individuals' income takes on a low ``disaster'' state while spending most time in the high income state.

\subsection{Results on stationary equilibria}\label{ssec:main-results-equilibria}

For fixed $\rho$ and $\alpha$, both Huggett and Aiyagari equilibria can be parameterized by the two parameters $(r,\tau)$. Indeed, in the Huggett model, this will become clear from the fact that, for fixed $(r,\tau)$, there is a unique optimal feedback control $c^*$ and a unique invariant measure $G^*$, which uniquely determines the asset demand $A^*=B^*$ in equilibrium. In the Aiyagari model, the same logic applies upon noting that there is a one-to-one correspondence between the interest rate $r^*$, capital $K^*$ and the wage $w^*$ in equilibrium, see Section \ref{ssec:proof-surplus} for more details. In light of these observations, we ask the following two questions:
\begin{enumerate}[(i)]
\item \label{question-given-r} Given an interest rate $r$, how many stationary equilibria with $r^*=r$ exist?
\item \label{question-given-tau} Given a primary surplus $\tau$, how many stationary equilibria with $\tau^*=\tau$ exist?
\end{enumerate}
Using a version of \emph{Walras' law}, see Section \ref{sec:general-equilibrium}, it is sufficient to ensure that \emph{either} the asset or goods market clears.  We summarize our main existence and uniqueness results for both the Huggett and Aiyagari models. First, for a fixed interest rate $r<\rho$, we prove the following theorem.

\begin{Prop}[Answer to \eqref{question-given-r}] \label{prop:fixed-r}
Under Assumption \ref{asm:standing}, the no-borrowing limit, and CRRA utility with $\gamma\leq1$, the following hold:

\begin{enumerate}[(i)]
\item \label{prop:fixed-r-huggett} For $r<\rho$, there exists a unique Huggett equilibrium with equilibrium interest rate $r$.
\item \label{prop:fixed-r-aiyagari} For $r\in(-\delta,\rho)$, there exists at most one Aiyagari equilibrium with $r=r^*$. 
An equilibrium exists whenever $r>\underline{r}$ and the capital elasticity of output $\alpha$ is small or the depreciation rate $\delta$ sufficiently large.
\end{enumerate}
\end{Prop}

We now state our results for a fixed primary surplus $\tau<1$. Recall the definition of the lower interest rate bound $\underline{r}$ from Theorem \ref{thm:properties-G} \eqref{thm:properties-G-hand-to-mouth}. The next theorem recovers the results of \cite{kaplan2023price} for the Huggett model.

\begin{Thm}[Answer to \eqref{question-given-tau} for surpluses]
\label{thm:fixed-tau-surplus}
Under Assumption \ref{asm:standing}, the no-borrowing limit, and CRRA utility with $\gamma\leq1$, the following hold:
\begin{enumerate}[(i)]
\item \label{thm:fixed-tau-surplus-positive} If $\tau\in(0,1)$, then there exists a unique stationary Huggett and Aiyagari equilibrium with $\tau^*=\tau$. The equilibrium interest rate satisfies $r^*\in(0,\rho)$. 
\item \label{thm:fixed-tau-surplus-zero} If $\tau=0$, then there exists a continuum of stationary Huggett equilibria with $\tau=\tau^*$ whose interest rates $r^*$ satisfy $r^*\in(-\infty,\underline{r}]\cup\{0\}$. There is a unique zero-debt Aiyagari equilibrium (i.e., $B^*=0$). An Aiyagari equilibrium with $r^*=0$ exists if and only if $A(0,1,0)\geq \alpha/((1-\alpha)\delta)$. If the inequality is strict, this gives a second equilibrium with $B^*>0$; if equality holds, it coincides with the zero-debt equilibrium.
\end{enumerate}
\end{Thm}

\begin{Thm}[Answer to \eqref{question-given-tau} for deficits] 
\label{thm:fixed-tau-deficit}
Let Assumptions \ref{asm:standing}, the no-borrowing limit, and CRRA utility hold.

\begin{enumerate}[(i)]
\item \label{thm:fixed-tau-deficit-small} Assume $\gamma\le1$ and $\tau<0$. A necessary condition for the existence of an equilibrium is that $\underline{r}<0$, in which case at least two Huggett equilibria exist for $\tau<0$ close to zero. In the case of Aiyagari equilibria, the same holds if additionally $\alpha$ is close to zero.

\item \label{thm:fixed-tau-deficit-multiple-equilibria}
Assume $d\ge2$ is an even integer. For any $\gamma>0$, there exists an income process with $d$ states satisfying Assumption \ref{asm:standing} and a primary deficit $\tau<0$ such that there are at least $d$ many stationary Huggett equilibria.

\item \label{thm:fixed-tau-deficit-large} Assume $\gamma\le1$. For sufficiently large deficits $\tau<0$, no equilibrium exists.
\end{enumerate}
\end{Thm}

We note that for primary deficits $\tau<0$, any equilibrium interest rate $r^*$ is strictly negative, $r^*<0$. 

We continue to illustrate Theorem \ref{thm:fixed-tau-surplus} \eqref{thm:fixed-tau-surplus-positive} and Theorem \ref{thm:fixed-tau-deficit} \eqref{thm:fixed-tau-deficit-small} and \eqref{thm:fixed-tau-deficit-large} using numerical computations. In the Huggett model, $w=1$ and we set $A(r,\tau) = A(r,1,\tau)$.
Note that asset market clearing together with the government's budget constraint imply that in any Huggett equilibrium $\Xi^*=(\tau^*, B^*, r^*, c^*, G^*)$ with $\tau^*=\tau$, we must have $r^*A(r^*,\tau)=\tau$. If $r^*$ is non-zero, we get the following condition for the interest rate:
$$
A(r^*,\tau)=\frac{\tau}{r^*}.
$$
We plot both sides of this equation in Figure \ref{fig:huggett-unique-eq} for a positive value of $\tau>0$, illustrating  Theorem \ref{thm:fixed-tau-surplus} \eqref{thm:fixed-tau-surplus-positive}. Note that there exists a unique interest rate that equates asset demand and supply.

\begin{figure}[h]
\centering
\includegraphics[width=0.5\linewidth]{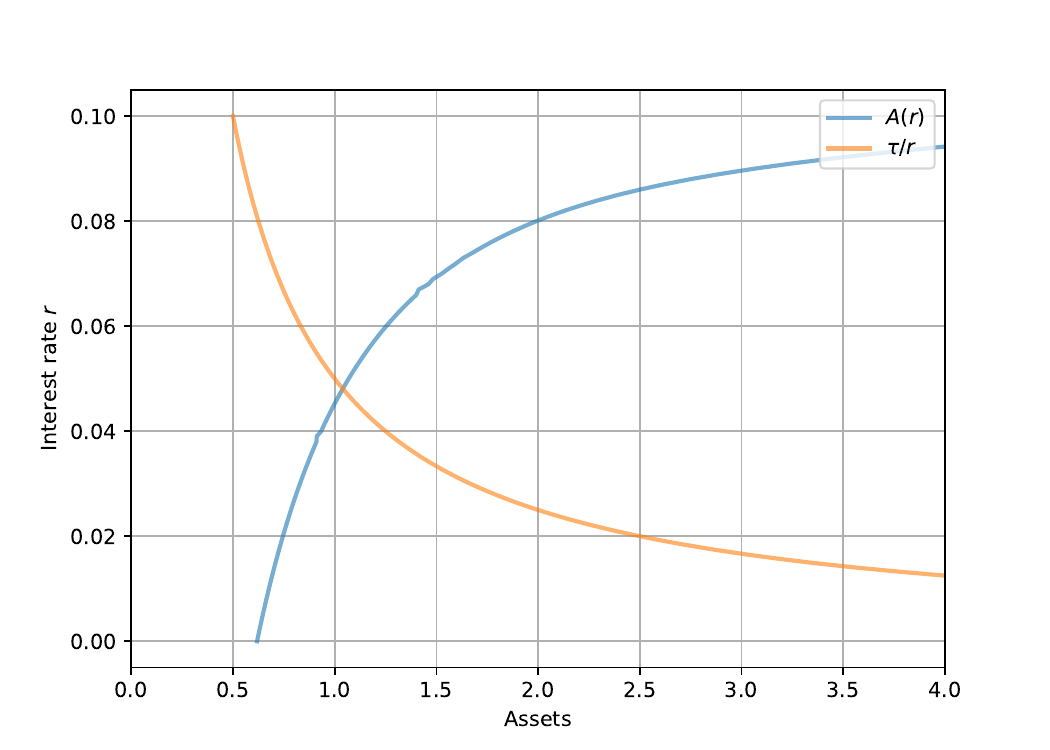}
\caption{{\bf (Huggett)} Plot of $r\mapsto A(r,\tau)$ and $r\mapsto \tau/r$ for $\tau\in(0,1)$.}
\label{fig:huggett-unique-eq}
\end{figure}

In Figure \ref{fig:huggett-negative-tau}, we illustrate Theorem \ref{thm:fixed-tau-deficit} \eqref{thm:fixed-tau-deficit-small} and \eqref{thm:fixed-tau-deficit-large}. Depending on the value of $\tau<0$, there might exist no or multiple equilibria.

\begin{figure}[H]
\centering
\begin{subfigure}[t]{0.495\textwidth}
\centering
\includegraphics[width=\linewidth]{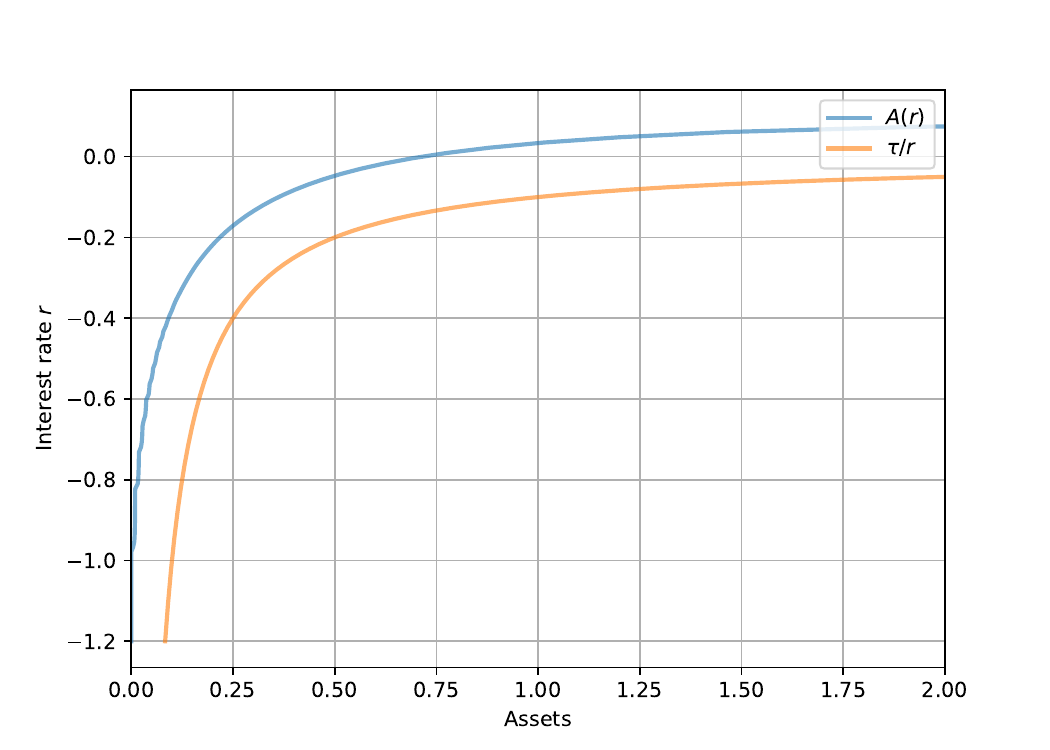}
\caption{No equilibria for sufficiently large $|\tau|$}
\label{fig:huggett-no-eqNLL}
\end{subfigure}
\begin{subfigure}[t]{0.495\textwidth}
\centering
\includegraphics[width=\linewidth]{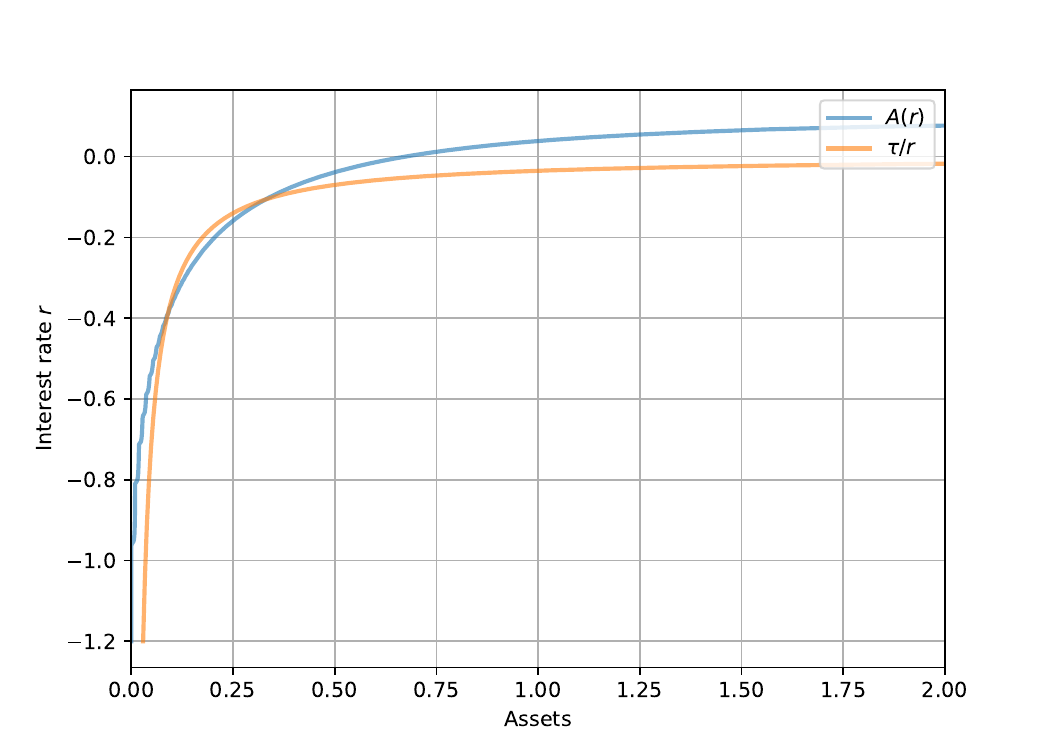}
\caption{Two equilibria for sufficiently small $|\tau|$}
\label{fig:huggett-multiple-eq}
\end{subfigure}
\caption{{\bf (Huggett)} Plots of $r\mapsto A(r,\tau)$ and $r\mapsto \tau/r$ for two values of $\tau<0$.}
\label{fig:huggett-negative-tau}
\end{figure}

We continue to illustrate these results numerically. In Section \ref{ssec:proof-surplus}, we see that Aiyagari equilibria $\Xi^*$ are characterized by values of $r^*$ that satisfy $A(r^*,1,\tau)=S(r^*)$ where 
$$
S(r):=\frac{\alpha}{1-\alpha} \,  \frac{1}{r+\delta} + \frac{\tau}{r},\qquad r\in(-\delta,\infty)\setminus\{0\}.
$$

Figure \ref{fig:aiyagari-negative-tau} illustrates Theorem \ref{thm:fixed-tau-deficit} \eqref{thm:fixed-tau-deficit-small} and \eqref{thm:fixed-tau-deficit-large} for the Aiyagari model: For sufficiently small $\alpha$ and $\tau<0$ there exist two equilibria in which $r^*<0$ and when $\alpha$ is sufficiently close to $1$, there are no equilibria. Note that intersection points that correspond to positive interest rates do not satisfy the government's budget constraint and the requirement $B^*\ge0$.

\begin{figure}[h]
\centering
\begin{subfigure}[t]{0.495\textwidth}
\centering
\includegraphics[width=\linewidth]{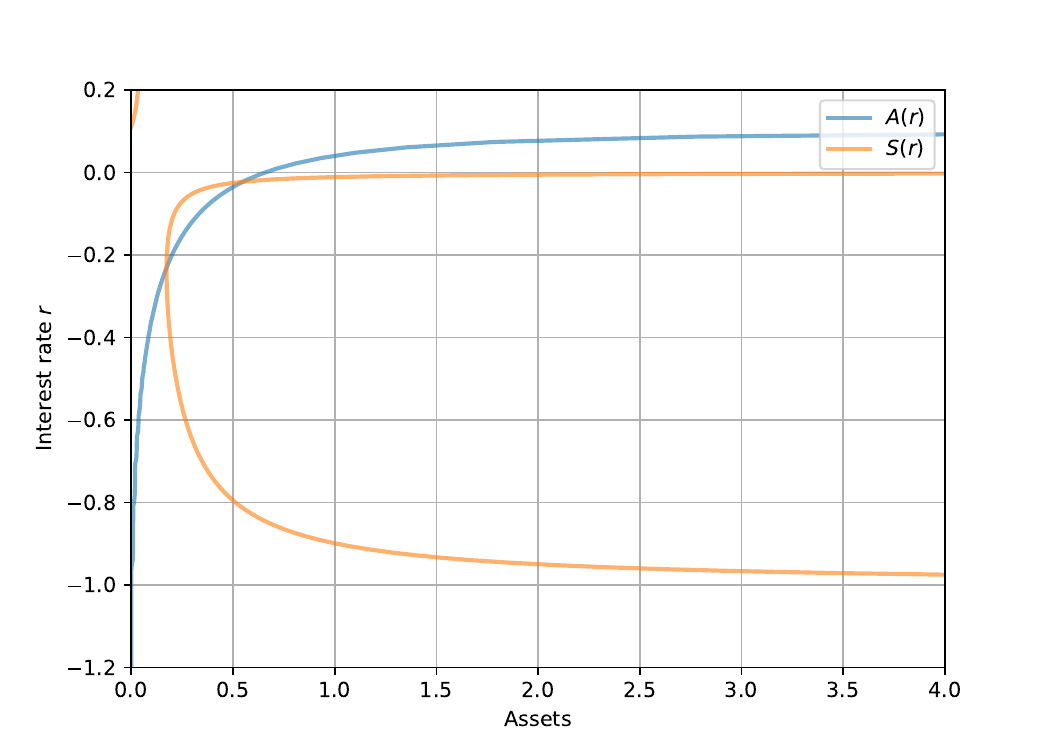}
\caption{Multiple equilibria with negative interest rate}
\label{fig:aiyagari-very-negative-tau}
\end{subfigure}
\begin{subfigure}[t]{0.495\textwidth}
\centering
\includegraphics[width=\linewidth]{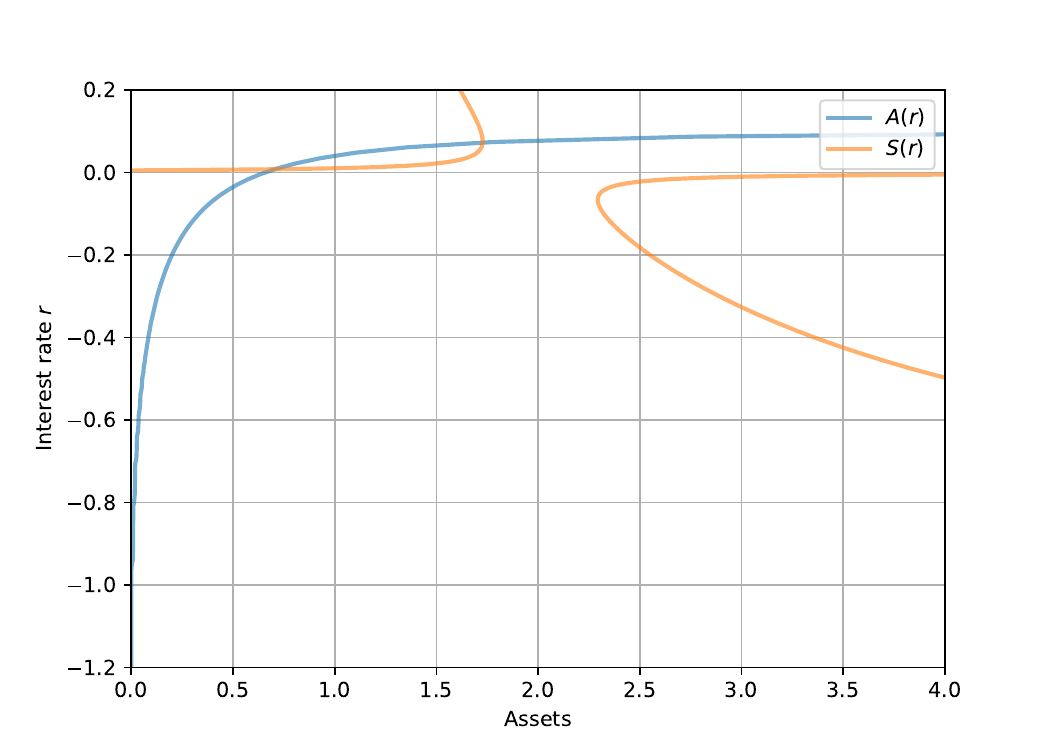}
\caption{No equilibria for a high $\alpha$}
\label{fig:aiyagari-neg-tau-high-alpha}
\end{subfigure}
\caption{{\bf (Aiyagari)} Plots of $r\mapsto A(r,1,\tau)$ and $r\mapsto S(r)$ for $\tau<0$ and  two values of $\alpha\in(0,1)$}
\label{fig:aiyagari-negative-tau}
\end{figure}

We finally discuss the relationship between stationary Huggett and Aiyagari equilibria.

\begin{Thm}\label{thm:huggett-to-aiyagari}
Let Assumption \ref{asm:standing}, the no-borrowing limit and CRRA utility with $\gamma\le1$ hold. For given $\tau\in(0,1)$, Aiyagari equilibria converge to the (unique) Huggett equilibrium as $\alpha\downarrow0$. Next, let $\tau<0$, $\alpha_n\downarrow0$, and let $r_n$ be any Aiyagari equilibrium interest rate corresponding to $\alpha_n$. If $\alpha_n/(r_n+\delta)\to0$, then any limit point of Aiyagari equilibria as $n\to\infty$ is a Huggett equilibrium. In this equilibrium, the interest rate $r^*$ satisfies $r^*\geq-\delta$.
\end{Thm}

\subsection{Discussion and related literature} \label{ssec:literature}

This work sits at the intersection of three streams of literature: (i) the Fiscal Theory of the Price Level (FTPL), (ii) the continuous-time analysis of heterogeneous-agent models and (iii) the mathematical theory of mean-field games.

We contribute to the question of price-level determinacy within FTPL models. Building on Brunnermeier, Merkel, and Sannikov \cite{brunnermeier2020fiscal}, who study uniqueness with idiosyncratic return risk in low-interest-rate environments, we adapt the setting to a framework where incomplete markets arise from idiosyncratic, uninsurable income shocks in the tradition of Bewley-Huggett-Aiyagari. In the version of our model without capital, we follow the setup of Kaplan, Nikolakoudis, and Violante \cite{kaplan2023price}. We complement their treatment with a mathematical analysis of an Aiyagari model with capital. For further references in the macroeconomic literature, see \cite{kaplan2023price}, and for additional background on the Fiscal Theory of the Price Level, see Cochrane \cite{cochrane2023fiscal}. We also mention the related work of Farmer \cite{farmer2024money}, who introduces fiat money into a discrete-time heterogeneous-agent economy with uninsurable income risk and shows that sufficiently high income uncertainty can generate both monetary and non-monetary steady states. 

Our analysis of the continuous-time economy builds on Shigeta’s recent preprints \cite{shigeta-control-problem,shigeta-equilibrium} which establish existence and uniqueness of an invariant distribution in a closely related framework. In contrast to the bounded-utility assumption of Shigeta, our analysis accommodates more general---possibly unbounded---utility functions and more general borrowing constraints. The explicit characterization of a hand-to-mouth economy in Theorem \ref{thm:properties-G} \eqref{thm:properties-G-hand-to-mouth} is novel to our knowledge. The other statements in Theorem \ref{thm:properties-G} appear in related contexts, and we provide rigorous proofs in the present continuous-time context. We draw extensively on the theory of viscosity solutions and stochastic optimal control, relying in particular on the notion of constrained viscosity solutions introduced by Soner \cite{soner1986optimal-I,soner1986optimal-II}. The failure of monotonicity of aggregate savings as a function of the interest rate, see Theorem \ref{thm:A-can-decrease-gamma-greater-one}, is new to the best of our knowledge and addresses an open question by Achdou, Han, Lasry, Lions, and Moll \cite{achdou2022income}.
The equilibrium analysis for the Huggett model for primary surpluses recovers results by Kaplan, Nikolakoudis, and Violante \cite{kaplan2023price}. In the case of primary deficits, they note that the existence of more than two equilibria cannot be ruled out. We provide a novel and explicit construction of a Huggett economy that confirms this. Such an economy can admit at least as many deficit equilibria as income states, see Theorem \ref{thm:fixed-tau-deficit} \eqref{thm:fixed-tau-deficit-multiple-equilibria}. 

Related work on the existence and uniqueness of stationary equilibria in heterogeneous-agent models includes the discrete-time analysis of \cite{acikgoz2018existence} and the continuous-time treatment of \cite{achdou2022income}. In the two-state income case $|\cZ|=2$, Bayer, Rendall, and W\"{a}lde \cite{bayer2019invariant} also proved the existence and uniqueness of an invariant measure.  Finally, we refer to Ambrose \cite{ambrose2021existence} and Chen et al.\ \cite{chen2025stochastic} for related time-dependent analyses.

\subsection{Connection to mean-field games} \label{ssec:MFG-connection}

This section formulates the problem of finding stationary Huggett and Aiyagari equilibria  as a mean-field game. We begin by recalling a generic mean-field game framework. Let $U$ be a control set, let $W$ be a Brownian motion and $\cN$ a stationary Poisson random measure. Fix an initial distribution $\mu_0\in\sP(\R^n)$, where $n\geq1$ denotes the dimension of the state process. Given a flow of probability distributions $\boldsymbol{\nu}=(\nu_t)_{t\geq0}$ on $\R^n\times U$, consider the stochastic optimal control problem
\begin{equation}\label{eq:MFG-control-problem}
\sup_{\alpha}\ \E\left[\int_0^\infty f(t,X^\alpha_t,\nu_t,\alpha_t)\,\d t\right],
\end{equation}
subject to some controlled dynamics 
$$
\d X^\alpha_t = b(t,X^\alpha_t,\nu_t,\alpha_t)\,\d t + \sigma(t,X^\alpha_t,\nu_t,\alpha_t)\,\d W_t+\int \xi(t-,X^\alpha_{t-},\nu_{t-},\alpha_t,\zeta)\,\cN(\d t,\d \zeta),\quad X_0\sim\mu_0.
$$ 
Here, the control $\alpha$ is a random process taking values in $U$. Apart from measurability and integrability requirements, one might further restrict the set of admissible controls by imposing state constraints.
Furthermore, $(b,\sigma,\xi)$ are coefficients of appropriate dimensions. We call the flow $\boldsymbol{\nu}$ a \emph{mean-field game Nash equilibrium} or \emph{solution to the mean-field game} starting from $\mu_0$ if 
$$
\nu_t=\mathrm{Law}(X^{\alpha^*}_t,\alpha^*_t),\qquad t\geq0,
$$
where $\alpha^*$ is an optimal admissible control given $\boldsymbol{\nu}$. A \emph{stationary mean-field game problem} looks for solutions to the mean-field game that are constant in time, in which case the initial condition $\mu_0$ is not given, but part of the solution.

We now showcase the connection to mean-field games in the Huggett model.  Fix a function $\tau(\cdot)$. Then, we can view stationary monetary Huggett equilibria, i.e.\ equilibria in which real government debt is strictly positive, as stationary solutions of the following mean-field game. Recall that $d$ denotes the number of possible income states. Let $\cN$ be a stationary Poisson random measure on $[0,\infty)\times\R^d$ with intensity measure 
$$
\nu(B):=\sum_{j=1}^d \, \mathrm{Leb}(B\cap S_j),\quad B\in\sB(\R^d),
$$ 
where $ \mathrm{Leb}$ is the one-dimensional Lebesgue measure and $S_j:=\{x\in\R^d\,:\,x_k=0\ \forall k\neq j\}$ is viewed as a subset of the real line. We then define the state $X=(a,z)$ along with
$$
\mathfrak{r}(\mu) := \frac{\int_\sX\tau(z)\,\mu(\d a,\d z)}{\int_\sX a\,\mu(\d a,\d z)},\qquad f(t,c):=e^{-\rho t}u(c),
$$
and coefficients
$$
b(X,\mu,c):= 
\begin{pmatrix}
\mathfrak{r}(\mu)a+z-\tau(z)-c\\
0
\end{pmatrix},
\qquad \sigma\equiv0,\qquad
\xi(X,\zeta) :=
\begin{pmatrix}
0\\
\sum_{j=1}^d(z_j-z)\,\chi_{(0,\lambda(z,z_j))}(\zeta_j)
\end{pmatrix},
$$
for $(t,X,\mu,c,\zeta)\in\R_+\times\sX\times\sP(\sX)\times(0,\infty)\times\R^d$ satisfying $0<\int a\mu(\d a,\d z)<\infty$. Note that, indeed, with this construction the second component of the state process $X$ follows a stationary Markov chain with rates given by $\lambda(\cdot,\cdot)$.  

To make the connection to mean-field games, we relax the equilibrium definitions in Subsection \ref{ssec:models} and only require the existence of an optimal control starting from a fixed initial condition $(a_0,z_0)\sim G^*$. More precisely, instead of a feedback function $c^*(a,z)$ we consider an open-loop control $(c^*_t)_{t\geq0}$ that is admissible starting from an initial condition $(a_0,z_0)\sim G^*$. A priori, this makes the notion of equilibria dependent on the underlying probabilistic structure. Usually, one appeals to results on the \emph{law-invariance} of the typical household's problem, see, for example, \cite{djete2022mckean, hofer2024optimal}, to obtain a definition that is independent of the underlying probabilistic structure. In the cases we study, optimal controls are naturally of feedback form, and hence a law-invariance principle automatically holds. We have the following result, which can be deduced along the lines of Lemma \ref{lem:Walras-Huggett}.

\begin{Prop}
Given a tax-and-transfer function $\tau(\cdot)$, let $\mu^*\in\sP(\sX)$ satisfy $0<\int a \mu^*(\d a,\d z)<\infty$. Then, $\mu^*$ is a stationary equilibrium of the above mean-field game if and only if there exists a stationary monetary Huggett equilibrium $\Xi=(\tau(\cdot),B,r,(c_t)_{t\geq0},G)$ with $G=\mu^*$.  
\end{Prop}

We emphasize that $\mu^*$ can be seen as a fixed point of the following mapping
\begin{align*}
\mu \longmapsto \mathfrak{r}(\mu) \mapsto \text{optimal state}\ X^*\ \text{given}\ r=\mathfrak{r}(\mu)\longmapsto \text{stationary measure}\ \mu^* \ \text{of}\ X^*. 
\end{align*}
In Subsection \ref{ssec:main-results-invariant-distribution} we will see that this indeed is a well-defined map. We conclude with a remark on common terminology.

\begin{Rmk}[Typical vs representative]
{\rm
In the mean-field literature, it is common to call a household that faces the optimal control problem \eqref{eq:MFG-control-problem} a \emph{representative} household. This is motivated by the fact that one imagines a mean-field of households, each solving the same optimal control problem but with independent idiosyncratic noise processes. In macroeconomic theory, however, a \emph{representative household} refers to a model without idiosyncratic shocks so that all households are identical, see \cite{kaplan2023price}, and we follow Acemoglu \cite[Chapter 17.4]{acemoglu2008introduction} in replacing the phrase \emph{representative household} by \emph{typical household}.
}
\end{Rmk}

\subsection{Structure}

In Section \ref{sec:household-problem}, we analyze the stochastic optimal control problem of a typical household. In particular, we prove Theorem \ref{thm:household}. In Section \ref{sec:invariant-distribution}, we prove the existence and uniqueness of an invariant distribution of the optimal state process, see Theorem \ref{thm:invariant-distribution}. We furthermore establish various properties of this cross-sectional distribution recorded in Theorem \ref{thm:properties-G} and \ref{thm:A-can-decrease-gamma-greater-one}. Finally, Section \ref{sec:general-equilibrium} uses the results established in the previous sections to prove the statements about stationary Huggett and Aiyagari equilibria recorded in Propositions and Theorems \ref{prop:fixed-r}--\ref{thm:huggett-to-aiyagari}.

\section{Household problem}\label{sec:household-problem}

This section analyzes the optimal control problem of a typical household for given values of $(r,w,\tau)$. Throughout this section we work under Assumptions \ref{asm:standing}, \ref{asm:borrowing-constraint} and \ref{asm:utility-function}. Furthermore, we assume that the interest rate $r$ satisfies $r\leq \rho$. Without loss of generality we set $w=1$ and $\tau=0$ in this section since they only rescale the income space $\cZ$. For a generic borrowing limit $\lba$, the state space is defined as
$$
\sX=[\lba,\infty)\times\cZ.
$$

Let $(\Omega,\sF,\F,\P)$ be a complete filtered probability space supporting the stationary Markov chain $(z_t)_{t\geq0}$. We assume that $\F=(\sF_t)_{t\ge0}$ satisfies the usual conditions with respect to $(\sF,\P)$. One may for example complete the right-continuous natural filtration $\sF^\circ_t$ of $(z_t)_{t\ge0}$ with $(\sF,\P)$-null sets $\cN$ and let $\sF_t := \sigma(\sF^\circ_t\cup\cN)$. 

For an initial condition $(a,z)\in\sX$ and a positive process $(c_t)_{t\ge0}$, define the controlled wealth process $\ba = \ba^{a,z,\bc} = (a^{a,z,\bc}_t)_{t\ge0}$ by
\begin{equation}\label{eq:hh-problem-state-equation}
a_t = a_0+\int_0^t(r a_s + z_s- c_s)\,\d s,\qquad z_0=z.
\end{equation}
We define the household's lifetime utility by 
\be\label{eq:hh-problem-lifetime-utility}
J(a,z,\bc) := \E_{a,z}\left[\int_0^\infty \! e^{-\rho t}u(c_t)\,\d t\right],\qquad (a,z)\in\sX.
\ee

\begin{Def}[Admissibility and optimality of controls]\label{def:adm-controls}
{\rm
\begin{enumerate}[(i)]
\item A progressively measurable stochastic process $\bc=(c_t)_{t\geq0}:\Omega\times[0,\infty)\mapsto(0,\infty)$ is an \emph{admissible control} or \emph{admissible consumption process} for the initial condition $(a,z)\in\sX$ if it is locally integrable in time, almost surely, and the corresponding controlled wealth process $ (a^{a,z,\bc}_t)_{t\ge0}$ satisfies $a^{a,z,\bc}_t\geq \underline{a}$ for all $t\geq0$ almost surely. The set of admissible open-loop controls for the initial condition $(a,z)$ is denoted by $\sC(a,z)$. A control $\bc^*\in\sC(a,z)$ is an \emph{optimal control starting from $(a,z)$} if it maximizes \eqref{eq:hh-problem-lifetime-utility} over $\sC(a,z)$.
\item  A continuous function $c:\sX\mapsto (0,\infty)$ is an \emph{admissible consumption policy} if the random ordinary differential equation
$$
a^{a,z,c}_t = a_0+\int_0^t(r a^{a,z,c}_s + z_s- c(a^{a,z,c}_s,z_s))\,\d s,\quad z_0=z.
$$
has an almost surely unique solution in the sense of Carath\'eodory for any initial condition $(a,z)\in\sX$ that stays above $\underline{a}$ for all $t\geq0$ a.s. The set of admissible closed-loop controls is denoted by $\sC_{\mathrm{cl}}$. A function $c^*\in\sC_{\mathrm{cl}}$ is an \emph{optimal feedback control} if, for any initial condition $(a,z)\in\sX$, the open-loop control $c^*_t:=c^*(a^{a,z,c^*}_t\!\!,z_t)$, $t\ge0$, maximizes \eqref{eq:hh-problem-lifetime-utility} .
\end{enumerate}
}
\end{Def}

We define the \emph{value function} by
\begin{equation}\label{eq:hh-problem-value-function}
v^*(a,z) := \sup\left\{J(a,z,(c_t)) \, :\, (c_t)\in \sC(a,z)\right\}, \qquad (a,z)\in\sX.
\end{equation}

\begin{Rmk}
{\rm
\begin{enumerate}
\item The set of admissible controls depends on the parameters of the model, such as the interest rate. To emphasize this dependence, we may write $\sC(a,z;r)$ instead of $\sC(a,z)$, and similarly for other parameters.
\item If the utility function is defined on $\R_+$, then we may allow consumption processes and consumption policies to take the value $0$. 
\end{enumerate}
}
\end{Rmk}

The borrowing limit imposes an upper bound on admissible consumption rules as demonstrated by the next lemma.

\begin{Lem}[Household's budget constraint]
Let $(a,z)\in\sX$. Then any $(c_t)\in\sC(a,z)$ satisfies, almost surely,
\be\label{hh-budget}
\int_0^\infty \! e^{-\rho t} c_t\,\dt \leq a_0-\frac{\rho-r}{\rho}\lba+\int_0^\infty \! e^{-\rho t}z_t\,\d t.
\ee
\end{Lem}

\begin{proof}
Solving the state equation and using admissibility of $(c_t)\in\sC(a,z)$, 
$$
e^{-rt}\lba\leq e^{-rt}a_t = a_0 + \int_0^t e^{-rs}[z_s-c_s]\,\ds,\quad\Rightarrow \quad \int_0^t e^{-rs}c_s\,\ds \leq a_0-e^{-rt}\lba+ \int_0^t e^{-rs} z_s\,\ds.
$$
If $r=\rho$, taking $t\to\infty$ yields the claim. If $r<\rho$, multiply both sides by $(\rho-r)\exp(-(\rho-r)t)$ and integrate over $[0,\infty)$ using Fubini-Tonelli's theorem to obtain
\begin{equation*}
\int_0^\infty \! e^{-\rho t} c_t\,\d t = (\rho-r)\int_0^\infty  e^{-(\rho-r)t} \int_0^t e^{-rs}c_s\,\ds \,\dt  \leq a_0-\frac{\rho-r}{\rho}\lba+\int_0^\infty \! e^{-\rho t}z_t\,\d t.  \qedhere
\end{equation*}
\end{proof}

\begin{Prop} \label{prop:basic-prop-value}
For every $z\in\cZ$, the value function $v^*(\cdot,z)$ of \eqref{eq:hh-problem-value-function} is finite, increasing and concave on $[\lba,\infty)$. Furthermore, we have the following upper bound
\be\label{eq:upper-bound-v}
v^*(a,z) \leq \frac{1}{\rho} u(\rho a-(\rho-r)\lba+\bar z),\qquad (a,z)\in\sX.
\ee
In particular, the value function is of sublinear growth. Here, $\bar z:= \max \cZ$.
\end{Prop}
\begin{proof}
\emph{Finiteness.} Define $\mathsf{c}:= (r\lba+\lbz)/2$.
By Assumption \ref{asm:borrowing-constraint}, $\mathsf{c}>0$. Set $c_t:=\mathsf{c}$ for all $t\geq0$. Savings at the borrowing constraint satisfy
$$
r\underline{a}+z_t-\mathsf{c}  \geq \frac{r\lba+\lbz}{2}>0.
$$
Hence the corresponding wealth process $a_t$ stays above $\underline{a}$ and, consequently, $v^*(a,z)>-\infty$ for any initial condition $(a,z)\in\sX$. To prove $v^*(a,z)<+\infty$, let $(c_t)\in\sC(a,z)$ be any admissible consumption process.  By \eqref{hh-budget},
$$
\E\left[\int_0^\infty \! e^{-\rho t} c_t\,\d t\right] \leq a - \frac{\rho-r}{\rho} \, \lba + \frac{\bar z}{\rho}.
$$
Then, Jensen's inequality, concavity of $u(\cdot)$ and positivity of $c_s$ imply
\begin{align*}
\E\left[\int_0^{\infty} e^{-\rho t} u(c_t)\,\d t\right] &\leq \frac{1}{\rho} u\left(\E\left[\int_0^{\infty} \rho e^{-\rho t} c_t\,\d t\right]\right)\leq   \frac{1}{\rho} u\left(\rho a-(\rho-r)\lba + \bar z \right) <\infty.
\end{align*}
Taking the supremum over all admissible controls implies \eqref{eq:upper-bound-v}. Since $u'(c)\to0$ as $c\to\infty$, the value function is of sublinear growth.

\emph{Monotonicity.} This is immediate since $\sC(a_1,z)\subset \sC(a_2,z)$ whenever $\lba \leq a_1\leq a_2<\infty$.

\emph{Concavity.} This is a direct consequence of the linear dynamics and concavity of the utility function.
\end{proof}

We now state the dynamic programming principle (DPP) for our problem. Using the concavity of the value function, $v^*(\cdot,z)$ is continuous on the interior $(\lba,\infty)$. Continuity at the borrowing constraint $a=\lba$, however, is less obvious, and basic versions of the (DPP) do not apply. A proof of continuity is given below in Proposition \ref{prop:diff-boundary} and relies on the dynamic programming principle. In our setting, we still have the (DPP), and we refer to \cite[Appendix A]{gassiat2014investment} for a proof in a more general setting.

\begin{Thm}[Dynamic programming principle] \label{thm:DPP}
For any finite $\F$-stopping time $\tau$ and any $(a,z)\in\sX$,
\begin{equation*}
v^*(a,z) = \sup_{(c_t)\in\sC(a,z)} \E_{a,z}\left[ \int_0^\tau e^{-\rho t} u(c_t) \,\dt + e^{-\rho \tau} v^*(a_\tau,z_\tau)\right].
\end{equation*}
\end{Thm}

Equipped with Theorem \ref{thm:DPP}, we are able to prove continuity and even differentiability at the borrowing constraint. Recall the notation $\lbz$ for the lowest income level.

\begin{Prop}\label{prop:diff-boundary}
For every $z\in\cZ$, the limit
\begin{equation}\label{eq:diff-bdry}
v^*_a(\underline{a},z):=\lim_{h\downarrow0}\, \frac{1}{h}(v^*(\lba+h,z)-v^*(\lba,z))\ \ \text{exists in} \ \ [0,\infty).
\end{equation}

\end{Prop}

\begin{proof}
By Assumption \ref{asm:borrowing-constraint} we may choose  $\delta_1,\delta_2>0$ with $0<\delta_1<\delta_2<r\lba+\lbz$ and define the feedback consumption rule
$$
\hat c(a,z) := \delta_1 + (ra+z-\delta_2)^+,\qquad (a,z)\in\sX.
$$
Set $\varepsilon:=\delta_2-\delta_1>0$. Starting from $(a_0,z_0)=(\underline{a},z)$, $z\in\cZ$, the corresponding wealth process $(\hat a_t )_{t\ge0}$ satisfies
$$
\d \hat a_t = [r \hat a_t + z_t - \hat c(\hat a_t,z_t)]\,\dt = 
\begin{cases}
\varepsilon\,\dt, & r \hat a_t+z_t\geq\delta_2,\\
[r \hat a_t +z_t -\delta_1]\,\dt, & r \hat a_t+z_t<\delta_2.
\end{cases}
$$
At $t=0$, $\d \hat a_t = \varepsilon\,\d t>0$. If $r\geq0$, then $\hat a_t\geq \lba +\varepsilon t$ for all $t\geq0$. If $r<0$, it is clear that there exists $t_*>0$ such that $\hat a_t\geq\lba+\varepsilon t$ on $t\in[0,t_*]$. We let $\hat c_t:=\hat c(\hat a_t,z_t)$ be the corresponding open-loop control. By the dynamic programming principle, for any $t\in[0,t_*]$,
\begin{align*}
v^*(\lba,z) &\geq \E_{\lba,z}\left[\int_0^t e^{-\rho s} u(\hat c_s)\,\d s + e^{-\rho t} v^*(\hat a_t,z_t)\right] \\
&\geq \E_{\lba,z}\left[\int_0^t e^{-\rho s} u(\delta_1)\,\d s + e^{-\rho t} v^*(\lba+\varepsilon t,z_t)\right]= \frac{1}{\rho}(1-e^{-\rho t}) u(\delta_1) + e^{-\rho t} \E_{\lba,z}[v^*(\lba+\varepsilon t,z_t)],
\end{align*}
where we have used monotonicity of $u$ and $v^*$ in the second line. Using the relation
$$
\E_{\lba,z}[v^*(\lba+\varepsilon t,z_t)] = \sum_{y\neq z}\P_{z}(z_t=y)
[v^*(\lba+\varepsilon t,y)-v^*(\lba+\varepsilon t,z)] + v^*(\lba+\varepsilon t,z)
$$
we obtain
\begin{equation*}
e^{-\rho t} v^*(\lba+\varepsilon t,z ) - v^*(\lba,z) \leq - \frac{1}{\rho}(1-e^{-\rho t}) u(\delta_1) - e^{-\rho t}\sum_{y\neq z}\P_{z}(z_t=y)[v^*(\lba+\varepsilon t, y)-v^*(\lba+\varepsilon t,z)] .
\end{equation*}
We first establish continuity at the borrowing constraint. To this end, note that
\begin{equation}\label{eq:proof-diff-bdry-I}
\P_{z}(z_t=y)  = \lambda(z,y)\,t + o(t),\quad \text{as} \ t\downarrow0,\qquad y\neq z,
\end{equation}
where $o(t)/t\to0$ as $t\downarrow0$. Then, since the value function is increasing and locally bounded,
\begin{align*}
0&\le \lim_{t\downarrow0}\, \,(v^*(\lba+\varepsilon t,z)- v^*(\lba,z)) \\
&= \lim_{t\downarrow0}\, (1-e^{-\rho t}) v^*(\lba+\varepsilon t,z) + \lim_{t\downarrow0} \, e^{-\rho t}v^*(\lba+\varepsilon t,z)-v^*(\lba,z)\\
&\leq \lim_{t\downarrow0} \,\Big(\!\!-\frac{1}{\rho}(1-e^{-\rho t}) u(\delta_1) - e^{-\rho t}\sum_{y\neq z}\P_{z}(z_t=y)[v^*(\lba+\varepsilon t, y)- v^*(\lba+\varepsilon t,z)]\Big)=0.
\end{align*}
This shows continuity of $v^*(\cdot,z)$ at $\lba$. Together with \eqref{eq:proof-diff-bdry-I}, this implies
$$
\lim_{t\downarrow0}\,  \sum_{y\neq z}\frac{\P_{z}(z_t=y)}{ t}[v^*(\lba+\varepsilon t, y)-v^*(\lba+\varepsilon t,z)]  =  \sL v^*(\lba,\cdot)(z).
$$
By concavity, the limit in \eqref{eq:diff-bdry} exists in $[0,\infty]$, hence we only need to exclude the possibility that it is equal to $\infty$. Using the previous results, as $t\downarrow0$,
\begin{align*}
0&\leq\frac{1}{\varepsilon t}  (v^*(\lba+\varepsilon t,z)-v^*(\lba,z)) \\
&\leq \frac{1-e^{-\rho t}}{\varepsilon t} (v^*(\lba+\varepsilon t,z) - \frac{1}{\rho}u(\delta_1)) - e^{-\rho t}\sum_{y\neq z}\frac{\P_{z}(z_t=y)}{\varepsilon t} \, (v^*(\lba+\varepsilon t, y)-v^*(\lba+\varepsilon t,z)) \\
&\to \frac{1}{\varepsilon}\left( \rho v^*(\lba,z)-u(\delta_1) - \sL v^*(\lba,\cdot)(z)\right)<\infty.
\end{align*}
For example, choosing $\delta_1=(r\lba+\lbz)/2$ and letting $\delta_2\uparrow (r\lba+\lbz)$, leads to the explicit bound
\begin{equation}\label{eq:diff-bdry-II}
v^*_a(\lba,z) \leq \frac{2}{r\lba+\lbz}(\rho v^*(\lba,z)-u((r\lba+\lbz)/2) - \sL v^*(\lba,\cdot)(z)).
\end{equation}
\end{proof}

\subsection{Viscosity characterization}\label{ssec:viscosity}
In this section, we set $\cO:=(\lba,\infty)$, so that $\sX=\bar\cO\times\cZ$. To introduce the notion of \emph{constrained viscosity solution}, set
\be\label{eq:hh-problem-HJB-operator}
F(a,z,v,p):= \rho v(z)- H(p) - p (ra+z) -  (\sL v)(z)
\ee
for $(a,z,v,p)\in\sX\times\R^d\times\R$.
We consider the following Hamilton-Jacobi-Bellman equation:
\begin{equation}\label{eq:hh-problem-HJB}
F(a,z,v(a,\cdot),v_a(a,z))=0.
\end{equation}

\begin{Def}[Constrained viscosity solution]\label{def:constraint-viscosity}
{\rm
A continuous function $v:\sX\mapsto\R$ is called a \emph{constrained viscosity solution} of \eqref{eq:hh-problem-HJB} if it is a \emph{subsolution} on $\bar \cO$ and a \emph{supersolution} on $\cO$. Here:
\begin{enumerate}[(1)]
\item $v$ is a \emph{viscosity subsolution on $\bar \cO$} of \eqref{eq:hh-problem-HJB} if the following holds for every $z_*\in\cZ$. Let $a_*\in\bar\cO$ and a function $\varphi\in C^1(\bar\cO)$ be given. Then
$$
v(a_*,z_*)-\varphi(a_*) = \max_{\bar\cO}(v(\cdot,z_*)-\varphi)\quad \Rightarrow \quad F(a_*, z_*, v(a_*,\cdot), \varphi_a(a_*)) \leq 0.
$$
\item $v$ is a \emph{viscosity supersolution on $\cO$} of \eqref{eq:hh-problem-HJB} if the following holds for every $z_*\in\cZ$. Let $a_*\in \cO$ and a function $\varphi\in C^1(\cO)$ be given. Then
$$
v(a_*,z_*)-\varphi(a_*) = \min_{\cO}(v(\cdot,z_*)-\varphi)\quad \Rightarrow \quad F(a_*, z_*, v(a_*,\cdot), \varphi_a(a_*)) \geq 0.
$$
\end{enumerate}
}
\end{Def}

\begin{Rmk}
{\rm 
It is classical that replacing the term \emph{maximum} in Definition \ref{def:constraint-viscosity} by \emph{(strict) global maximum} or \emph{(strict) local maximum} or imposing ``$v(a_*,z_*)-\varphi(a_*)=0$'' yield equivalent definitions. An analogous statement holds for \emph{minima}. We will make use of these observations in the following.
}
\end{Rmk}

\noindent
As explained in Soner \cite{soner1986optimal-I}, the subsolution property on $\bar\cO$ imposes a boundary condition on the derivative.
Indeed, assume that $v$ is a constrained viscosity solution of \eqref{eq:hh-problem-HJB} which is continuously differentiable on $\bar\cO$. Let $z\in\cZ$. Then, for any $\varphi$ such that $v(\cdot,z)-\varphi$ has a maximum in $a=\lba$ so that $v_a(\lba,z)+\alpha=\varphi_a(\lba)$ for some $\alpha\ge0$. By the subsolution property at $\lba$, $F(\lba,z,v(\lba,\cdot),\varphi_a(\lba))\leq 0=F(\lba,z,v(\lba,\cdot),v_a(\lba,z))$. This implies that $H(v_a(\lba,z))\geq H(v_a(\lba,z)+\alpha) + \alpha(r\lba+z)$. Since $H$ is differentiable, this translates into $H'(v_a(\lba,z))\leq r\lba+z$. By Legendre–Fenchel duality, $H'(p)=-(u')^{-1}(p)$, so that $v_a(\lba,z)\geq u'(r\lba+z)$ as in \eqref{eq:HJB}.

The following result follows from the dynamic programming principle, see Soner \cite{soner1986optimal-I,soner1986optimal-II}. For closely related
continuous-time consumption problems with finite-state Markov-chain uncertainty and borrowing constraints, see Shigeta \cite[Proposition 7]{shigeta-control-problem} and Gassiat, Gozzi, and Pham \cite[Proposition 4.1]{gassiat2014investment}.

\begin{Thm}\label{thm:value-is-constr-viscosity}
The value function $v^*$ of \eqref{eq:hh-problem-value-function} is a constrained viscosity solution of \eqref{eq:hh-problem-HJB}.
\end{Thm}

\noindent 
In order to prove Theorem \ref{thm:household}, we first establish continuous differentiability of the value function and then prove a comparison principle.

\begin{Prop}\label{prop:cts-differentiability}
For every $z\in\cZ$, the value function $v^*(\cdot,z)$ is continuously differentiable on $\bar\cO$.
\end{Prop}

\begin{proof}
Introduce the \emph{full Hamiltonian}, for $(a,z,p)\in\sX\times (0,\infty)$ and $\varphi:\cZ\mapsto\R$,
$$
\sH(a,z,\varphi,p):= H(p)+p(ra+z)+\sL \varphi(z).
$$
First note that $v^*(\cdot,z)$ is both continuous and Lebesgue almost everywhere differentiable on $\cO$ by the concavity established in Proposition \ref{prop:basic-prop-value}. Towards a contradiction, assume there exists $(a_*,z_*)\in \cO\times\cZ$ such that the subdifferential $\partial_a v(a_*,z_*)$ consists of more than one element. It is classical that $\partial_a v^*(a_*,z_*)=[\partial_a^+ v^*(a_*,z_*),\partial_a^-v^*(a_*,z_*)]$ where $\partial_a^\pm$ denote the one-sided derivatives. Set
$$
p_* := \frac12 \left(\partial_a^- v^*(a_*,z_*) + \partial_a^+  v^*(a_*,z_*)\right),
$$
and define the test function $\varphi(a,z):=v^*(a_*,z)+p_*(a-a_*)$, $(a,z)\in\sX$. By concavity of $v^*(\cdot,z_*)$, 
$$
0=v^*(a_*,z_*)-\varphi(a_*,z_*)= \max_{\bar\cO} (v^*(\cdot,z_*)-\varphi(\cdot,z_*)).
$$
Since $v$ is a viscosity subsolution and the Hamiltonian is strictly convex in $p$, we have
\begin{equation}\label{eq:proof-cts-diff}
\begin{aligned}
\rho v^*(a_*,z_*) &\leq\sH(a_*,z_*,v^*(a_*,\cdot),p_*)\\
&< \frac12\left(\sH(a_*,z_*,v^*(a_*,\cdot),\partial^+_av^*(a_*,z_*)) +\sH(a_*,z_*,v^*(a_*,\cdot),\partial^-_av^*(a_*,z_*))\right)
\end{aligned}
\end{equation}
Now let $a_n\in\cO$ be a sequence converging to $a_*$ from below such that $v(\cdot,z_*)$ is differentiable at each $a_n$. By the supersolution property, 
$$
\rho v^*(a_n,z_*) \geq \sH(a_n,z_*,v^*(a_n,\cdot),v^*_a(a_n,z_*)),\qquad n\ge1.
$$
Using concavity of $v^*(\cdot,z_*)$ it is classical that $v^*_a(a_n,z_*)\to \partial^-_a v^*(a_*,z_*)$. Using continuity of the value function and of the Hamiltonian $\sH$, 
$$
\rho v^*(a_*,z_*) \geq \sH(a_*,z_*,v^*(a_*,\cdot),\partial^-_a v^*(a_*,z_*)).
$$
Repeating this argument with a sequence $a_n\downarrow a_*$ yields
$$
\rho v^*(a_*,z_*) \geq \sH(a_*,z_*,v^*(a_*,\cdot),\partial^+_a v^*(a_*,z_*)).
$$
Together, these contradict \eqref{eq:proof-cts-diff}, and therefore establish differentiability on $\cO$. Using concavity of the value function and \cite[Corollary 25.5.1]{Rockafellar_1970}, this establishes continuous differentiability on $\cO$. Combined with differentiability at the boundary (see Proposition \ref{prop:diff-boundary}), the claim follows.
\end{proof}

\noindent 
Define the admissible set of interest rates 
\be\label{eq:admissible-interest-rates}
R:=\{r\leq \rho: r<-\lbz/\lba \ \ \text{if} \ \ \lba<0\}
\ee
that respect Assumption \ref{asm:borrowing-constraint}.

\begin{Lem}\label{lem:boundedness-v-v_a}
Both the value function $v^*$ and its derivative $v^*_a$ are locally bounded
functions of $(a,z,r)\in\sX\times R$. More precisely,
\begin{equation}\label{eq:equicontinuity}
\sup_{(a,z)\in\sX} v^*_a(a,z;r)
\leq
\max_{z\in\cZ} v_a^*(\lba,z;r)
\leq
\frac{2(\rho+\Lambda)}{r\lba+\lbz}
\left(
\max_{z\in\cZ} v^*(\lba,z;r)-\frac{1}{\rho}u\!\left(\frac{r\lba+\lbz}{2}\right)
\right).
\end{equation}
where $\Lambda:=\max_{z\in\cZ}\sum_{y\neq z}\lambda(z,y).$
\end{Lem}

\begin{proof}
Fix $r\in R$ and set
 $$
b_r:=r\lba+\lbz,\qquad
M_r:=\max_{z\in\cZ} v^*(\lba,z;r),
\qquad
\Lambda:=\max_{z\in\cZ}\sum_{y\neq z}\lambda(z,y).
 $$ 
Clearly, for a fixed interest rate, the value function is
increasing and bounded from below. As in Proposition
\ref{prop:basic-prop-value}, the constant consumption rule
$c_t\equiv b_r/2$ is admissible, and the budget estimate gives
 $$
\frac{1}{\rho}u \Big(\frac{b_r}{2}\Big)
\leq
v^*(a,z;r)
\leq \frac{1}{\rho}
u\bigl(\rho a-(\rho-r)\lba+\bar z\bigr),
\qquad (a,z)\in\sX.
 $$
In particular, $v^*$ is locally bounded as a function of $(a,z,r)$.

We now prove the derivative bound. By Propositions
\ref{prop:basic-prop-value}, \ref{prop:diff-boundary}, and
\ref{prop:cts-differentiability}, the derivative
$v_a^*(\cdot,z;r)$ is continuous, non-negative, and decreasing on
$[\lba,\infty)$. Hence
 $$
\sup_{(a,z)\in\sX}v_a^*(a,z;r)
\leq \max_{z\in\cZ}v_a^*(\lba,z;r).
 $$
Let $m_r:=u (b_r/2)/\rho$. The lower bound above gives $v^*(\lba,y;r)\ge m_r$ for every $y\in\cZ$.
Therefore, for every $z\in\cZ$,
 $$
-\sL v^*(\lba,\cdot;r)(z)
=
\sum_{y\neq z}\lambda(z,y)
[v^*(\lba,z;r)-v^*(\lba,y;r)]
\leq
\Lambda(M_r-m_r).
 $$
Using the boundary estimate \eqref{eq:diff-bdry-II}, we obtain
 $$
\begin{aligned}
v_a^*(\lba,z;r)
&\leq
\frac{2}{b_r}
\Big(
\rho v^*(\lba,z;r)
-u \Big(\frac{b_r}{2}\Big)
-\sL v^*(\lba,\cdot;r)(z)
\Big)  \\
&\leq
\frac{2}{b_r}
\Big(
\rho(M_r-m_r)+\Lambda(M_r-m_r)
\Big)  = \frac{2(\rho+\Lambda)}{b_r}\, (
M_r-m_r).
\end{aligned}
 $$
Taking the maximum over $z\in\cZ$ proves \eqref{eq:equicontinuity}.

Finally, the right-hand side of \eqref{eq:equicontinuity} is locally bounded in $r\in R$: on compact subsets of $R$, $b_r=r\lba+\lbz$ is bounded away from zero, while $M_r$ is locally bounded by the preceding upper bound for $v^*$. Hence, $v_a^*$ is locally bounded in $(a,z,r)$.
\end{proof}

\noindent
To define a feedback control via the maximizer of the Hamiltonian, we need to ensure that that the value function is \emph{strictly} increasing, or equivalently, $v_a^*>0$.

\begin{Lem}\label{lem:value-strict-increasing}
For every $(a,z)\in\sX$, $v_a^*(a,z)>0$.
\end{Lem}

\begin{proof}
By Propositions \ref{prop:basic-prop-value} and \ref{prop:cts-differentiability}, for every $z\in\mathcal Z$ the function
$v^*(\cdot ,z)$ is increasing, concave, and continuously differentiable on
$[\lba,\infty)$. Hence $p_z(a):=v_a^*(a,z)$ is non-negative and decreasing in $a$.

At the borrowing constraint, the boundary condition together with Assumptions \ref{asm:borrowing-constraint} and \ref{asm:utility-function} imply that $p_z(\lba)\ge u'(r\lba+z)>0$. It thus remains
to rule out $p_z(\cdot)=0$ at interior points. Towards a contradiction suppose that $p_{z_*}(a_*)=0$ for some $(a_*,z_*)\in\cO\times\cZ$. Since $p_{z_*}(\cdot)$ is non-negative and decreasing, this implies that $v^*(\cdot,z_*)$ is constant on $[a_*,\infty)$.

Let $\bar u:= \sup_{c>0} u(c) \in \R\cup\{\infty\}$. If $\bar u=+\infty$, then $H(0)=+\infty$, contradicting the HJB equation on $[a_*,\infty)\times\cZ$. It remains to check the case $H(0)=\bar u<\infty$. The HJB equation in state $z_*$ becomes
$$
\rho v^*(a,z_*) = \bar u+\sL v^*(a,\cdot)(z_*),\qquad a\ge a_*.
$$
This implies that the map
$$
a\longmapsto \sum_{y\ne z_*}\lambda(z_*,y)v^*(a,y)
$$
must be constant on $[a_*,\infty)$. Now, each $v^*(\cdot,y)$ is increasing, and
the coefficients $\lambda(z_*,y)$ are non-negative. Therefore, for every state $y\in\cZ$ with
$\lambda(z_*,y)>0$, $v^*(\cdot,y)$ is constant on $[a_*,\infty)$. Using irreducibility of the
income process, we conclude that there are constants $K(z)$, $z\in\cZ$, such that  $v^*(a,z)=K(z)$ for every $a\ge a_*$. Using the HJB equation, $K(\cdot)$ satisfies
$$
\rho K=\bar u+\sL K.
$$
Let $z_+$ maximize $K(z)$ and $z_-$ minimize $K(z)$ over $z\in\cZ$. Then, $\sL K(z_+)\le0\le \sL K(z_-)$, hence $\rho K(z_+) \leq \bar u\leq \rho K(z_-)$, so that $K(z)=\bar u/\rho$ for all $z$.
This contradicts the upper bound from Proposition \ref{prop:basic-prop-value},
\begin{equation*}
\frac{\bar u}{\rho}=v^*(a,z)\le \frac1\rho
u(\rho a-(\rho-r)\lba+\bar z)
<\frac{\bar u}{\rho},\qquad a\ge a_*.\qedhere
\end{equation*}
\end{proof}

\noindent
We now define the feedback policy
$$
c^*(a,z):=(u')^{-1}(v_a^*(a,z)),\qquad (a,z)\in\sX.
$$
Since $(u')^{-1}$ is strictly decreasing and using Lemma \ref{lem:value-strict-increasing}, $c^*(\cdot,z)$ is well-defined, continuous, and increasing. Moreover, as discussed after Definition \ref{def:constraint-viscosity}, $c^*(\lba,z)\leq r\lba +z$ for every $z\in\cZ$. However, it may fail to be Lipschitz at the borrowing constraint, and we now address admissibility and optimality of the feedback control $c^*(a,z)$.

\begin{proof}[Proof of Theorem \ref{thm:household} \eqref{thm:household-control}]

\emph{Admissibility.} To show $c^*\in\sC_{\mathrm{cl}}$, it is sufficient to establish the existence and uniqueness of the deterministic initial-value problems
$$
\frac{\d}{\dt} a(t) = ra(t) + z - c^*(a(t),z),\qquad a(0)\in\bar\cO,\qquad \text{for fixed} \ z\in\cZ,
$$
which describe the wealth evolution between jump times.  Since $c^*(\cdot,z)$ is a positive and continuous function with $c^*(\lba,z)\leq r\lba+z$, a global solution which does not leave $\bar\cO$ when started in $\bar\cO$ exists by Peano's existence theorem. To argue uniqueness, we let $\tilde a(t)$ and $a(t)$ be two solutions to the same initial-value problem and set $w(t):=\tilde a(t)-a(t)$. Note that
$$
(c^*(\tilde a ,z)-c^*(a,z))\, (\tilde a-a) \geq0, \qquad a,\tilde a \in\bar \cO,
$$
since $c^*(\cdot,z)$ is increasing. Therefore,
$$
\frac12 \frac{\d}{\dt} \big(w(t)^2\bigr) = \bigl[rw(t)-(c^*(\tilde a(t),z)-c^*(a(t),z))\bigr] w(t)\leq rw(t)^2.
$$
By Gr\"onwall's inequality, $w(t)^2 \leq w(0)^2 e^{2rt}=0$, showing uniqueness. 

\vspace{1em}
\emph{Optimality.} Fix an initial condition $(a,z)\in\sX$. First observe that any controlled state process $(a_t)_{t\ge0}$ satisfies $\lba \le a_t\le \eta(1+ t+ e^{r_+t})$ for some $\eta<\infty$. Since the value function is of at most linear growth and $r<\rho$, we see that $e^{-\rho t} \E_{a,z}[|v^*(a_t,z_t)|]\to0$ as $t\to\infty$. Let $(a^*_t)$ denote the wealth process controlled by $c^*(\cdot,\cdot)$ and started from $(a,z)$. For any $T\ge0$, by Dynkin's formula and the fact that $v^*$ is a classical solution to the dynamic programming equation,
\begin{align*}
&\E_{a,z}[e^{-\rho T}v^*(a^*_T,z_T)] \\
&\qquad = v^*(a,z)+\E_{a,z}\left[ \int_0^T e^{-\rho t} [-\rho v^*(a^*_t,z_t)+v^*_a(a^*_t,z_t)(ra_t^*+z_t-c^*(a^*_t,z_t))+\sL v^*(a^*_t,\cdot)(z_t)]\,\d t \right]\\
&\qquad = v^*(a,z)-\E_{a,z}\left[ \int_0^T e^{-\rho t} u(c^*(a^*_t,z_t))\,\d t \right].
\end{align*}
Taking $T\to\infty$ yields 
$$
\E_{a,z}\left[ \int_0^\infty \! e^{-\rho t} u(c^*(a^*_t,z_t))\,\d t \right] = v^*(a,z),
$$
establishing optimality of the feedback control $c^*$ since the initial condition $(a,z)\in\sX$ was arbitrary.

\vspace{1em}
\emph{Uniqueness.} Let $\tilde c \in \sC_{\mathrm{cl}}$ be an optimal feedback control. The claim is that $\tilde c(\cdot,z)=c^*(\cdot,z)$, for every $z\in\cZ$. Define 
$$
\Delta(a,z) := H(v_a^*(a,z)) - (u(\tilde c(a,z))- v_a^*(a,z) \tilde c(a,z))\geq0,\qquad (a,z)\in\sX.
$$
Since $u$ is assumed to be strictly concave, we have $\Delta(a,z)=0$ if and only if $c^*(a,z)=\tilde c(a,z)$, and it suffices to prove $\Delta(\cdot,z)=0$, for every $z\in\cZ$. For an initial condition $(a,z)\in\sX$, let $(\tilde a_t,z_t)$ denote the state process generated by $\tilde c(\cdot,\cdot)$. Since $\tilde c$ is optimal for every initial condition,
$$
v^*(a,z)= \E_{a,z} \left[ \int_0^T e^{-\rho t} u(\tilde c(\tilde a_t,z_t))\,\dt + e^{-\rho T} v^*(\tilde a_T,z_T)\right],\qquad T\ge0.
$$
At the same time, Dynkin's formula and the HJB equation satisfied by $v^*$ yield
\begin{align*}
&\E_{a,z} \left[ \int_0^T e^{-\rho t} u(\tilde c(\tilde a_t,z_t))\,\dt + e^{-\rho T} v^*(\tilde a_T,z_T)\right]\\
&\qquad = v^*(a,z) + \E_{a,z}\left[\int_0^T e^{-\rho t}[u(\tilde c(\tilde a_t,z_t))-H(v_a^*(\tilde a_t,z_t)) - v_a^*(\tilde a_t,z_t)\tilde c(\tilde a_t,z_t)]\,\dt \right].
\end{align*}
Combining the last two displays,
\be\label{eq:uniqueness-proof-I}
\E_{a,z}\left[\int_0^T e^{-\rho t}\Delta (\tilde a_t,z_t)\,\dt \right]=0,\qquad \forall T\ge0,\quad \forall (a,z)\in\sX.
\ee
Towards a contradiction, assume that there exist $(a_*,z_*)\in\sX$ with $\Delta(a_*,z_*)>0$. By continuity, we may choose $\delta,T>0$ and an interval $a_*\in J\subset[\lba,\infty)$ such that  $\Delta(\cdot,z_*)\geq \delta $  on $J$ and such that the solution $x(t)$ to
$$
\dot x(t) = r x(t) + z_*- \tilde c(x(t),z_*),\qquad x(0)=a_*
$$
satisfies $x(t)\in J$ for all $t\leq T$. Define the event $B:=\{z_t=z_*\ \text{for all} \ 0\le t\le T\}$, which has strictly positive probability. Let $(\tilde a_t,z_t)$ be the state process started from $(a_*,z_*)$ and controlled by $\tilde c$. Then $\tilde a_t=x(t)$ on $B$ and by choice of $J$ and $\delta$,
\begin{align*}
\E_{a_*,z_*}\left[\int_0^T e^{-\rho t}\Delta (\tilde a_t,z_t)\,\dt \right] &\geq \E_{a_*,z_*}\left[\chi_B\, \int_0^T e^{-\rho t}\Delta (\tilde a_t,z_t)\,\dt \right]\geq \delta\, \P(B) \int_0^T e^{-\rho t}\,\dt>0,
\end{align*}
contradicting \eqref{eq:uniqueness-proof-I}.
\end{proof}

Using the existence of a maximizer, we easily obtain strict concavity of the value function $v^*(\cdot,z)$.

\begin{Cor}
For every $z\in\cZ$, the value function $v^*(\cdot,z)$ is strictly concave.
\end{Cor}

\begin{proof}
Fix $z\in\cZ$ and $a_1< a_2$ in $\bar\cO$, and let $(c^{*,i}_t)\in\sC(a_i,z)$ be an optimal control for the problem with initial condition $(a_i,z)$, $i=1,2$. 
By strict monotonicity of the value function (Lemma \ref{lem:value-strict-increasing}), it is clear that the optimal controls $(c^{*,1}_t)$ and $(c^{*,2}_t)$ differ on a set of positive $\P\otimes\mathrm{Leb}$-measure. By strict concavity of $u(\cdot)$, for any $\delta\in(0,1)$,
\begin{align*}
v^*(\delta a_1 + (1-\delta) a_2,z) &\geq \E\int_0^\infty \! e^{-\rho t} u(\delta c^{*,1}_t + (1-\delta) c_t^{*,2})\,\d t\\
&> \delta \E\int_0^\infty \! e^{-\rho t} u( c_t^{*,1} )\,\d t + (1-\delta) \E\int_0^\infty \! e^{-\rho t}u(c_t^{*,2})\,\d t\\
&= \delta v^*(a_1,z) + (1-\delta) v^*(a_2,z),
\end{align*}
showing strict concavity.
\end{proof}

Next, we present a comparison result. 

\begin{Thm}\label{thm:comparison}
Let $r<\rho$. If $v$ is a viscosity subsolution of \eqref{eq:hh-problem-HJB} on $\bar \cO$ satisfying the growth condition \eqref{eq:growth-condition} and $w$ a viscosity supersolution of \eqref{eq:hh-problem-HJB} on $\cO$ which is bounded from below, then $v\leq w$ on $\bar\cO\times\cZ$. If $r=\rho$, the same conclusion holds true provided the growth condition \eqref{eq:growth-condition} is strengthened to $\limsup_{a\to\infty}v(a,z)/a=0$ for all $z\in\cZ$.
\end{Thm}

We provide a proof of Theorem \ref{thm:comparison} in Appendix \ref{app:comparison}. Equipped with it, we are in position to prove Theorem \ref{thm:household} \eqref{thm:household-value}.

\begin{proof}[Proof of Theorem \ref{thm:household} \eqref{thm:household-value}]
By Theorem \ref{thm:value-is-constr-viscosity}, $v^*$ is a constrained viscosity solution of the HJB equation. By Proposition \ref{prop:cts-differentiability}, it is $C^1$ on $\bar \cO$, so that it is a classical solution. Uniqueness follows from Theorem \ref{thm:comparison}.
\end{proof}

We conclude this section with a continuity result in the interest rate. To state this, let $v^*(a,z;r)$ denote the value function corresponding to an interest rate $r$. Recall the admissible set of interest rates $R$ from \eqref{eq:admissible-interest-rates}. We first present a short lemma that will be used repeatedly.

\begin{Lem}\label{lem:Gamma}
For $b>0$, define $\Gamma_b(p):=H(p)+bp$ for $p>0$. Then, $\Gamma_b$ is minimized at $p=u'(b)$, $C^2$, and strictly increasing on $[u'(b),\infty)$ with $\Gamma_b(\infty)=\infty$.
\end{Lem}

\begin{Lem}\label{lem:joint-continuity-(a,r)}
For every $z\in\cZ$, both $v^*(a,z;r)$ and $v^*_a(a,z;r)$ are jointly continuous in $(a,r)\in[\lba,\infty)\times R$. In fact, as $r'\to r$, we have
$$
v^*(\cdot,\cdot;r')\to v^*(\cdot,\cdot;r)\quad\text{and}\quad
v^*_a(\cdot,\cdot;r')\to v^*_a(\cdot,\cdot;r)
$$
locally uniformly on $\sX=\bar\cO\times\cZ$.
\end{Lem}

\begin{proof}
Let $r_k\in R$ with $r_k\to r$, and write $v^k(a,z):=v^*(a,z;r_k)$ and $v(a,z):=v^*(a,z;r)$. 

\emph{Step 1.} For any compact subset $I\subset\bar\cO$,
$$
\sup_{k\ge1}\sup_{(a,z)\in I\times\cZ}
\bigl(|v^k(a,z)|+|v^k_a(a,z)|\bigr)<\infty
$$
by 
Lemma
\ref{lem:boundedness-v-v_a}.
Hence, the sequence $(v^k)_{k\ge1}$ is uniformly bounded and equicontinuous on $I\times\cZ$. By Arzela-Ascoli, any subsequence has a further subsequence, still denoted by
$v^k$, such that $v^k\to\hat v$ locally uniformly on $\sX$.

\emph{Step 2.}
For $k\ge1$, let
$$
F_k(a,z,\varphi,p):=\rho \varphi(z)-H(p)-p(r_k a+z)-\sL \varphi(z),
$$
and let $F$ be as in \eqref{eq:hh-problem-HJB-operator}. We claim that $\hat v$ is a constrained viscosity solution of the limiting HJB equation \eqref{eq:hh-problem-HJB} with
interest rate $r$. We briefly outline the subsolution argument. Fix $z\in\cZ$ and let $\varphi\in C^1$ touch $\hat v(\cdot,z)$ from above at $a_*\in\bar\cO$, with a strict
local maximum. By the local uniform convergence, there are $a_k\to a_*$ such that $v^k(\cdot,z)-\varphi$ has a local maximum at $a_k$, so that $ F_k(a_k,z,v^k(a_k,\cdot),\varphi_a(a_k))\le0$ for $k\ge1$.  If $H(\varphi_a(a_*))=+\infty$, the desired subsolution inequality follows. Otherwise, $H$ is continuous along $\varphi_a(a_k)\to\varphi_a(a_*)$, and passing to the limit gives $F(a_*,z,\hat v(a_*,\cdot),\varphi_a(a_*))\le0$ as desired. For the supersolution property, we note that necessarily $\varphi_a(a_*)>0$. Indeed, $F_k\ge0$ together with the boundedness of $(\varphi_a(a_k))$ and local boundedness of $v^*$ implies  that $H(\varphi_a(a_k))$ is uniformly bounded in $k\ge1$. Passing to the limit then gives the supersolution inequality.

The limit $\hat v$ is bounded from below, and the upper bound in Proposition
\ref{prop:basic-prop-value}, together with $r_k\to r$, gives
$$
\hat v(a,z)
\le
\frac1\rho u\bigl(\rho a-(\rho-r)\lba+\bar z\bigr).
$$
Since $u'(c)\to0$ as $c\to\infty$ by Assumption \ref{asm:utility-function}, the right-hand side is sublinear in $a$. Hence, by the comparison Theorem \ref{thm:comparison}, $\hat v=v^*(\cdot,\cdot;r)$ and we conclude $v^*(\cdot,\cdot;r_k)\to v^*(\cdot,\cdot;r)$
locally uniformly on $\sX$.

\emph{Step 3.}
We now prove convergence of the derivatives, and fix $z\in\cZ$. On compact intervals of $\cO$, local uniform convergence of
concave functions to the continuously differentiable concave function $v(\cdot,z)$ implies $v^k_a\to v_a$  pointwise on $\cO$ by \cite[Thm.~25.7]{Rockafellar_1970}. For the convergence at the borrowing constraint set
$$
p_k(z):=v^k_a(\lba,z),\qquad
p(z):=v_a(\lba,z),
\qquad
b_k(z):=r_k\lba+z,\qquad
b(z):=r\lba+z .
$$
The HJB equation reads
$$
\Gamma_{b_k(z)}(p_k(z))
=
\rho v^k(\lba,z)-\sL v^k(\lba,\cdot)(z)
\qquad \text{and} \qquad 
p_k(z)\ge u'(b_k(z)),
$$
where $\Gamma_b$ is as in Lemma \ref{lem:Gamma}. Since $v^k(\lba,\cdot)\to v(\lba,\cdot)$, the right-hand side converges to $\rho v(\lba,z)-\sL v(\lba,\cdot)(z)$, which equals $H(p(z))+b(z)p(z)$ by the limiting HJB equation. By Lemma \ref{lem:Gamma}, $\Gamma_b$ has continuous inverse, so that $p_k(z)\to p(z)$, $z\in\cZ$. 

Taken together, this proves $v_a^k\to v_a$ pointwise on $\sX$. Now, for each fixed $k,z$, the function $a\mapsto v_a^k(a,z)$ is decreasing, and $a\mapsto v_a(a,z)$ is continuous. Since pointwise convergence of monotone functions to a continuous limit is uniform on compact intervals, and $\cZ$ is finite, $v_a^*(\cdot,\cdot;r_k)\to v_a^*(\cdot,\cdot;r)$
locally uniformly on $\bar\cO\times\cZ$.
\end{proof}

\section{Invariant distribution}\label{sec:invariant-distribution}

This section proves the existence, uniqueness, and further properties of the invariant distribution of the optimal state process. Throughout this section, Assumptions
\ref{asm:standing}, \ref{asm:borrowing-constraint} and \ref{asm:utility-function} hold. Whenever the existence of an invariant distribution is asserted, Assumption \ref{asm:utility-tail} is also imposed. 

\subsection{Proof of Theorem \ref{thm:invariant-distribution}}

This subsection follows arguments presented in \cite{acikgoz2018existence, achdou2022income, shigeta-equilibrium} to establish that the optimally controlled state is an exponentially ergodic Markov process. 

\begin{Lem}[Euler's equation]\label{lem:Euler}
For any $a>\lba$ and $z\in\cZ$,
\begin{equation}
(\rho-r)v^*_a(a,z) = v^*_{aa}(a,z)s^*(a,z) + \sL v^*_a(a,\cdot)(z),\label{eq:Euler-interior} 
\end{equation}
where we define $v^*_{aa}(a,z)s^*(a,z):=0$ if $s^*(a,z)=0$, regardless of whether $v^*_{aa}(a,z)$ exists or not.
At the boundary $\{a=\lba\}$, 
\begin{align}\label{eq:Euler-brdy-nonzero}
(\rho-r)v^*_a(\lba,z) &= v^*_{aa}(\lba,z)s^*(\lba,z) + \sL v^*_a(\lba,\cdot)(z), &\text{if} \ \ s^*(\lba,z)> 0,\\
(\rho-r)v^*_a(\lba,z) &\geq \sL v^*_a(\lba,\cdot)(z), &\text{if}\ \  s^*(\lba,z)= 0.\label{eq:Euler-bdry-zero}
\end{align}
\end{Lem}

\begin{proof}
Fix $z\in\cZ$ and set
$$
p_z(a):=v_a^*(a,z),\quad
c_z(a):=(u')^{-1}(p_z(a)),\quad
s_z(a):=ra+z-c_z(a),\quad R_z(a):=(\rho-r)p_z(a)-\sL p(a,\cdot)(z).
$$
Recall $\Gamma_b(p)=H(p)+bp$ from Lemma \ref{lem:Gamma}. The HJB equation \eqref{eq:hh-problem-HJB} may be written as
\be\label{eq:proof-euler-I}
\Gamma_{ra+z}(p_z(a))
=
\rho v^*(a,z)-\sL v^*(a,\cdot)(z).
\ee

\emph{Interior case $a>\lba$.}
First suppose $a>\lba$ and $s_z(a)\ne0$. Since $H'(p_z(a))=-c_z(a)$, 
$$
\partial_p\{\Gamma_{ra+z}(p)\}\big|_{p=p_z(a)}
=
H'(p_z(a))+ra+z
=
s_z(a)\ne0,
$$
and the right-hand side of \eqref{eq:proof-euler-I} is continuously differentiable in $a$,
the implicit function theorem implies that $p_z$ is $C^1$
in a neighborhood of $a$. Differentiating \eqref{eq:proof-euler-I} on this neighborhood yields \eqref{eq:Euler-interior}.

It remains to consider an interior point $a>\lba$ with $s_z(a)=0$.
Set $b:=ra+z$ and $p:=p_z(a)$. Then $c_z(a)=b$ and $p = u'(b)$.
Subtracting the HJB equation at $a$ from the HJB equation at $a+h$ leads to
$$
\rho\frac{v^*(a+h,z)-v^*(a,z)}{h} - r p_z(a+h)  - \sL\left(
\frac{v^*(a+h,\cdot)-v^*(a,\cdot)}{h}
\right)(z)
=
\frac{
\Gamma_b(p_z(a+h))-\Gamma_b(p)
}{h} .
$$
For $h>0$, the right-hand side is non-negative. Letting
$h\downarrow0$ yields
$$
R_z(a)=(\rho-r)p_z(a)-\sL p(a,\cdot)(z)\ge0.
$$
For $h<0$, the first term on the right-hand side is non-positive. Letting
$h\uparrow0$ yields $R_z(a)\le0$, hence $R_z(a)=0$. With the convention
$v_{aa}^*(a,z)s^*(a,z):=0$ whenever $s^*(a,z)=0$, this proves the interior identity.

\emph{Boundary case $a=\lba$.}  If $s_z(\lba)>0$, then by continuity
$s_z>0$ on a right-neighborhood of $\lba$. The preceding implicit-function
argument applies on $(\lba,\lba+\varepsilon)$ and gives
$$
p_z'(a)=\frac{R_z(a)}{s_z(a)}.
$$
The right-hand side extends continuously to $a=\lba$, so $p_z$ has a right
derivative at $\lba$ and \eqref{eq:Euler-brdy-nonzero} follows. Finally, suppose $s_z(\lba)=0$. Repeating the preceding difference-quotient
argument with $a=\lba$ and $h>0$ only gives $R_z(\lba)\ge0$. This proves the boundary inequality \eqref{eq:Euler-bdry-zero}.
\end{proof}

If $r<\rho$, then Euler's equation immediately implies that there exists an income state in which the individuals dissaves. Indeed, for any $a>\lba$, let 
$$
z(a)=\mathrm{arg\,max}\{v^*_a(a,z):z\in\cZ\}=\mathrm{arg\,min}\{c^*(a,z):z\in\cZ\}.
$$
If $s^*(a,z(a))\geq0$, then Lemma \ref{lem:Euler} shows 
$$
0<(\rho-r)v^*_a(a,z(a)) \leq \sum_{y\neq z(a)} \lambda(z(a),y)\left(v^*_a(a,y)-v^*_a(a,z(a))\right)\leq 0,
$$
where the last inequality follows by definition of $z(a)$. This is a contradiction, and hence the lowest consumption level always corresponds to negative savings. In fact, the following lemma holds.

\begin{Lem}\label{lem:properties-saving-rules}
Assume $r\le \rho$. Then:
\begin{enumerate}[(i)]
\item \label{lem:properties-saving-rules-I} $s^*(a,\lbz)<s^*(\lba,\lbz)=0$ for all $a>\lba$.
\item  \label{lem:properties-saving-rules-II} $\sL v^*_a(\lba,\cdot)(\lbz)=\sum_{y\neq \lbz}\lambda(\lbz,y)(v^*_a(\lba,y)-v^*_a(\lba,\lbz))<0$.
\item  \label{lem:properties-saving-rules-III} If $r<\rho$, there exists $\bar a<\infty$ such that $s^*(a,z)<0$ for all $(a,z)\in(\bar a,\infty)\times\cZ$. The upper wealth limit $\bar a$, viewed as a function of $r$, is locally bounded.
\item  \label{lem:properties-saving-rules-IV} Consider the (deterministic) ODE
$$
\dot x(t) = s^*(x(t),\underline{z}), \quad x(0)\geq\lba,
$$
and define $\tau:=\inf\{t\geq0:x(t)=\lba\}$ as the first hitting time of $\lba$. For any $L<\infty$ and any compact set $K$ of admissible interest rates, there exists $T_{K,L}<\infty$ such that $\tau\leq T_{K,L}$, uniformly in $x(0)\in[\lba,L]$ and $r\in K$.
\end{enumerate}
The parts \eqref{lem:properties-saving-rules-I} and \eqref{lem:properties-saving-rules-II} hold under Assumptions \ref{asm:standing}--\ref{asm:utility-function}. Parts \eqref{lem:properties-saving-rules-III} and \eqref{lem:properties-saving-rules-IV} require Assumption \ref{asm:utility-tail} in addition.
\end{Lem}

\begin{proof}
We provide proofs for \eqref{lem:properties-saving-rules-I} and \eqref{lem:properties-saving-rules-II} in the case $r=\rho$ in Appendix \ref{app:r-equal-rho}.\\

$(i)$ Let $r<\rho$. Towards a contradiction assume $s^*(a,\underline{z})\geq0$ for some $a>\lba$. Let $z_*\in \mathrm{arg\,max}\{v_a^*(a,z):z\in\cZ\}$. Then 
$$
c^*(a,z_*)\leq c^*(a,\lbz)\leq ra+\lbz\leq ra+z_*\qquad \Rightarrow \qquad s^*(a,z_*)\geq0.
$$
Further, $\sL v^*_a(a,\cdot)(z_*)\le0$ and, by concavity, $v_{aa}^*(a,z_*)s^*(a,z_*)\le 0$. By Lemma \ref{lem:Euler},
$$
0< (\rho-r) v_a^*(a,z_*) = v_{aa}^*(a,z_*)s^*(a,z_*)+\sL v^*_a(a,\cdot)(z_*)\le0,
$$
contradiction. Hence, $s^*(a,\lbz)<0$ for $a>\lba$. Finally, continuity and the boundary condition yield $s^*(\lba,\lbz)=0$.

$(ii)$ Let $r<\rho$. Towards a contradiction, assume $\sL v^*_a(\lba,\cdot)(\lbz)\geq0$. Then, there exists at least one $z\neq\lbz$ with $v_a^*(\lba,z)\geq v_a^*(\lba,\lbz)$. Let $z_*\in \mathrm{arg\,max}\{v_a^*(\lba,z):z\neq \lbz\}$. As above, this implies $\sL v^*_a(\lba,\cdot)(z_*)\leq0$, and also
$$
s^*(\lba,z_*)= r\lba + z_*-c^*(\lba,z_*)\geq r\lba + z_*-c^*(\lba,\lbz) > s^*(\lba,\lbz)=0.
$$
Lemma \ref{lem:Euler} then implies
$$
\sL v^*_a(\lba,\cdot)(z_*) = (\rho-r) v_a^*(\lba,z_*)- v_{aa}^*(\lba,z_*) s^*(\lba,z_*)>0,
$$
contradiction.

$(iii)$ The claim is obvious if $r<0$, so without loss of generality assume $r\in[0,\rho)$.

\emph{Case $r=0$.} Proposition \ref{prop:basic-prop-value} implies that the value function $v^*(\cdot,z;0)$ is concave, increasing, and sublinear for every $z\in\cZ$. Since it is continuously differentiable by Proposition \ref{prop:cts-differentiability}, $v^*_a(\cdot,z;0)$ is non-negative and decreasing. Hence $\ell(z):=\lim_{a\to\infty} v^*_a(a,z;0)$ exists in $[0,\infty)$. Sublinearity forces $\ell(z)=0$. Therefore, $c^*(a,z;0)=(u')^{-1}(v^*_a(a,z;0))\to\infty$ as $a\to\infty$. Since $\cZ$ is finite, there exists $\overline{a}<\infty$ such that $c^*(a,z;0)> \overline{z}:=\max \cZ$ for all $a>\overline{a}$ and $z\in\cZ$. Thus, $s^*(a,z;0)=z-c^*(a,z;0)<0$ for all $a>\overline{a}$, $z\in\cZ$ as desired.

\emph{Case $r\in(0,\rho)$.} From the above discussion, $z(a):=\mathrm{arg\,max}\{v_a(a,z):z\in\cZ\}$ satisfies $s^*(a,z(a))<0$. Then, we have the following implications
\begin{align*}
s^*(a,z)\geq0\quad &\Rightarrow\quad c^*(a,z)\leq ra+z\quad \Rightarrow\quad v_a(a,z)\geq u'(ra+z),\quad \text{and}\\
s^*(a,z(a))<0\quad &\Rightarrow\quad c^*(a,z(a))> ra+z(a)\quad \Rightarrow\quad v_a(a,z(a))< u'(ra+z(a)).
\end{align*}
Fix $z\in\cZ$ and let $\Lambda:=\max_{z\in\cZ}\sum_{y\neq z}\lambda(z,y)\in(0,\infty)$. Then, $s^*(a,z)\geq0$ implies
\begin{align*}
0<\rho-r &\leq \sum_{y\neq z} \lambda(z,y)\left(\frac{v_a(a,y)}{v_a(a,z)}-1\right) \leq \Lambda \, \left(\frac{v_a(a,z(a))}{v_a(a,z)}-1\right) \leq \Lambda\,\left(\frac{u'(ra+z(a))}{u'(ra+z)}-1\right).
\end{align*}
Define 
$$
\bar a :=\inf\{a\geq\lba:\max_{z\in\cZ}u'(ra+z(a))/u'(ra+z)\leq 1+(\rho-r)/(2\Lambda)\}.
$$
By Assumption \ref{asm:utility-tail}, $\bar a <\infty$.
By the previous display, for all $a\ge \bar a $ savings must be strictly negative.

Let $\bar a(r)$ denote the upper wealth limit as a function of the interest rate. Local boundedness of $\bar a(r)$ on compact subsets of $(0,\rho)$ follows directly from the preceding construction. It remains to justify local boundedness at $r=0$.

By Assumption \ref{asm:utility-tail}, we may choose $R>0$ such that
$$
\max_{y,z\in\cZ}\frac{u'(x+y)}{u'(x+z)} \le 1+\frac{\rho}{4\Lambda}, \qquad x\ge R.
$$
For $r=0$, we know that $c^*(a,z;0)\to\infty$. Choose
$A_0>\max\{\lba,0\}$ so large that $c^*(A_0,z;0)>R+\bar z+1$ for all $z\in\cZ$.  By Lemma \ref{lem:joint-continuity-(a,r)} and the continuity of $(u')^{-1}$, there is
$r_0\in(0,\rho/2]$  with
$$
c^*(A_0,z;r)>R+\bar z,\qquad \forall \, z\in\cZ,\  |r|\le r_0 .
$$
Because $a\mapsto c^*(a,z;r)$ is increasing, it follows that, for $|r|\le r_0$ and $a\ge A_0$,
$$
s^*(a,z;r)<0
$$
whenever either $r\le0$, or $r>0$ and $ra\le R$. If instead $0<r\le r_0$ and $ra\ge R$, then
$$
\max_{y,z\in\cZ}\frac{u'(ra+y)}{u'(ra+z)}
\le 1+\frac{\rho}{4\Lambda}
\le 1+\frac{\rho-r}{2\Lambda},
$$
and the argument above rules out $s^*(a,z;r)\ge0$. Hence
$$
s^*(a,z;r)<0,\qquad \forall \, a\ge A_0,\ z\in\cZ,\ |r|\le r_0,
$$
so $\bar a(r)\le A_0$ in a neighborhood of $r=0$.\\

$(iv)$ Without the claim about local boundedness, this follows along the lines of \cite[Proposition 1]{achdou2022income} for the two-state case, see also \cite[Lemma 10]{shigeta-equilibrium}. Let $r\leq\rho$. First note that, for any $a>\lba$
$$
s^*_a(a,\lbz)=r-\frac{v^*_{aa}(a,\lbz)}{u''(c^*(a,\lbz))}\quad \Rightarrow\quad v^*_{aa}(a,\lbz) = u''(c^*(a,\lbz))(r-s^*_a(a,\lbz)).
$$
Therefore, \eqref{eq:Euler-interior} implies that
$$
(r-s^*_a(a,\lbz))s^*(a,\lbz) = \frac{(\rho-r)v^*_a(a,\lbz)-\sL v^*_a(a,\cdot)(\lbz)}{u''(c^*(a,\lbz))}
$$
converges as $a\downarrow\lba$ to a finite, strictly negative value by $(ii)$, say $-\fs(r)<0$. By Lemma \ref{lem:joint-continuity-(a,r)}, $\fs(r)$ is a continuous function. By l'H\^{o}spital's rule, for each $r$,
$$
\lim_{a\downarrow\lba} \frac{[s^*(a,\lbz)]^2}{a-\lba} = 2\lim_{a\downarrow\lba} s_a^*(a,\lbz) s^*(a,\lbz) = 2\mathfrak{s}(r)>0.
$$
Therefore, there exists $\varepsilon>0$ such that, for all $\lba<a<\lba+\varepsilon$,
$$
\frac{[s^*(a,\lbz)]^2}{a-\lba} \geq \mathfrak{s}(r)>0,\qquad \Longrightarrow\qquad s^*(a,\lbz)\leq - \sqrt{\mathfrak{s}(r)(a-\lba)}.
$$
Since $s^*(a,\lbz)<0$ on $(\lba,\infty)$, it reaches any neighborhood of $\lba$ in finite time. Consider two initial-value problems:
\begin{equation*}
\left\{ 
\begin{aligned}
\dot x(t) &= s^*(x(t),\lbz)\\
x(0)
&=\lba+\varepsilon
\end{aligned}
\right.
\qquad \text{and} \qquad  
\left\{ 
\begin{aligned}
\dot y(t) &= -\sqrt{\fs(r)(y(t)-\lba)}\\
y(0)
&=\lba+\varepsilon
\end{aligned}
\right.
\end{equation*}
Then,
$$
y(t) = \lba + \Big(\!\max\{\sqrt{\varepsilon}-\frac{\sqrt{\fs(r)}}{2 }t,0\}\!\Big)^2,\quad \text{so that}\quad y(t)=\lba \ \text{for}\ t\geq 2\sqrt{\varepsilon/\fs(r)}.
$$
By comparison, $\lba \leq x(t) \leq y(t)$ until $y(\cdot)$ hits $\lba$. Hence $x(\cdot)$ hits $\lba$ no later than $y(\cdot)$, which implies $\tau\leq 2\sqrt{\varepsilon/\fs(r)}<\infty$. This implies the claim.
\end{proof}

We now turn to the existence and uniqueness of an invariant measure. For $(t,a,z,B)\in[0,\infty)\times\sX\times\sB(\sX)$, define the transition kernels 
$$
P_t(a,z,B):=\P_{a,z}[(a^*_t,z_t)\in B]
$$
where
$$
\frac{\d}{\dt}a^*_t = r a^*_t + z_t - c^*(a_t^*,z_t),\qquad (a_0,z_0)=(a,z)
$$
denotes the optimally controlled joint Markov state process.  We let $\bar a$ be as in the above Lemma \ref{lem:properties-saving-rules} and set 
$$
\cS := [\lba,\bar a]\times \cZ.
$$
Clearly, since the drift of the optimal process $X^*=(a^*_t,z_t)_{t\ge0}$ is strictly negative outside of $\cS$, we can restrict our attention to $\cS$. Using the compactness of $\cS$ and the weak Feller property of $X^*$, it is easy to see that an invariant probability measure for $X^*$ exists.\footnote{In a related model, Bayer, Rendall and W\"alde \cite{bayer2019invariant} observe that the optimal process does not satisfy the \emph{strong} Feller property. This is also true for our process $X^*$.} We follow \cite{acikgoz2018existence, shigeta-equilibrium} who rely on the general theory of ergodic Markov processes to establish uniform ergodicity. The idea is to invoke \cite[Theorem 16.0.2]{meyn2012markov} for skeleton chains, which proves the equivalence of the state space $\cS$ being ``small'' and uniform ergodicity. Proving smallness involves lower-bounding the transition probabilities by a non-trivial measure. In light of Lemma \ref{lem:properties-saving-rules} \eqref{lem:properties-saving-rules-IV}, we may choose a Dirac mass at the borrowing constraint $(\lba,\lbz)$, and we refer to \cite[Proposition 6]{shigeta-equilibrium} and \cite{acikgoz2018existence} for details.

\begin{Prop}\label{prop:existence-uniqueness-G}
Assume $r<\rho$.
There exists a unique invariant measure $G^*$ of the Markov process $X^*=(a^*_t,z_t)$ on the state space $\cS$. Further, $X^*$ is exponentially ergodic, i.e.,
$$
\sup_{(a,z)\in\cS}\,\|P_t(a,z,\cdot)-G^*\|_{\mathrm{TV}}\leq C \kappa^t
$$
where $C<\infty$ and $\kappa<1$. Here, $\|\cdot\|_{\mathrm{TV}}$ denotes the total variation norm.
\end{Prop}

\subsection{Proof of Theorem \ref{thm:properties-G} \eqref{thm:properties-G-hand-to-mouth}}

\emph{Sufficiency.} For $(a,z)\in\sX$, define the function 
$$
w(a,z) := \alpha(z)(a-\lba)+\beta(z),\quad \text{where} \quad \alpha(z) := u'(r\lba+z)\quad \text{and}\quad \beta(z) := v^*(\lba,z).
$$
Now, $w$ is a viscosity supersolution on $(\lba,\infty)$ to \eqref{eq:hh-problem-HJB} if and only if
$$
F(a,z,w(a,\cdot),w_a(a,z)) 
= \rho w(a,z) - H(\alpha(z)) - \alpha(z)(r a+z) - \sL w(a,\cdot)(z)\geq0\qquad \forall a>\lba.
$$
By considering linear and constant terms separately, this is equivalent to
$$
(\rho-r)\alpha(z) - \sum_{y\neq z} \lambda(z,y)(\alpha(y)-\alpha(z))\geq0\qquad \text{and}\qquad F(\lba,z,\beta,\alpha(z))\geq0.
$$
By the boundary condition, $v^*_a(\lba,z)\geq u'(r\lba+z)=\alpha(z)$. Since $p\mapsto H(p)+py$ has a minimum in $p=u'(y)$, 
\begin{align*}
F(\lba,z,w(\lba,z),w_a(\lba,z)) &= \rho w(\lba,z) - H(\alpha(z))- \alpha(z)(r\lba +z) - \sL w(\lba,\cdot)(z)\\
&\geq \rho v^*(\lba,z) - H(v^*_a(\lba ,z))- v^*_a(\lba,z)(r\lba +z) - \sL v^*(\lba,\cdot)(z)\\
&=F(\lba,z,v^*(\lba,z),v^*_a(\lba,z))=0.
\end{align*}
Therefore, $w$ is a supersolution if and only if 
\be\label{eq:hand-to-mouth-proof}
r\leq \rho- \max_{z\in\cZ}\ \sum_{y\neq z} \lambda(z,y)\left(\frac{u'(r\lba+y)}{u'(r\lba+z)} -1\right).
\ee
In this case, by the comparison result in Theorem \ref{thm:comparison}, $w-v^*\geq 0$ and by construction, $w(\lba,z)-v^*(\lba,z)=0$. Since both functions admit right derivatives at $\lba$, we have $w_a(\lba,z)-v^*_a(\lba,z)\geq0$. Using the boundary constraint, this implies that $v_a(\lba,z)=u'(r\lba +z)$. By the characterization of optimal feedback controls $c^*=(u')^{-1}(v^*_a)$, we see that optimal savings satisfy $s^*(\lba,z;r)=0$. Therefore, the product measure $\delta_{\{\lba\}}(\d a) \otimes \nu(\dz)$, where $\nu\in\sP(\cZ)$ is the stationary distribution of the Markov chain $(z_t)_{t\geq0}$, is an invariant distribution of the optimal state process. By uniqueness of the stationary measure (see Proposition \ref{prop:existence-uniqueness-G}), it is the unique invariant distribution.

\vspace{1em}
\emph{Necessity.} Now assume that the stationary distribution $G^*$  is supported on $\{\lba\}\times\cZ$. Let $(\pi_z)_{z\in\cZ}$ denote the (unique) stationary distribution of the income process $(z_t)$. Now $s^*(\lba,z)\ge0$ due to the boundary constraint, and by stationarity of the optimal wealth process,
$$
\sum_{z\in\cZ} \pi_z s^*(\lba,z)= \int s^*(a,z)\, G^*(\da,\dz)=0\quad \text{so that}\quad s^*(\lba,z)=0\quad \forall z\in\cZ.
$$
This implies $v^*_a(\lba,z)=\alpha(z)$ as before. Using the HJB equations at $\lba$ and $\lba+h$, we obtain
\begin{align*}
\rho \frac{v^*(\lba+h,z)-v^*(\lba,z)}{h} &= \frac{1}{h}\Big(H(v^*_a(\lba+h,z)) +v_a^*(\lba+h,z)(r\lba +z)  - H(v_a^*(\lba,z)) - v^*_a(\lba,z)(r\lba+z))\Big) \\
&\quad + r v_a^*(\lba+h,z) + \sL \left(\frac{v^*(\lba+h,\cdot)-v^*(\lba,\cdot)}{h}\right)(z).
\end{align*}
Now 
$$
H(v^*_a(\lba+h,z)) + v_a^*(\lba+h,z)(r\lba+z) - H(\alpha(z))-\alpha(z) (r\lba+z)\geq0,
$$
so that taking $h\downarrow0$ shows
$$
(\rho-r)\alpha(z)=(\rho-r) v_a^*(\lba,z)\geq \sL v^*_a(\lba,\cdot)(z)=\sL\alpha(z).
$$
Dividing by $\alpha(z)$ yields the claim.
\hfill $\Box$

\subsection{Proof of Theorem \ref{thm:properties-G} \eqref{thm:properties-G-linearity-of-A-in-tau} }

We first observe the following scaling property of the optimal feedback controls inherited from the CRRA utility function. Let $v^*(a,z;w,\tau)$ and $c^*(a,z;w,\tau)$ denote the value function and corresponding optimal feedback consumption, respectively, for the parameters $(w,\tau)$.

\begin{Lem}\label{lem:CRRA-scaling}
In addition to the standing assumptions of this section, \ref{asm:standing}--\ref{asm:utility-function}, assume the no-borrowing constraint holds and assume CRRA utility.
For any $\tau,\tau'<1$, $w,w'>0$ and $(a,z)\in\sX$, 
\begin{equation*}
\alpha \, c^*(a,z;w',\tau') = \alpha'\, c^*\left(\frac{\alpha}{\alpha'}a,z;w,\tau \right),
\end{equation*}
where $\alpha:=w(1-\tau)>0$ and $\alpha':=w'(1-\tau')>0$.
\end{Lem}

\begin{proof}
First, observe that under CRRA utility,
$$
H(p) = \sup_{c>0} \Big(\frac{c^{1-\gamma}}{1-\gamma} - pc\Big)  = \frac{\gamma}{1-\gamma} \ p^{\frac{\gamma-1}{\gamma}},\qquad p>0.
$$
In the case of log-utility ($\gamma=1$), $H(p)=-\log(p)-1$.
We claim that 
\begin{equation*}
V(a,z) := 
\begin{cases}
(\alpha'/\alpha)^{1-\gamma} v^*\left(\alpha a/\alpha',z \right), & \gamma\neq 1,\\
v^*\left(\alpha a/\alpha',z \right) + \rho^{-1} \log(\alpha'/\alpha), & \gamma=1,
\end{cases}
\end{equation*}
is a classical solution to the HJB equation \eqref{eq:HJB} with parameters $(w',\tau')$. Indeed, let $(a,z)\in\sX$ be given. We present the proof for $\gamma\neq 1$ as the logarithmic case is analogous. We compute
\begin{equation*}
V_a(a,z) =(\alpha'/\alpha)^{-\gamma}  v^*_a(\alpha a/\alpha',z) \quad \Longrightarrow \quad H(V_a(a,z)) = (\alpha'/\alpha)^{1-\gamma} H(v^*_a(\alpha a/\alpha' ,z)).
\end{equation*}
Since $v^*$ solves \eqref{eq:HJB},
\begin{equation*}
\rho v^*(\alpha a/\alpha',z) = H(v^*_a(\alpha a/\alpha',z)) + v^*_a(\alpha a/\alpha',z)[r\alpha a/\alpha'+\alpha z] + \sL v^*(\alpha a/\alpha',\cdot)(z)
\end{equation*}
Multiplying by $(\alpha'/\alpha)^{1-\gamma}$ yields
\begin{align*}
\rho V(a,z) &= H(V_a(a,z)) + (\alpha'/\alpha)V_a(a,z)[r\alpha a/\alpha'+\alpha z] + \sL V(a,\cdot)(z)\\
&= H(V_a(a,z)) + V_a(a,z)[ra+\alpha' z] + \sL V(a,\cdot)(z),
\end{align*}
as claimed. Using the feedback characterization of optimal consumption rules, we obtain
\begin{align*}
c^*(a,z;w',\tau') &= (V_a(a,z;w',\tau'))^{-\frac{1}{\gamma}} = \frac{\alpha'}{\alpha} (v^*_a\left( \alpha a/\alpha', z;w,\tau\right))^{-\frac{1}{\gamma}}= \frac{\alpha'}{\alpha} \, c^*(\alpha a/\alpha',z;w,\tau).\qedhere
\end{align*}
\end{proof}

\begin{proof}[Proof of Theorem \ref{thm:properties-G} \eqref{thm:properties-G-linearity-of-A-in-tau}]
Let $G$ and $G'$ denote the unique stationary measures corresponding to $(w,\tau)$ and $(w',\tau')$, respectively. Let $(a'_t)_{t\ge0}$ denote an optimal wealth process, controlled under the parameters $(w',\tau')$. Using the notation of Lemma \ref{lem:CRRA-scaling},
$$
\d \Big(\frac{\alpha}{\alpha'}a'_t\Big)  = \Big[ r\frac{\alpha}{\alpha'}a'_t+ \alpha z_t -  c^*\left(\frac{\alpha}{\alpha'}a'_t,z_t;w,\tau \right) \!\Big] \,\d t,
$$
Hence, $b_t:=\alpha a'_t/\alpha'$ satisfies $\d b_t  = [ rb_t+ \alpha z_t -  c^*\left(b_t,z_t;w,\tau \right)] \,\d t.$
By Proposition \ref{prop:existence-uniqueness-G}, $G^*(\da,\dz;w',\tau')= G^*(\alpha/\alpha'\,\da,\dz;w,\tau)$. In particular, 
$$
A(w',\tau')=\int_{\sX} a\,G^*(\da,\dz; w',\tau')= \frac{\alpha'}{\alpha} \int_{\sX} a \, G^*(\da,\dz;w,\tau)= \frac{\alpha'}{\alpha}A(w,\tau).
$$
Finally, by stationarity, $C(w',\tau')=rA(w',\tau')+\alpha'=\alpha'/\alpha\,(rA(w,\tau)+\alpha)=\alpha'/\alpha\, C(w,\tau)$.
\end{proof}

\subsection{Proofs of Theorem \ref{thm:properties-G} \eqref{thm:properties-G-continuity} and \eqref{thm:properties-G-explosion-at-rho}}

We continue to establish continuity and the limiting behavior of aggregate savings $(-\infty,\rho)\ni r\mapsto A(r)$ as a function of the interest rate.

\begin{proof}[Proof of Theorem \ref{thm:properties-G} 
\eqref{thm:properties-G-continuity}]
We suppress the dependence on the fixed parameters $w,\tau$. Let $r_n\to r$ in
$R(w,\tau)$, and write
$$
G_n:=G^*(\da,\dz;r_n),\qquad G:=G^*(\da,\dz;r).
$$
Since the weak convergence topology on the space of Borel probability measures is first countable, it is enough to prove that $G_n\Rightarrow G$. By Lemma \ref{lem:properties-saving-rules}
\eqref{lem:properties-saving-rules-III}, the upper support bound is locally
bounded in the interest rate. Hence, there exists $M<\infty$ such that both $\mathrm{supp}\, G_n$ and $\mathrm{supp}\, G$ are contained in $[\lba,M]\times\cZ$, for all $n\ge1$. Hence, the sequence $(G_n)$ is tight. Let $G_{n_k}\Rightarrow \widehat G$
be an arbitrary weakly convergent subsequence. We show that $\widehat G$
is invariant for the optimally controlled process at interest rate $r$.

Let  $s_n(a,z)$ (resp.\ $s(a,z)$) be the optimal savings rule associated with $r_n$ (resp.\ $r$). By Lemma \ref{lem:joint-continuity-(a,r)} and the feedback characterization, we have $c^*(\cdot,\cdot;r_n,w,\tau)\to c^*(\cdot,\cdot;r,w,\tau)$ uniformly on $[\lba,M]\times\cZ$. Therefore, $s_n\to s$ uniformly on $[\lba,M]\times\cZ$.

Let $\varphi: [\lba,M]\times\cZ\to\R$ be smooth. By stationarity, Dynkin's formula
gives
$$
0=
\int
\left[
s_n(a,z)\partial_a\varphi(a,z)
+
\sL\varphi(a,\cdot)(z)
\right]G_n(\da,\dz).
$$
Passing to the subsequential limit, using weak convergence of $G_{n_k}$,
uniform convergence $s_{n_k}\to s$, and finiteness of $\cZ$, gives
$$
0=
\int
\left[
s(a,z)\partial_a\varphi(a,z)
+
\sL\varphi(a,\cdot)(z)
\right]\widehat G(\da,\dz).
$$
Thus $\widehat G$ is an invariant distribution
for the optimally controlled process at interest rate $r$. By Proposition \ref{prop:existence-uniqueness-G}, this invariant distribution is unique. Hence, the full sequence converges, $G_n\Rightarrow G$.

Finally, because all the measures are supported on the common compact set
$[\lba,M]\times\cZ$,
\begin{equation*}
A(r_n,w,\tau)
=
\int a\,G^*(\da,\dz;r_n,w,\tau)
\longrightarrow
\int a\,G^*(\da,\dz;r,w,\tau)
=
A(r,w,\tau).\qedhere
\end{equation*}
\end{proof}

We conclude this section with a proof of Theorem \ref{thm:properties-G} \eqref{thm:properties-G-explosion-at-rho} and specialize to the case of a no-borrowing limit for households $\lba=0$. The proof is inspired by the discrete-time approach of A\c{c}\i kg\"{o}z \cite{acikgoz2018existence}.

\begin{Lem}\label{lem:atom}
For any $L>0$ there exists $T>0$ and $\underline{p}\in(0,1]$ such that 
$$
\inf_{r\in[0,\rho]} \ \inf_{(a,z)\in[0,L]\times\cZ}\ P_T(a,z,\{(0,\lbz)\};r)\geq \underline{p}
$$
\end{Lem}

\begin{proof}
As before, let $\lbz:=\min \cZ$, $\overline z:=\max \cZ$, and set $\Lambda:=\sum_{y\neq \lbz}\lambda(\lbz,y)$. Fix $L>0$.

First, note that, starting from any $(a,z)\in[0,L]\times \cZ$, the (optimal) asset process $(a_t)$ satisfies 
\be\label{eq:atom-proof-I}
a_1\le L'
:=
e^\rho L+\overline z\,\frac{e^\rho-1}{\rho}.
\ee
For each $r\in[0,\rho]$ and $b\in[0,L']$, let $\tau^r(b)$ be the hitting time of $0$ for the deterministic ODE
$$
\dot x(t)=s^*(x(t),\lbz;r),\qquad x(0)=b.
$$
By Lemma \ref{lem:properties-saving-rules} \eqref{lem:properties-saving-rules-I}, $s^*(a,\lbz;r)<0$ for $a>0$ and $s^*(0,\lbz;r)=0$. By Lemma \ref{lem:properties-saving-rules} \eqref{lem:properties-saving-rules-IV}, the hitting time $\tau^r(L')$ is locally bounded as a function of $r$. Since $[0,\rho]$ is compact, there exists $T_0<\infty$ such that
$$
\tau^r(L')\le T_0
\qquad\text{for every }r\in[0,\rho].
$$
By comparison, $\tau^r(b)\le\tau^r(L')$ for every $b\in[0,L']$. Therefore, if the income process starts from $\lbz$ and has no jump during $[0,T_0]$, the asset process reaches $0$ by time $T_0$ and then stays at $0$. Hence,
\be \label{eq:atom-proof-II}~
P_{T_0}^r(b,\lbz,\{(0,\lbz)\})
\ge
\exp(-\Lambda T_0)
\qquad
\forall b\in[0,L'],\, r\in[0,\rho].
\ee

Next, let $ p_0:=\min_{z\in \cZ}\P_z(z_1=\lbz)>0$ and set
$$
T:=T_0+1,
\qquad
\underline{p}:=p_0\exp(-\Lambda T_0)>0.
$$
Using the Markov property at time $1$, together with \eqref{eq:atom-proof-I} and \eqref{eq:atom-proof-II}, we obtain, for every
$(a,z)\in[0,L]\times \cZ$ and every $r\in[0,\rho]$,
$$
\begin{aligned}
P_T^r(a,z,\{(0,\lbz)\})
&= \E_{a,z} \left[
\P((a_{1+T_0},z_{1+T_0})= (0,\lbz)\,|\,\sF_1)
\right] \\
&= \E_{a,z} \left[
P_{T_0}^r(a_1,z_1,\{(0,\lbz)\})
\right]\\
&\ge \E_{a,z} \left[
\mathbf{1}_{\{z_1=\lbz\}}
P_{T_0}^r(a_1,\lbz,\{(0,\lbz)\})
\right] \\
&\ge \P_z(z_1=\lbz)\,
\exp(-\Lambda T_0)
\\
&\ge \underline{p}.
\end{aligned}
$$
Taking the infimum over $r\in[0,\rho]$ and $(a,z)\in[0,L]\times \cZ$ proves the claim.
\end{proof}

We next record a short lemma that justifies the application of Dynkin's formula to $v^*_a$.

\begin{Lem}
\label{lem:dynkin-v_a}
Fix $r<\rho$. With the convention of Lemma \ref{lem:Euler} for the product $v^*_{aa}s^*=0$ whenever $s^*=0$, we have 
$$
0=
\int_\sX
\left[
v^*_{aa}(a,z)s^*(a,z)
+
\sL v^*_a(a,\cdot)(z)
\right]
G^*(\da,\dz).
$$
\end{Lem}
\begin{proof}
Write $p(a,z):=v_a^*(a,z)$ and $s(a,z):=ra+z-c^*(a,z)$, and let $\Phi_t^{z}(a)$ be the deterministic flow
$$
\frac{\d}{\dt}\Phi_t^{z}(a)=s(\Phi_t^{z}(a),z),
\qquad
\Phi_0^{z}(a)=a.
$$
We first note that the following limit exists:
$$
D_{s}p(a,z)
:=
\lim_{t\downarrow0}
\frac{
p(\Phi_t^{z}(a),z)-p(a,z)
}{t}.
$$
Indeed, by Proposition \ref{prop:cts-differentiability} and concavity of the value
function, $p(\cdot,z)$ is continuous and non-increasing for every
$z\in\cZ$. By Lemma \ref{lem:Euler}, $p(\cdot,z)$ is $C^1$ on every
open set on which $s(\cdot,z)\neq0$, and there
$$
p_a(a,z)s(a,z)
=
(\rho-r)p(a,z)-\sL p(a,\cdot)(z).
$$
If $s(a,z)=0$, then the deterministic flow with income state fixed at
$z$ is constant at $a$, and hence $D_{s}p(a,z)=0$. At the
boundary, the same conclusion holds if $s(\underline a,z)=0$, while if $s(\underline a,z)>0$, Lemma \ref{lem:Euler} gives the right derivative. Thus $D_{s}p$ is well-defined.

Let $(a_t,z_t)$ be the optimally controlled process. Between jump
times of $z_t$, the wealth component follows the deterministic flow above.
Hence, between jumps,
$$
\frac{\d}{\dt}p(a_t,z_t)
=
D_{s}p(a_t,z_t).
$$
By classical arguments, see \cite[Appendix B]{FS},
$$
M_t
:=
p(a_t,z_t)-p(a_0,z_0)
-
\int_0^t
\left[
D_{s}p(a_s,z_s)
+
\sL p(a_s,\cdot)(z_s)
\right]\d s
$$
is a martingale. Taking expectations under the invariant distribution implies the claim.
\end{proof}

\begin{Lem}\label{lem:explosion-at-rho}
For any $L>0$, there exists $C_L<\infty$ such that $G^*([0,L]\times\cZ;r)\leq C_L(\rho-r)$ as $r\uparrow \rho$.
\end{Lem}

\begin{proof}
Fix $L>0$ and set
$$
p^r(a,z):=v_a^*(a,z;r),\qquad
s^r(a,z):=ra+z-c^*(a,z;r),
\qquad
G^r:=G^*(\cdot\,;r).
$$
By the Euler equation, Lemma \ref{lem:Euler},
$$
p_a^r(a,z)s^r(a,z)+ \sL p^r(a,\cdot)(z)
=
(\rho-r)p^r(a,z),
\qquad a>0,
$$
where the product $p_a^r s^r$ is interpreted as $0$ on $\{s^r=0\}$. By Lemma \ref{lem:dynkin-v_a} and the boundary inequality \eqref{eq:Euler-bdry-zero},
\begin{align*}
0&=\int_\sX  \big( p_a^r(a,z)s^r(a,z)+ \sL p^r(a,\cdot)(z) \big)\,G^r(\da,\dz)\\
&\le
(\rho-r)\int_\sX p^r(a,z)\,G^r(\da,\dz)
+
G^r(\{(0,\lbz)\})\,\sL p^r(0,\cdot)(\lbz).
\end{align*}
By Lemma \ref{lem:properties-saving-rules} \eqref{lem:properties-saving-rules-II}, $\sL p^r(0,\cdot)(\lbz)<0$. Hence
$$
G^r(\{(0,\lbz)\})
\le
\frac{(\rho-r)\int_\sX p^r(a,z)\,G^r(\da,\dz)}
{-\sL p^r(0,\cdot)(\lbz)}.
$$
Next, by Lemma \ref{lem:atom}, there exist $T>0$ and $\underline{p}>0$, independent of $r\in(0,\rho)$, such that
$$
\inf_{(a,z)\in[0,L]\times \cZ}
P_T^r\bigl(a,z,\{(0,\lbz)\}\bigr)
\ge \underline{p} .
$$
Using invariance of $G^r$,
$$
G^r(\{(0,\lbz)\})
=
\int_\sX
P_T^r\bigl(a,z,\{(0,\lbz)\}\bigr)\,
G^r(\da,\dz)  \ge
\underline{p} \,G^r([0,L]\times \cZ).
$$
Hence,
$$
G^r([0,L]\times \cZ)
\le
\frac{(\rho-r)\int_\sX p^r(a,z)\,G^r(\da,\dz)}
{\underline{p}\, (-\sL p^r(0,\cdot)(\lbz))}.
$$

It remains to pass to the limit. Since $p^r(\cdot,z)$ is decreasing in $a$, $\sup_\sX p^r
=
\max_{z\in \cZ} p^r(0,z)$.
By Lemma \ref{lem:boundedness-v-v_a}, this quantity remains bounded as $r\uparrow\rho$. Moreover, by Lemma \ref{lem:joint-continuity-(a,r)} and Lemma \ref{lem:properties-saving-rules} \eqref{lem:properties-saving-rules-II} at $r=\rho$,
$$
\sL p^r(0,\cdot)(\lbz)
\longrightarrow
\sL p^\rho(0,\cdot)(\lbz)<0,\qquad \text{as}\quad r\uparrow\rho.
$$
Hence, there exist constants $M<\infty$ and $\eta>0$ such that, for all $r$ close to $\rho$,
$$
\int_\sX p^r\,\d G^r\le M,
\quad \text{and}\quad 
-\sL p^r(0,\cdot)(\lbz)\ge \eta,\qquad \text{so that}\qquad G^r([0,L]\times \cZ)
\le
\frac{M}{\underline{p} \eta}(\rho-r). 
$$
Setting $C_L:=M/(\underline{p} \eta)$ proves the claim.
\end{proof}

\begin{proof}[Proof of Theorem \ref{thm:properties-G} \eqref{thm:properties-G-explosion-at-rho}]
We estimate
$$
L\left(1-G^*([0,L]\times\cZ;r)\right)=\int_{(L,\infty)\times\cZ}L \ G^*(\d a,\d z;r) \leq  \int_{(L,\infty)\times\cZ}a \ G^*(\d a,\d z;r)\leq A(r).
$$
The preceding Lemma \ref{lem:explosion-at-rho} shows that $ \liminf_{r\uparrow\rho} \ A(r)\geq L$ for any $L>0$. Hence, $\lim_{r\uparrow\rho}A(r)=\infty$, as claimed.
\end{proof}

\subsection{Proof of Theorem \ref{thm:properties-G} \eqref{thm:properties-A-increasing-in-r}} \label{ssec:A-increasing}

This section establishes the monotonicity of $r\mapsto A(r):=A(r,1,0)$ under the no-borrowing limit and CRRA utility with $\gamma\leq 1$.
By Theorem \ref{thm:properties-G} \eqref{thm:properties-G-linearity-of-A-in-tau}, it is enough to prove the claim for $w=1$ and $\tau=0$. 

Fix $r<\rho$ and set
$$
p_z^r(a):=v_a^*(a,z;r),\qquad
c_z^r(a):=(p_z^r(a))^{-1/\gamma},\qquad
s_z^r(a):=ra+z-c_z^r(a),\qquad (a,z)\in\sX.
$$
The Euler equation may be written as
\begin{equation}\label{eq:marginal-equation-r}
\cE^r_z(a,p,\partial_a p)=0,\qquad (a,z)\in\sX,
\end{equation}
where, for $(a,z)\in\sX$ and $p,p':\sX\mapsto\R$,
$$
\cE^r_z(a,p,p')
:=
(\rho-r)p_z(a)
-
\bigl(ra+z-p_z(a)^{-1/\gamma}\bigr)p'_z
-
\sL p(a,\cdot)(z).
$$
We interpret \eqref{eq:marginal-equation-r} in the viscosity sense on the open domain $\cO\times\cZ$, and note that $p^r$ is a viscosity solution of this equation by Lemma \ref{lem:Euler}. We first record a comparison principle.

\begin{Lem}\label{lem:Euler-comparison}
Let $p,q:\sX\to(0,\infty)$ be continuous functions such that $p_z(a),q_z(a)\to0$ as $a\to\infty$, for every $z\in\cZ$. Assume that, $p_z$ is decreasing on $\cO$, and that $p_z(0)\le q_z(0)$, $z\in\cZ$. If $p$ is a viscosity subsolution and $q$ is a viscosity supersolution of \eqref{eq:marginal-equation-r}, then $p\le q$ on $\sX$.
\end{Lem}

\begin{proof}
Towards a contradiction, suppose that $M:=\sup_{\sX}(p-q)>0$.
Since $p_z(a)-q_z(a)\to0$ as $a\to\infty$ and $p_z(0)\le q_z(0)$, the supremum
is attained at some point in $\cO\times\cZ$. 

For $k\ge1$, define
$$
\Phi_k(x,y,z):=p_z(x)-q_z(y)-\frac{k}{2}(x-y)^2,
\qquad x,y\ge0,\ z\in\cZ,
$$
and let $(x_k,y_k,z_k)$ be a maximizer. The standard doubling-of-variables reasoning implies, along a subsequence,
$$
x_k,y_k\to a_*,\qquad
p_{z_k}(x_k)-q_{z_k}(y_k)\to M,\qquad
k(x_k-y_k)^2\to0,\qquad \text{as}\ k\to\infty,
$$
for some $a_*>0$. In particular, $x_k,y_k>0$ for all large $k$. By maximality of $(x_k,y_k,z_k)$, the function
$$
\phi_k(x):=
p_{z_k}(x_k)
+\frac{k}{2}(x-y_k)^2
-\frac{k}{2}(x_k-y_k)^2
$$
touches $p_{z_k}$ from above at $x_k$. As $p_{z_k}$ is decreasing,
any $C^1$ upper test function at an interior point has non-positive derivative. Hence, $d_k:=k(x_k-y_k)=\phi_k'(x_k)\le0$.

The viscosity inequalities give
$$
\cE_{z_k}^r(x_k,p,d_k)\le0,
\qquad
\cE_{z_k}^r(y_k,q,d_k)\ge0.
$$
Subtracting yields
$$
(\rho-r)\bigl(p_{z_k}(x_k)-q_{z_k}(y_k)\bigr)
\le
d_k\Bigl[ r(x_k-y_k) -\bigl(p_{z_k}(x_k)^{-1/\gamma} -q_{z_k}(y_k)^{-1/\gamma}\bigr)
\Bigr]
+\sL\bigl(p(x_k,\cdot)-q(y_k,\cdot)\bigr)(z_k).
$$
The last term is non-positive because $(x_k,y_k,z_k)$ maximizes $\Phi_k$.
Moreover, for large $k$,
$$
p_{z_k}(x_k)>q_{z_k}(y_k),\qquad \Longrightarrow \qquad 
p_{z_k}(x_k)^{-1/\gamma}
-q_{z_k}(y_k)^{-1/\gamma}<0.
$$
Since $d_k\le0$, the term
$$
-d_k\bigl(p_{z_k}(x_k)^{-1/\gamma}
-q_{z_k}(y_k)^{-1/\gamma}\bigr)
$$
is non-positive, while
$$
d_k r(x_k-y_k)=r k(x_k-y_k)^2\to0.
$$
Taking the limsup gives $(\rho-r)M\le0$, contradicting $r<\rho$ and $M>0$. Hence $p\le q$.
\end{proof}

\begin{Lem}
For every $z\in\cZ$, the map
\begin{equation}\label{eq:ap-monotone}
[0,\infty)\ni a\longmapsto a\,p_z^r(a)
\quad\text{is increasing}.
\end{equation}
\end{Lem}

\begin{proof}
Fix $\lambda>1$ and define $\widetilde p_z(a):=\lambda p_z^r(\lambda a)$. We claim that $\widetilde p$ is a viscosity supersolution of
\eqref{eq:marginal-equation-r}. Let $\phi\in C^1$ be an admissible lower test function of $\widetilde p_z$ at $a>0$, and
set
$$
b:=\lambda a,\qquad c:=c_z^r(b)=p_z^r(b)^{-1/\gamma}.
$$
Then $\psi(x):=\lambda^{-1}\phi(x/\lambda)$ is an admissible test function for $p_z^r$ at $b$. Since $p^r$ is a supersolution,
$$
(\rho-r)p_z^r(b)
-
(rb+z-c)\psi'(b)
-
\sL p^r(b,\cdot)(z)
\ge0.
$$
Using $\psi'(b)=\lambda^{-2}\phi'(a)$ and $\widetilde p_z(a)^{-1/\gamma} =\lambda^{-1/\gamma}c,$ we compute
$$
\cE_z^r(a, \widetilde p,\phi')
=
\lambda[(\rho-r)p_z^r(b) - (rb+z-c)\psi'(b) - \sL p^r(b,\cdot)(z)]  +
\Big( \frac{rb+z-c}{\lambda}
-
(ra+z-\lambda^{-1/\gamma}c)
\Big)\phi'(a).
$$
The first term is non-negative by the supersolution property of $p^r$. The bracket in the second term equals
$$
\Big(\frac1\lambda-1\Big)z
+
\Big(\lambda^{-1/\gamma}-\frac1\lambda\Big)c.
$$
Because $\gamma\le1$, one has $\lambda^{-1/\gamma}\le\lambda^{-1}$, so this bracket is
non-positive. Since $\widetilde p_z$ is decreasing, the test function satisfies
$\phi'(a)\le0$, so that the second term is also non-negative, proving that $\widetilde p$ is indeed a supersolution.

On the other hand, $p^r$ is a subsolution. Moreover,
$$
p_z^r(0)<\lambda p_z^r(0)=\widetilde p_z(0),
\qquad
p_z^r(a),\widetilde p_z(a)\to0
\quad\text{as }a\to\infty.
$$
By Lemma \ref{lem:Euler-comparison}, for every $(a,z)\in\sX$, $p_z^r(a)\le \lambda p_z^r(\lambda a)$. With $b=\lambda a$, this writes $a\,p_z^r(a)\le b\,p_z^r(b)$, $0<a<b$, 
establishing \eqref{eq:ap-monotone}. 
\end{proof}

In particular, if $\phi$ is an admissible upper test function of $p_z^r$ at an interior point $a>0$ with $\phi(a)=p_z^r(a)$, then
\begin{equation}\label{eq:upper-test-elasticity}
p_z^r(a)+a\phi'(a)\ge0.
\end{equation}
Indeed, by \eqref{eq:ap-monotone},
$$
0
\le
\underset{h\downarrow0}{\limsup}
\frac{(a+h)p_z^r(a+h)-a p_z^r(a)}{h}\leq \lim_{h\downarrow0} \frac{(a+h)\phi(a+h)-a\phi(a)}{h}=
\phi(a)+a\phi'(a)=p^r_z(a)+a\phi'(a).
$$

\begin{Lem}
Fix $r_1<r_2<\rho$. Then,
\begin{equation}\label{eq:p-monotone-in-r}
p_z^{r_1}(a)\le p_z^{r_2}(a),
\qquad (a,z)\in\sX.
\end{equation}
\end{Lem}

\begin{proof}
Write
$$
p_{i,z}(a):=p_z^{r_i}(a),\qquad
c_{i,z}(a):=p_{i,z}(a)^{-1/\gamma},\qquad
s_{i,z}(a):=r_i a+z-c_{i,z}(a),
\qquad i=1,2.
$$

We first observe that $p_1$ is a viscosity subsolution of the marginal
equation with interest rate $r_2$. Indeed, let $\phi$ be an admissible
upper test function of $p_{1,z}$ at an interior point $a>0$. Then
$$
\cE_z^{r_2}(a,p_1,\phi'(a)) =
\cE_z^{r_1}(a,p_1,\phi'(a))
- (r_2-r_1)\bigl(p_{1,z}(a)+a\phi'(a)\bigr)  \le 0.
$$
Here, the first term is nonpositive because $p_1$ is a subsolution at rate
$r_1$, while the second term is nonpositive by
\eqref{eq:upper-test-elasticity}. Towards a contradiction, suppose that $M:=\sup_{\sX}(p_1-p_2)>0$.  Since $p_{i,z}(a)\to0$ as $a\to\infty$, the supremum is attained. If a maximizer were located at an interior point $a>0$, the proof of Lemma \ref{lem:Euler-comparison}, applied to the subsolution $p_1$ and supersolution $p_2$ at rate $r_2$, would imply $(\rho-r_2)M\le0$, which is impossible. Hence $M$ is attained at the boundary: there exists
$z_0\in\cZ$ such that
$$
M=p_{1,z_0}(0)-p_{2,z_0}(0)>0.
$$
Using Lemma
\ref{lem:properties-saving-rules}, we see $z_0\neq \lbz$. Set
$$
\sigma_i:=s_{i,z_0}(0)=z_0-c_{i,z_0}(0),\qquad i=1,2.
$$
Since $p_{1,z_0}(0)>p_{2,z_0}(0)$ and $p\mapsto p^{-1/\gamma}$ is
decreasing, $c_{1,z_0}(0)<c_{2,z_0}(0)$, and hence
$$
\sigma_1>\sigma_2\ge0.
$$
Using $\sigma_1>0$, by \eqref{eq:Euler-brdy-nonzero},
$$
(\rho-r_1)p_{1,z_0}(0)
=
\sigma_1\,\partial_ap_{1,z_0}(0)
+
\sL p_1(0,\cdot)(z_0),
$$
and, hence,
\begin{equation}\label{eq:p-monotone-boundary-p1}
(\rho-r_2)p_{1,z_0}(0)
-
\sigma_1\,\partial_ap_{1,z_0}(0)
-
\sL p_1(0,\cdot)(z_0)
=
-(r_2-r_1)p_{1,z_0}(0)<0.
\end{equation}
For $p_2$, Lemma \ref{lem:Euler} gives
\begin{equation}\label{eq:p-monotone-boundary-p2}
(\rho-r_2)p_{2,z_0}(0)
-
\sigma_2\,\partial_ap_{2,z_0}(0)
-
\sL p_2(0,\cdot)(z_0)
\ge0,
\end{equation}
where, if $\sigma_2=0$, the derivative term is omitted. If $\sigma_2>0$, subtracting \eqref{eq:p-monotone-boundary-p2} from
\eqref{eq:p-monotone-boundary-p1} yields
$$
\begin{aligned}
(\rho-r_2)M
&<
\sigma_1\partial_a p_{1,z_0}(0)
-
\sigma_2\partial_a p_{2,z_0}(0)
+
\sL(p_1-p_2)(0,\cdot)(z_0)\\
&=
\sigma_1\partial_a (p_1-p_2)_{z_0}(0)
+
(\sigma_1-\sigma_2)\partial_a p_{2,z_0}(0)
+
\sL(p_1-p_2)(0,\cdot)(z_0)
\le0.
\end{aligned}
$$
The last inequality uses that $(0,z_0)$ is a global maximizer of
$p_1-p_2$ and that $p_{2,z_0}$ is decreasing. This contradicts $M>0$.

If $\sigma_2=0$, the same reasoning yields
$$
(\rho-r_2)M < \sigma_1\partial_a p_{1,z_0}(0) +
\sL(p_1-p_2)(0,\cdot)(z_0)
\le0,
$$
since $p_{1,z_0}$ is decreasing and $(0,z_0)$ maximizes
$p_1-p_2$. This is again impossible. This establishes \eqref{eq:p-monotone-in-r}.
\end{proof}

\noindent
Since $p\mapsto p^{-1/\gamma}$ is decreasing, \eqref{eq:p-monotone-in-r} implies $c_z^{r_2}(a)\le c_z^{r_1}(a)$, for $(a,z)\in\sX$. Consequently,
\begin{equation}\label{eq:savings-monotone-in-r}
s_z^{r_2}(a)-s_z^{r_1}(a)
=
(r_2-r_1)a+c_z^{r_1}(a)-c_z^{r_2}(a)
\ge0,
\end{equation}
which implies that
$$
G^*(\cdot\,;r_2,1,0)
\quad\text{first-order stochastically dominates}\quad
G^*(\cdot\,;r_1,1,0).
$$
Therefore,
$$
A(r_2,1,0)
=
\int a\,G^*(\da,\dz;r_2,1,0)
\ge
\int a\,G^*(\da,\dz;r_1,1,0)
=
A(r_1,1,0).
$$
This proves that $r\mapsto A(r,1,0)$ is increasing. Multiplying by the positive factor
$w(1-\tau)$ from Theorem \ref{thm:properties-G} \eqref{thm:properties-G-linearity-of-A-in-tau} proves the claim for general
$(w,\tau)$.

\subsection{Proof of Theorem \ref{thm:A-can-decrease-gamma-greater-one}}\label{ssec:A-fail-increasing}

Fix $\gamma>1$. We give a constructive proof that it is possible for  $r\mapsto A(r):=A(r,1,0)$ not to be increasing. Define 
$$
\kappa(r):=\frac{\rho+(\gamma-1)r}{\gamma},\quad 
I_\gamma:=(-\rho/(\gamma-1),\rho)\subset\R,\quad \text{and}\quad L_\gamma(r):=\frac{1}{\kappa(r)(\rho-r)^{1/\gamma}} \quad \text{for}\ r\in I_\gamma.
$$
We construct a sequence of two-state income processes. Let $\eta_n\downarrow0$ be arbitrary and fix $q\in(0,\gamma)$. We define transition rates  
$$
n\text{th rate from} \ z_h \to z_\ell := \lambda_n := \eta_n^\gamma,\qquad n\text{th rate from} \ z_\ell \to z_h := \mu_n:=\eta_n^q,
$$
so that the stationary probabilities of the income process are given by
$$
\pi_\ell^n :=\frac{\lambda_n}{\lambda_n+\mu_n},
\qquad
\pi_h^n :=\frac{\mu_n}{\lambda_n+\mu_n}.
$$
Further, choose
$$
z_\ell^n=o(\eta_n),\qquad z_h^n :=\frac{1-\pi_\ell^n z_\ell^n}{\pi_h^n},\qquad \cZ_n:=\{z_\ell^n,z_h^n\}.
$$
Then, the mean income is one, $\pi_\ell^n z_\ell^n+\pi_h^n z_h^n=1$, and as $\lambda_n/\mu_n=\eta_n^{\gamma-q}\to0$, we have, as $n\to\infty$,
$$
\pi_\ell^n\to0,\qquad \pi_h^n\to1,\qquad z_\ell^n\to0,\qquad z_h^n\to1.
$$
Let $A_n(r)$ denote the aggregate asset demand in the $n$th economy. 

\begin{Lem}\label{lem:A-scaling}
For every compact interval
$K\subset I_\gamma$,
\be\label{eq:A-scaling}
\lim_{n\to\infty} \, \sup_{r\in K}\,
\Big|
\frac{A_n(r)}{\eta_n}
-
L_\gamma(r)
\Big| =0.
\ee
\end{Lem}

\noindent
The proof of this Lemma is presented in Appendix \ref{app:A-scaling}. Now compute
$$
\frac{\d}{\d r}\log L_\gamma(r)
=
\frac{1}{\gamma(\rho-r)}
-
\frac{\gamma-1}{\rho+(\gamma-1)r}.
$$
Hence, $L_\gamma$ is strictly decreasing on $I_\gamma$ whenever 
$$ 
r<
r_\gamma:=
\rho\,\frac{\gamma^2-\gamma-1}{\gamma^2-1}.
$$
Since $I_\gamma\cap(-\infty,r_\gamma)\neq\emptyset$,
we may choose $r_1<r_2$ in $I_\gamma$ such that $L_\gamma(r_1)>L_\gamma(r_2)$.
In addition, if $\gamma>(1+\sqrt5)/2$, then $r_\gamma>0$, so that $r_1,r_2$ may be chosen positive.

Now choose $\delta>0$ such that
$
L_\gamma(r_1)>L_\gamma(r_2)+2\delta.
$
Then, for all sufficiently large $n$, by \eqref{eq:A-scaling},
$$
\frac{A_n(r_1)}{\eta_n}\ge L_\gamma(r_1)-\delta
>
L_\gamma(r_2)+\delta
\ge 
\frac{A_n(r_2)}{\eta_n}\qquad \Rightarrow\qquad A_n(r_1)>A_n(r_2).
$$
For such $n$ this provides the desired counterexample.

\section{General equilibrium}\label{sec:general-equilibrium}

The goal of this section is to prove Propositions and Theorems  \ref{prop:fixed-r}, \ref{thm:fixed-tau-surplus}, \ref{thm:fixed-tau-deficit}, and  \ref{thm:huggett-to-aiyagari}. We let Assumptions \ref{asm:standing}, the no-borrowing limit, and Assumption \ref{asm:utility-function} hold.

\subsection{Walras' law}

Given an interest rate $r$, wage $w$, and tax-and-transfer rate $\tau$, let $c^*(a,z)=c^*(a,z;r,w,\tau)$ and $G^*(\d a, \d z)=G^*(\d a, \d z;r,w,\tau)$ denote the optimal consumption policy and unique stationary distribution, respectively. We define
\be \label{eq:aggregates}
A(r,w,\tau) := \int_\sX a\, G^*(\da,\dz;r,w,\tau),\quad C(r,w,\tau): = \int_\sX c^*(a,z;r,w,\tau) \, G^*(\da,\dz;r,w,\tau).
\ee

In the Huggett model, $w=1$ and we set $A(r,\tau) = A(r,1,\tau)$ and $C(r,\tau) = C(r,1,\tau).$

\begin{Lem}[Walras' law for Huggett]\label{lem:Walras-Huggett}
Assume that $\Xi=(\tau,B,r,c,G)$ with $r<\rho$ and $\tau<1$ satisfies conditions \eqref{def:Hug-I} and \eqref{def:Hug-II} of the definition of a real stationary Huggett equilibrium. For $r\neq0$, consider the following conditions:
\begin{enumerate}[(i)]
\item $\Xi$ satisfies \eqref{def:Hug-III} (government's budget constraint);
\item $\int a\, G(\da, \dz)=B$ (asset market clearing);
\item $\int c(a,z)\, G(\da, \dz)=1$ (goods market clearing).
\end{enumerate}
Then, any two conditions above imply the other one, in which case $\Xi$ is a Huggett equilibrium.
\end{Lem}

\begin{proof}
By conditions \eqref{def:Hug-I} and \eqref{def:Hug-II} of a stationary Huggett equilibrium, $c=c(a,z)$ is an optimal control of the household problem given the parameters $r$, $w=1$ and $\tau$. By Theorem \ref{thm:household} \eqref{thm:household-control}, it has to coincide with the optimal feedback control $c^*=c^*(a,z;r,\tau)$. Likewise, by uniqueness of the invariant distribution, we must have $G(\da,\dz)=G^*(\da,\dz;r,\tau)$. By stationarity of $G^*(\da,\dz;r,\tau)$ and the fact that it is compactly supported on $\sX$, we may integrate the optimal state dynamics with respect to $G$ to obtain
$$
0=rA(r,\tau)+(1-\tau)-C(r,\tau).
$$
Then,
$$
rB=\tau\quad \text{and}\quad A(r,\tau)=B \qquad \Longrightarrow \qquad r A(r,\tau) = \tau \qquad \Longrightarrow \qquad C(r,\tau)=1.
$$
The other two implications are similarly obtained.
\end{proof}

\noindent 
We now turn to the Aiyagari version of our problem and first observe that in any stationary equilibrium $\Xi=(\tau,B,K,r,w,c,G)$, there is a one-to-one correspondence between $r\in(-\delta,\infty)$, $K\in(0,\infty)$ and $w\in(0,\infty)$ using equations \eqref{eq:aiy-market-prices-eq}. Using the competitive equilibrium between firms (see \eqref{def:aiyagari-5}) in the definition of Aiyagari equilibria, we may  express $K=K^*(r)$ and $w=w^*(r)$ in equilibrium as functions of the interest rate $r$,
\be\label{eq:K-w-func-of-r}
K^*(r) := \left( \frac{\alpha}{r+\delta} \right)^{\frac{1}{1-\alpha}}\!\!\!, \quad 
w^*(r) := (1-\alpha) K^*(r)^\alpha= (1-\alpha) \left( \frac{\alpha}{r+\delta} \right)^{\frac{\alpha}{1-\alpha}}\!\!\!,\qquad r\in(-\delta,\infty).
\ee

\begin{Lem}[Walras' law for Aiyagari]\label{lem:aiy-walras}
Assume that $\Xi=(\tau,B,K,r,w,c,G)$ satisfies conditions \eqref{def:aiyagari-1}, \eqref{def:aiyagari-2} and \eqref{def:aiyagari-5} of the definition of a real stationary Aiyagari equilibrium. For $r\in(-\delta,\rho)\setminus\{0\}$ and $\tau<1$, consider the following:
\begin{enumerate}[(i)]
\item $\Xi$ satisfies \eqref{def:aiyagari-3} (government's budget constraint);
\item $\int a \,\d G=K+B$ (asset market clearing);
\item $\int c(a,z)\, G(\da, \dz)+\delta K = F(K,1)$ (goods market clearing).
\end{enumerate}
Then, any of the two conditions above imply the other one, in which case $\Xi$ is an Aiyagari equilibrium.
\end{Lem}

\begin{proof}
As in the proof of Lemma \ref{lem:Walras-Huggett}, we see that $c(a,z)=c^*(a,z;r,w,\tau)$, $G(\da,\dz)=G^*(\da,\dz;r,w,\tau)$, and therefore
$$
0 = rA(r,w,\tau) + w(1-\tau) -C(r,w,\tau).
$$
By condition \eqref{def:aiyagari-5} and equation \eqref{eq:K-w-func-of-r}, $K=K^*(r)$ and $w=w^*(r)$. Using these relations,
$$
rK + w - \big(K^\alpha-\delta K\big)  = r K +(1-\alpha)K^\alpha - \big(K^\alpha-\delta K\big) = K \big((r+\delta)  - \alpha K^{\alpha-1}\big)=0.
$$
Then, 
\begin{align*}
A(r,w,\tau) = K + B \quad \text{and} \quad rB=w\tau\qquad &\Longrightarrow \qquad r A(r,w,\tau) = rK + w\tau \\
&\Longrightarrow\qquad C(r,w,\tau) = rK+w\\
&\Longrightarrow\qquad  C(r,w,\tau) = K^\alpha - \delta K.
\end{align*}
Here, we used the stationarity of $G$ in the second line and the computation in the previous display in the third line. The other implications are obtained in a similar manner.
\end{proof}

\subsection{Proof of Proposition \ref{prop:fixed-r}}

\begin{proof}[Proof of Proposition \ref{prop:fixed-r} \eqref{prop:fixed-r-huggett}]
Let $r<\rho$ be given and first assume $r\neq0$. By Lemma \ref{lem:Walras-Huggett}, any $\tau^*<1$ that satisfies $C(r,\tau^*)=1$ leads to an equilibrium upon defining $B:=A(r,\tau^*)$. By Theorem \ref{thm:properties-G} \eqref{thm:properties-G-linearity-of-A-in-tau}, $C(r,\tau)=(1-\tau)\,C(r,0)$ and since the optimal consumption policy is strictly positive everywhere on $\sX$, we have $C(r,0)>0$. Therefore, there exists a unique value of $\tau^*\in(-\infty,1)$ such that $(1-\tau^*)\,C(r,0)=1$. In the case $r=0$, as in Lemma \ref{lem:Walras-Huggett}, we see $C(0,\tau)=1-\tau$ using stationarity of $G^*$ and the state dynamics. Hence $\tau^*=0$ is the unique value that satisfies goods market clearing. Setting $B:=A(0,0)$ again leads to an equilibrium.
\end{proof}

\begin{proof}[Proof of Proposition \ref{prop:fixed-r} \eqref{prop:fixed-r-aiyagari}]
\emph{Uniqueness.} Let $r\in(-\delta,\rho)$ be given and consider an equilibrium 
$$
\Xi=(\tau,B,K,r,w,c,G).
$$
We must have $K=K^*(r)$, $w=w^*(r)$ and $c(a,z)=c^*(a,z;r,w,\tau)$, $G(\da,\dz)=G^*(\da,\dz;r,w,\tau)$. Asset market clearing, together with the facts that $K>0$ and $B\ge0$ necessitate $r>\underline{r}$. By Theorem \ref{thm:properties-G} \eqref{thm:properties-G-linearity-of-A-in-tau}, for any $\tau<1$, $C(r,w,\tau)=(1-\tau)C(r,w,0)$. Since $C(r,w,0)>0$, there exists at most one value of $\tau<1$ such that the goods market clears: $(1-\tau)C(r,w,0)=K^\alpha-\delta K$.

\emph{Existence.} Let $r>\underline{r}$ be given. In the case $r\neq 0$, we set $K=K^*(r)$, $w=w^*(r)$, and follow the previous reasoning to seek $\tau<1$ such that the goods market clears. Then, by Lemma \ref{lem:aiy-walras}, a unique equilibrium exists if $B:=A(r,w,\tau)-K\geq0$. This is equivalent to
$$
\frac{K^\alpha-\delta K}{C(r,w,0)}= 1-\tau \geq \frac{K}{A(r,w,0)}\quad \Longleftrightarrow \quad \frac{r+\delta}{\alpha} - \delta \geq \frac{C(r,1,0)}{A(r,1,0)}.
$$
This is possible, for example, if $\alpha$ is small or $\delta$ is large. In the case $r=0$, we directly verify that $C(0,w,0)=w=K^\alpha-\delta K$. The government's budget constraint \eqref{def:aiyagari-3} necessitates $\tau=0$. In order to satisfy the asset market clearing condition, we are again led to the condition $A(0,w,0)-K\ge 0$. This is equivalent to $A(0,w,0)\ge (\alpha/\delta)^{1/(1-\alpha)}$, which is again possible for sufficiently large $\delta$ or small $\alpha$.
\end{proof}

\subsection{Proof of Theorem \ref{thm:fixed-tau-surplus}} \label{ssec:proof-surplus}

Let $\tau<1$ be given. Using Lemma \ref{lem:Walras-Huggett}, Huggett equilibria are characterized by interest rates $r^*$ such that 
\be\label{eq:proof-hugget-eq}
r^* A(r^*,\tau)=\tau
\ee
upon letting $B:=A(r^*,\tau)$. Turning to Aiyagari equilibria, we see that in any equilibrium, we must have $(K^*,w^*)=(K^*(r^*),w^*(r^*))$. Using Theorem \ref{thm:properties-G} \eqref{thm:properties-G-linearity-of-A-in-tau},
$$
A(r,w^*(r),\tau) = w^*(r) A(r,1,\tau),\qquad r<\rho.
$$
For $r\neq0$, the government's budget constraint implies $B=w^*(r)\tau/r$, so that the asset market clearing condition becomes
\be\label{eq:aiyagari-asset-clearing}
A(r,1,\tau) \overset{!}{=} \frac{K^*(r)}{w^*(r)} + \frac{\tau}{r} = \frac{\alpha}{1-\alpha} \,  \frac{1}{r+\delta} + \frac{\tau}{r} =: S(r).
\ee
If $r^*\in(-\delta,\rho)\setminus\{0\}$ satisfies this equation and $\tau/r^*\geq0$, then setting $B^*:=w^*(r^*)\tau/r^*$ leads to a stationary Aiyagari equilibrium by Lemma \ref{lem:aiy-walras}.

\begin{proof}[Proof of Theorem \ref{thm:fixed-tau-surplus} \eqref{thm:fixed-tau-surplus-positive}]
Assume that $\tau\in(0,1)$. The map $(-\infty,\rho)\ni r\mapsto A(r,w,\tau)$ is continuous and increasing with $\lim_{r\uparrow\rho}A(r,1,\tau)=\infty$.\\

\emph{Huggett model.}  It is clear that there exists exactly one intersection point of $r\mapsto A(r,1,\tau)$ with the strictly decreasing map $r\mapsto \tau/r$ and that the intersection point $r^*$ satisfies $0<r^*<\rho$. \\

\emph{Aiyagari model.} $S(r)$ is strictly decreasing on  $r\in(-\delta,0)$ and $(0,\infty)$ with 
$$
\lim_{r\downarrow-\delta}S(r) = \lim_{r\downarrow0} S(r) = \infty,\quad \lim_{r\uparrow0} S(r)=-\infty.
$$
Therefore, there exist exactly two intersection points $-\delta<r^{(\ell)}<0<r^{(h)}<\rho$ that satisfy $A(r^{(i)},1,\tau)=S(r^{(i)})$, $i=\ell,h$. However, only $r^*=r^{(h)}$ satisfies $\tau/r^*\geq0$. By the government's budget constraint, $r^*=0$ cannot lead to an equilibrium. Therefore, there is a unique equilibrium in which $r^*=r^{(h)}\in(0,\rho)$.
\end{proof}

\begin{proof}[Proof of Theorem \ref{thm:fixed-tau-surplus} \eqref{thm:fixed-tau-surplus-zero}]
Let $\tau=0$. \\

\emph{Huggett model.} Equilibria are again characterized by interest rates $r^*$ that satisfy $r^*A(r^*,0)=0$. We must either have $B^*=A(r^*,0)=0$ or $r^*=0$. By Theorem \ref{thm:properties-G} \eqref{thm:properties-G-hand-to-mouth}, the former holds if $r^*\le \underline{r}$. For $r=0$, set $B=A(0,0)$. Then, stationarity implies goods market clearing and the government's budget constraint is trivially satisfied. Hence, $r=0$ leads to an equilibrium as well.\\

\emph{Aiyagari model.} If $\tau=0$, then $S(r)$ is strictly decreasing on $(-\delta,\infty)$. By the government's budget constraint, we must either have $B^*=0$ or $r^*=0$. If $B^*=0$, then there exists exactly one  $r^*\in(-\delta,\rho)$ such that $A(r^*,1,\tau)=S(r^*)$, which leads to an equilibrium. If $\tau=0$ and $r=0$, then the government budget constraint is void. The equilibrium quantities $K^*,w^*$ are determined by
$$
K^*= \left(\frac{\alpha}{\delta}\right)^{\frac{1}{1-\alpha}},\qquad w^*= (1-\alpha)(K^*)^\alpha.
$$
By firm optimality (see proof of Lemma \ref{lem:aiy-walras}), $(K^*)^\alpha-\delta K^*=w^*$. Together with stationarity, $C^*=w^*$, so that the goods market clears. Asset market clearing requires $B^*= A(0,w^*,0)-K^*\ge0$, which is equivalent to $A(0,1,0)\ge K^*/w^*=\alpha/((1-\alpha)\delta)$. If this inequality is strict, this gives an equilibrium with $B^*>0$. If equality holds, this coincides with the zero-debt equilibrium.
\end{proof}

\subsection{Proof of Theorem \ref{thm:fixed-tau-deficit}}
This subsection presents the proofs of Theorem \ref{thm:fixed-tau-deficit} \eqref{thm:fixed-tau-deficit-small} and \eqref{thm:fixed-tau-deficit-large}. Due to its length, the proof of Theorem \ref{thm:fixed-tau-deficit} \eqref{thm:fixed-tau-deficit-multiple-equilibria} is presented in the next subsection.

\begin{proof}[Proof of Theorem \ref{thm:fixed-tau-deficit} \eqref{thm:fixed-tau-deficit-small}]

\emph{Huggett model.} If $\tau<0$, we need $r^*$ to satisfy $r^*A(r^*,\tau)=\tau$. Now, the only possible intersection points of the maps $r\mapsto A(r,\tau)$ and $r\mapsto \tau/r$ defined on $(-\infty,\rho)$ are in $(-\infty,0)$ since $A(r,\tau)\geq0$. Hence, if the lower interest rate bound $\underline{r}$ is non-negative, then there is no intersection point. Next, consider the case $\underline{r}<0$. By Theorem \ref{thm:properties-G} \eqref{thm:properties-G-linearity-of-A-in-tau}, $A(r,\tau)=(1-\tau)A(r,0)$. The equilibrium condition \eqref{eq:proof-hugget-eq} becomes $r^* A(r^*,0)=\tau/(1-\tau)$. Let $f(r^*)$ denote the left, and $g(\tau)$ the right side, respectively. Then $f$ is continuous, $f(\underline{r})=f(0)=0$ and $f<0$ on $(\underline{r},0)$. Clearly, $g(\tau)$ is a strictly increasing continuous function on $(-\infty,0)$ with $g(-\infty)=-1$ and $g(0)=0$. Choose any $r_0\in(\underline{r},0)$ such that $f(r_0)<0$. For $\tau<0$ sufficiently close to zero, $g(\tau)\in (f(r_0),0)$. Applying the intermediate value theorem on $[\underline{r},r_0]$ and $[r_0,0]$ implies the existence of $r^*_1,r^*_2\in (\underline{r},0)$ with $f(r^*_1)=f(r^*_2)=g(\tau)$.\\

\emph{Aiyagari model. } If $\tau<0$, then the government's budget constraint necessitates that any equilibrium interest rate is strictly negative. Since $S(r)>0$ on $(-\delta,0)$, a necessary condition for an equilibrium is $\underline{r}<0$. As in the proof for Huggett equilibria, we see that there exist at least two equilibria whenever both $\alpha>0$ and $\tau<0$ are close to zero.
\end{proof}

\begin{proof}[Proof of Theorem \ref{thm:fixed-tau-deficit} \eqref{thm:fixed-tau-deficit-large}]
\emph{Huggett model.} 
We now prove that, for sufficiently large deficits, no Huggett equilibrium exists. 
Let $\underline r$ denote the lower interest-rate bound from Theorem \ref{thm:properties-G} \eqref{thm:properties-G-hand-to-mouth}, so that,
under the no-borrowing constraint, $A(r,0)=0$ for all $r\le \underline r$. First observe that $f(r)>-1$ for $r\le0$. Indeed, stationarity of the invariant distribution for the household problem with $\tau=0$ gives $0=rA(r,0)+1-C(r,0)$, and hence $f(r)=C(r,0)-1$. Since optimal consumption is strictly positive, $C(r,0)>0$, and therefore $f(r)>-1$. We continue to assume $\underline r<0$. Since $A(r,0)=0$ for $r\le \underline r$, possible
intersections with a negative level can only occur on the compact interval
$[\underline r,0]$. As shown above, $f$ is continuous on this interval with $f>-1$. Hence $m:=\min\{ f(r):r\in[\underline r,0]\}$ is well-defined and satisfies $-1<m<0$. 
The strict inequality $m<0$ follows from $A(r,0)>0$ for $r>\underline{r}$. Choose $\tau<0$ so negative that
$$
\frac{\tau}{1-\tau}<m.
$$
For such $\tau$,
the equation $ rA(r,0)=\tau/(1-\tau)$ has no solution, and hence, no stationary Huggett equilibrium exists.\\

\emph{Aiyagari model.} We conclude by showing that, for sufficiently large deficits, no Aiyagari equilibrium exists. Since any equilibrium with $\tau<0$ must have $r\in(-\delta,0)$, set
$ m(\delta):=\min\{ f(r):r\in[-\delta,0]\}$ where $f(r):=r A(r,1,0)$. As in the Huggett proof, stationarity at $\tau=0$ gives $f(r)=C(r,1,0)-1>-1$. Hence, $-1<m(\delta)\le0$. If
$$
\frac{\tau}{1-\tau}<m(\delta),
$$
then we claim that no Aiyagari equilibrium exists. Indeed, if an equilibrium existed, then $r\in(-\delta,0)$ and, using the identity $A(r,1,\tau)=(1-\tau)A(r,1,0)$ and the  equilibrium condition
$$
A(r,1,\tau)
= \frac{\alpha}{1-\alpha}\frac1{r+\delta}+\frac{\tau}{r}
$$
would imply, after multiplying by $r<0$ and dividing by $1-\tau>0$,
$$
f(r)
=
\frac{\tau}{1-\tau}
+ \frac{\alpha}{1-\alpha}
\frac{r}{(r+\delta)(1-\tau)}.
$$
But for $r\in(-\delta,0)$ the second term on the right-hand side is strictly negative. Therefore,
$$
f(r)
< \frac{\tau}{1-\tau} < m(\delta)
\le
f(r),
$$
contradiction. Hence no equilibrium exists for such $\tau$.
\end{proof}

\subsection{Proof of Theorem \ref{thm:fixed-tau-deficit} \eqref{thm:fixed-tau-deficit-multiple-equilibria}}

Using Theorem \ref{thm:properties-G} \eqref{thm:properties-G-linearity-of-A-in-tau}, equilibria with $r\neq0$ are characterized by the condition 
$$
D(r) := rA(r,1,0) = \frac{\tau}{1-\tau}.
$$
This section constructs an irreducible income process for which the function $D(r)$ crosses a negative level at least $d\in 2\N$ times. Let us set $N=d/2$.

\emph{Step 1. Two-state economy.} Let $b<0$, and let $0<z_\ell<1<z_h$ be arbitrary income states satisfying $(z_\ell+z_h)/2=1$. We first construct an economy that has $b$ as its lower interest rate bound. To this end, choose symmetric transition rates
$$
\lambda_{\ell h}
:=
\lambda_{h\ell}
:=
\frac{\rho-b}{(z_h/z_\ell)^\gamma-1}.
$$
If $(z_t)$ is a Markov chain with these specifications, then it has stationary mean one. For CRRA utility and under the
no-borrowing constraint, the interest-rate criterion Theorem \ref{thm:properties-G} \eqref{thm:properties-G-hand-to-mouth} reads
$$
\underline r
= \rho-\lambda_{h\ell}\left[\left(\frac{z_h}{z_\ell}\right)^\gamma-1\right] = b.
$$
Let $A_b(r)$ and $C_b(r)$ denote aggregate assets and consumption in this two-state economy, respectively, with $(w,\tau)=(1,0)$. Let $D_b(r):=rA_b(r)$ be the debt demand function. By the lower-bound characterization,
$$
D_b(r)=0\quad\text{for }r\le b,
\qquad
D_b(0)=0,
\qquad
D_b(r)<0\quad\text{for }r\in(b,0).
$$
Moreover $D_b$ is continuous, and $D_b(r)>-1$ for $r<0$, because stationarity gives $D_b(r)=C_b(r)-1$, and aggregate consumption is strictly positive.

\emph{Step 2.}
We now inductively construct $N=d/2$ many economies. Choose $b_1<0$
and set $F_0(r)\equiv0$. Let $1\le k\le N$, and suppose $b_k<0$ and a function $F_{k-1}(r)$ have been chosen. Let $D_k(r)$ be the debt-demand function of a two-state economy with lower bound $b_k$ and with income states $\cZ_k$ disjoint from the previously chosen income states. Choose $c_k\in(b_k,0)$. Then $D_k(c_k)<0$. Since
$F_{k-1}(c_k)$ is finite, choose $\alpha_k>0$ so large that
$$
F_{k-1}(c_k)+\alpha_kD_k(c_k)<-1,\qquad \text{and set}\qquad F_k(r):=F_{k-1}(r)+\alpha_kD_k(r).
$$
If $k<N$, since $F_k(0)=0$, we may choose
$b_{k+1}\in(c_k,0)$ close enough to zero that
$$
F_k(b_{k+1})>-\frac12.
$$
After $N$ steps, set $b_{N+1}:=0$, and define
$$
F(r):=F_N(r)=\sum_{j=1}^N\alpha_jD_j(r),
\qquad
S:=\sum_{j=1}^N\alpha_j.
$$
Since $D_j(r)=0$ for $r\le b_j$, 
\be \label{eq:proof-multiple-eq-I}
F(b_j)>-\frac12,
\quad 1\le j\le N+1,\qquad \text{and}\qquad 
F(c_j)<-1,
\quad 1\le j\le N
\ee
Note $S>1$, since
$\alpha_1D_1(c_1)<-1$ while $D_1(c_1)>-1$. Set 
$$
\beta_j:=\frac{\alpha_j}{S},
\qquad j=1,\ldots,N,
$$
and define a $d$-state income process on the state space $\cZ:=\cZ_1\uplus \cdots\uplus \cZ_N$ with generator
$$
(\sL_0 \varphi)(z) = (\sL_j(\varphi|_{\cZ_j}))(z),
\qquad z\in\cZ_j
$$
where $\sL_j$ denotes the generator of the $j$th income chain. Let $G^r_j\in\sP([0,\infty)\times\cZ_j)$ denote the unique stationary distribution of the household problem of the $j$th two-state economy with interest rate $r$. We view $G_j^r$ as a measure on $[0,\infty)\times\cZ$ by extending it
by zero outside $[0,\infty)\times\cZ_j$, and define
$$
G^r_0(\da,\dz) := \sum_{j=1}^N \beta_j G^r_j(\da,\d z)\in \sP([0,\infty)\times \cZ).
$$
Since the states in $\cZ_j$ and $\cZ_k$ do not communicate for $j\neq k$, it is easy to see that $G^r_0$ is a stationary distribution\footnote{By no means is it the \emph{unique} one.} of the problem corresponding to the $d$-state income chain with generator $\sL_0$. The corresponding aggregate
real debt demand is
$$
D_0(r):=r \, \int_{[0,\infty)\times \cZ}a \, G^r_0(\da,\dz)= \sum_{j=1}^N\beta_jD_j(r)=\frac{F(r)}{S}.
$$
From \eqref{eq:proof-multiple-eq-I},
$$
D_0(b_j)>-\frac{1}{2S},
\quad j=1,\ldots,N+1,\qquad \text{and}\qquad 
D_0(c_j)<-\frac1S,
\quad j=1,\ldots,N.
$$
Choose $\eta:= -3/(4S)\in(-1,0)$. Then, for every $j=1,\ldots,N$,
\be \label{eq:proof-multiple-eq-II}
D_0(b_j)>\eta,
\qquad
D_0(c_j)<\eta,
\qquad
D_0(b_{j+1})>\eta.
\ee
By continuity, the equation $D_0(r)=\eta$ has at least two solutions in each interval $(b_j,b_{j+1})$. Hence it has at least
$2N=d$ solutions in $(b_1,0)$. 

\emph{Step 3.} It remains to perturb the income chain into an irreducible process without destroying these crossings. Recall that $(1/2,1/2)\in\sP(\cZ_j)$ is the invariant distribution of $\sL_j$. Define a stationary income distribution
$$
\bar\pi(z):=\beta_j/2,
\qquad z\in\cZ_j.
$$
Let $R$ be an irreducible transition rate matrix on $\cZ$ with $\bar \pi$ as its invariant distribution, for instance
$$
R_{zy}:=\bar\pi(y),\quad z\ne y,
\qquad
R_{zz}:=-(1-\bar\pi(z)).
$$
For $\varepsilon>0$, set $\sL_\varepsilon:=\sL_0+\varepsilon R$, and let $G_\varepsilon^r\in\sP([0,\infty)\times \cZ)$ be the unique invariant distribution of the optimally controlled process associated with $\sL_\varepsilon$. We claim that, as $\varepsilon\downarrow0$, for every $r<0$,
\be \label{eq:proof-multiple-eq-III}
A_\varepsilon(r):=\int a\,G_\varepsilon^r(\da,\dz)
\longrightarrow
A_0(r):=\int a\,G^r_0(\da,\dz).
\ee
This implies $D_\varepsilon(r):=rA_\varepsilon(r)
\to D_0(r)$ for every $r<0$.  It remains to establish \eqref{eq:proof-multiple-eq-III}.

Let $v^\varepsilon$ and $c_\varepsilon$ denote the value function and
optimal feedback rule associated with $\sL_\varepsilon$. Define $v_0:[0,\infty)\times\cZ\to\R$ by $
v_0(a,z):=v_j(a,z)$, $ z\in\cZ_j$,
where $v_j$ is the value function of the household problem in the $j$th economy. We first note that $v_\varepsilon\to v_0$ and $\partial_a v^\varepsilon\to \partial_a v^0$ locally uniformly on $[0,\infty)\times\cZ$ by the same logic as in the proof of Lemma \ref{lem:joint-continuity-(a,r)}, with the convergence
$\sL_\varepsilon\to\sL_0$ replacing the convergence of interest rates.

Next, because $r<0$, the invariant measures $G_\varepsilon^r$ have a common
compact support. Indeed, if $\bar z:=\max\cZ$, choose $M<\infty$ so
large that $rM+\bar z<0$. Since $c_\varepsilon>0$,
$$
ra+z-c_\varepsilon(a,z)\le ra+\bar z<0
\qquad\text{for }a\ge M,
$$
uniformly in $z$ and $\varepsilon$. Thus, any invariant distribution
is supported on $[0,M]\times\cZ$. 

Let $\varepsilon_n\downarrow0$. By compactness, after passing to a
subsequence if necessary, $G_{\varepsilon_n}^r\Rightarrow \widehat G$ for some probability measure $\widehat G\in\sP([0,M]\times\cZ)$. For any smooth test function $\varphi:[0,M]\times\cZ\to\R$, stationarity of $G_{\varepsilon_n}^r$ implies
$$
0=
\int
\left[
s_{\varepsilon_n}(a,z)\partial_a\varphi(a,z) +
\sL_{\varepsilon_n}\varphi(a,\cdot)(z)
\right]\, 
G_{\varepsilon_n}^r(\da,\dz),
$$
where $s_{\varepsilon}(a,z):=ra+z-c_\varepsilon(a,z)$. Using the locally uniform convergence $s_{\varepsilon_n}\to s_0$, $\sL_{\varepsilon_n}\to\sL_0$, and weak
convergence of $G_{\varepsilon_n}^r$, we may pass to the limit to obtain
\be\label{eq:proof-multiple-eq-G-hat}
0=
\int
[
s_0(a,z)\partial_a\varphi(a,z) +
\sL_0\varphi(a,\cdot)(z)
]\, 
\widehat G(\da,\dz).
\ee
Thus $\widehat G$ is invariant for the limiting process.

We now argue that $\widehat G=G_0^r$.
First, since $\bar \pi\sL_0=\bar\pi R=0$ and $\sL_\varepsilon$ is irreducible, $\bar \pi$ is the unique invariant distribution of $\sL_\varepsilon$, for any $\varepsilon>0$. Therefore, the income marginal of $\widehat G$ is also $\bar\pi$, so that indeed $\widehat G([0,\infty)\times\cZ_j)=\beta_j$. Next, define $\widehat G_j
:=
\beta_j^{-1}\,
\widehat G|_{[0,\infty)\times\cZ_j}$. 
Taking test functions supported on $[0,\infty)\times\cZ_j$ in \eqref{eq:proof-multiple-eq-G-hat} shows that $\widehat G_j$ is invariant for the optimally controlled process of the $j$th economy. By Proposition
\ref{prop:existence-uniqueness-G}, the invariant distribution for each economy is unique, so $\widehat G_j=G_j^r$, and therefore
$$
\widehat G
=
\sum_{j=1}^N\beta_jG_j^r
=
G_0^r.
$$
Taken together, this shows that every subsequential weak limit is equal to $G_0^r$, hence $G_\varepsilon^r\Rightarrow G_0^r$. Since all measures are supported on the common compact set
$[0,M]\times\cZ$, weak convergence implies convergence of first moments:
$$
\int_{[0,\infty)\times \cZ} a\,G_\varepsilon^r(\da,\dz)
\longrightarrow
\int_{[0,\infty)\times \cZ} a\,G_0^r(\da,\dz).
$$

\emph{Step 4.} Using the strict inequalities \eqref{eq:proof-multiple-eq-II}, we see that, for sufficiently small $\varepsilon$, the equation $D_\varepsilon(r)=\eta$
has at least two solutions in each interval $(b_j,b_{j+1})$, and hence at least $d$
solutions in total. Finally set $\tau:=\eta/(1+\eta)<0$,  so that $\eta=\tau/(1-\tau)$. Since every solution of $D_\varepsilon(r)=\eta$ gives rise to a stationary Huggett equilibrium with
primary deficit $\tau$, the economy corresponding to the income process with generator $\sL_\varepsilon$ has at least $d$ stationary Huggett equilibria. This completes the proof.

\subsection{Proof of Theorem \ref{thm:huggett-to-aiyagari}}

Fix $\delta>0$ and, for $r\in(-\delta,\rho)$, set
$$
A_\tau(r):=A(r,1,\tau),\qquad
\chi_\alpha(r):=\frac{\alpha}{1-\alpha}\frac{1}{r+\delta}.
$$
For $r\neq0$, the Aiyagari equilibrium condition is
\begin{equation}\label{eq:aiy-to-hug-proof-aiy}
A_\tau(r)=\frac{\tau}{r}+\chi_\alpha(r),
\qquad \frac{\tau}{r}\ge0.
\end{equation}
The corresponding Huggett equilibrium condition is
\begin{equation}\label{eq:aiy-to-hug-proof-hug}
A_\tau(r)=\frac{\tau}{r}.
\end{equation}
First, let $\tau\in(0,1)$ be given. For each $\alpha\in(0,1)$, let $r_\alpha\in(0,\rho)$ be the
unique Aiyagari equilibrium interest rate, and let $r_H\in(0,\rho)$ be the unique Huggett
equilibrium interest rate. Take any sequence $\alpha_n\downarrow0$ and write
$r_n:=r_{\alpha_n}$. Passing to a subsequence if necessary, let $r_n\to \bar r\in[0,\rho]$.
We first rule out $\bar r\in\{0,\rho\}$. Indeed, if $\bar r=0$, then the right-hand side of \eqref{eq:aiy-to-hug-proof-aiy} diverges to $+\infty$, while $A_\tau(r_n)\to A_\tau(0)<\infty$,
which is impossible. If $\bar r=\rho$, then $A_\tau(r_n)\to\infty$, whereas $\tau/r_n+\chi_{\alpha_n}(r_n)$ remains bounded, again a contradiction. Hence $\bar r\in(0,\rho)$. Letting $n\to\infty$ in
\eqref{eq:aiy-to-hug-proof-aiy}, using continuity of $A_\tau$ and
$\chi_{\alpha_n}(r_n)\to0$, gives $A_\tau(\bar r)=\tau/\bar r$. By uniqueness of the Huggett equilibrium for $\tau>0$, $\bar r=r_H$. Since every convergent subsequence has the same limit, $r_\alpha\to r_H$.

Next, we consider the equilibrium quantities. The Aiyagari capital and wage are
$$
K_\alpha=\left(\frac{\alpha}{r_\alpha+\delta}\right)^{\frac1{1-\alpha}},
\qquad
w_\alpha=(1-\alpha)
\left(\frac{\alpha}{r_\alpha+\delta}\right)^{\frac{\alpha}{1-\alpha}}.
$$
Since $r_\alpha\to r_H>0$, we have $K_\alpha\to0$ and $w_\alpha\to1$. Moreover, the government budget constraint gives
$$
B_\alpha=\frac{w_\alpha\tau}{r_\alpha}
\longrightarrow
\frac{\tau}{r_H}
=
A_\tau(r_H)
=:B_H.
$$
By Lemma \ref{lem:CRRA-scaling},
$$
c^*(a,z;r,w,\tau)
=
w\,c^*(a/w,z;r,1,\tau)\qquad \text{and}\qquad 
G^*(\cdot\,;r,w,\tau)
=
(a\mapsto wa)_\#G^*(\cdot\,;r,1,\tau).
$$
Together with the continuity of the optimal consumption policy and invariant measure in $r$, this implies
$$
c^*(\cdot,\cdot;r_\alpha,w_\alpha,\tau)
\to
c^*(\cdot,\cdot;r_H,1,\tau) \quad \text{locally uniformly, and }\quad G^*(\cdot\,;r_\alpha,w_\alpha,\tau)
\Rightarrow
G^*(\cdot\,;r_H,1,\tau).
$$
Thus, the Aiyagari equilibrium converges to the unique Huggett equilibrium.

Now let $\tau<0$, let $\alpha_n\downarrow0$, and let
$$
\Xi_n=(\tau,B_n,K_n,r_n,w_n,c_n,G_n)
$$
be Aiyagari equilibria satisfying $x_n:=\alpha_n/(r_n+\delta)\to0$. Since $\tau<0$, every Aiyagari equilibrium has $r_n\in(-\delta,0)$. Passing to a subsequence, without relabeling, let $\bar r\in[-\delta,0]$ be a limit point of
$(r_n)$. By assumption, $\chi_{\alpha_n}(r_n)=x_n/(1-\alpha_n)\to0$. If $\bar r=0$, then $\tau/r_n\to+\infty$, contradicting the finiteness and continuity of
$A_\tau$ near $0$ in \eqref{eq:aiy-to-hug-proof-aiy}. Hence, $\bar r\in[-\delta,0)$.
Letting $n\to\infty$ in \eqref{eq:aiy-to-hug-proof-aiy} gives $A_\tau(\bar r)=\tau/\bar r$,
or equivalently $\bar rA_\tau(\bar r)=\tau$.
Thus, with
$$
\bar B:=A_\tau(\bar r)=\frac{\tau}{\bar r}>0,
$$
the government budget constraint and asset market clearing condition of the Huggett model hold.

It remains only to identify the limiting household objects.  Since $x_n\ge \alpha_n/\delta$,
$$
0\ge \alpha_n\log x_n
\ge \alpha_n\log(\alpha_n/\delta)
\longrightarrow0,
$$
and therefore
$$
K_n=x_n^{1/(1-\alpha_n)}\to0,
\qquad
w_n=(1-\alpha_n)x_n^{\alpha_n/(1-\alpha_n)}\to1.
$$
Consequently,
$$
B_n=\frac{w_n\tau}{r_n}\to \frac{\tau}{\bar r}=\bar B.
$$
Using again Lemma \ref{lem:CRRA-scaling} and continuity of the household policy and invariant measure, we obtain
$$
c_n(\cdot,\cdot;r_n,w_n,\tau)\to c^*(\cdot,\cdot;\bar r,1,\tau),
\qquad
G_n(\cdot\,;r_n,w_n,\tau)\Rightarrow G^*(\cdot\,;\bar r,1,\tau).
$$
Therefore every limit point of Aiyagari equilibria is the Huggett equilibrium
$$
\bigl(\tau,\bar B,\bar r,c^*(\cdot,\cdot;\bar r,1,\tau),
G^*(\cdot\,;\bar r,1,\tau)\bigr).
$$
Since $r_n>-\delta$ for all $n$, the limiting Huggett equilibrium satisfies $\bar r\ge -\delta$.

\appendix

\section{Proof of Lemma \ref{lem:A-scaling}} \label{app:A-scaling}
Let $K\subset I_\gamma$ be compact.
For $a\ge0$, $i\in\{\ell,h\}$, and $r\in K$, set
$$
p_i^n(a;r):=v^n_a(a,z_i^n;r),
\qquad
c_i^n(a;r):=\bigl(p_i^n(a;r)\bigr)^{-1/\gamma}, \qquad s_i^n(a;r):=ra+z_i^n-c_i^n(a;r).
$$
We suppress the dependence on $r$ when no confusion can arise.\\

\emph{Step 1. Low-income scaling.}
Set
$$
V_\ell^n(x;r):=\eta_n^{\gamma-1}v^n(\eta_n x,z_\ell^n;r),
\qquad x>0,\ r\in K.
$$
We first prove that
\begin{equation}\label{eq:proof-lem-A-scaling-value}
V_\ell^n(x;r)\longrightarrow
V_0(x;r):=
\frac{\kappa(r)^{-\gamma}}{1-\gamma}x^{1-\gamma},\qquad \text{locally uniformly on} \ (0,\infty)\times K\ \text{as}\ n\to\infty.
\end{equation}
Note that $V_0(\cdot;r)$ is the value function of the deterministic zero-income problem
$$
\sup_{c_t>0}\ \int_0^\infty e^{-\rho t}\frac{c_t^{1-\gamma}}{1-\gamma}\,\dt,
\qquad
\dot x_t=rx_t-c_t,\qquad x_t\ge0.
$$
Now, scaling any admissible control for the deterministic problem started from $x$ by $\eta_n$ gives an admissible control for the $n$th stochastic problem started from $(\eta_nx,z^n_\ell)$. Hence, $V_\ell^n(x;r)\ge V_0(x;r)$.

For the reverse inequality, fix $T>0$ and let
$$
E_n^T:=\{z_t=z_\ell^n\ \text{for all }0\le t\le T\}.
$$
Then $\P_{z^n_\ell}(E_n^T)=e^{-\mu_nT}$. Let $(c_t)$ be any admissible control
started from $(\eta_nx,z_\ell^n)$ and set $\tilde c_t:=c_t/\eta_n$. Since
$\gamma>1$, we have $u\le0$ and therefore
$$
\eta_n^{\gamma-1}
\E\left[\int_0^\infty e^{-\rho t}u(c_t)\,\dt\right]
=
\E\left[\int_0^\infty e^{-\rho t}u(\tilde c_t)\,\dt\right]
\le
\E\left[\chi_{E_n^T}\int_0^T e^{-\rho t}u(\tilde c_t)\,\dt\right].
$$
On $E_n^T$, the budget constraint gives
$$
\int_0^T e^{-rt}\tilde c_t\,\dt
\le
x+\frac{z_\ell^n}{\eta_n}D_T(r),
\qquad
D_T(r):=\int_0^T e^{-rt}\,\dt .
$$
Thus, on $E_n^T$, the last integral is pathwise bounded above by
$$
\sup\Big\{
\int_0^T e^{-\rho t}\frac{k_t^{1-\gamma}}{1-\gamma}\,\dt\,:\, k_t>0\ \text{and} \ \int_0^T e^{-rt}k_t\,\dt\le X\Big\}
=
\frac{B_T(r)^\gamma}{1-\gamma}X^{1-\gamma},
$$
with $B_T(r):=\int_0^T e^{-\kappa(r)t}\,\dt$ and $X:=x+z_\ell^n D_T(r)/\eta_n$.
Consequently,
$$
V_\ell^n(x;r)
\le
e^{-\mu_nT}
\frac{B_T(r)^\gamma}{1-\gamma}
\left(x+\frac{z_\ell^n}{\eta_n}D_T(r)\right)^{1-\gamma}.
$$
Since $z_\ell^n/\eta_n\to0$, $\mu_n\to0$, and $D_T$ is bounded on $K$, it follows that, for every compact interval
$J\subset(0,\infty)$,
$$
\underset{n\to\infty}{\limsup}\,\sup_{(x,r)\in J\times K}\,
\bigl(V_\ell^n(x;r)-V_0(x;r)\bigr)
\le
\sup_{(x,r)\in J\times K}
\frac{B_T(r)^\gamma-\kappa(r)^{-\gamma}}{1-\gamma}
x^{1-\gamma}.
$$
Now  $B_T(r)\uparrow \kappa(r)^{-1}$ uniformly in $r\in K$ as $T\to\infty$.  Combining this with
$V_\ell^n\ge V_0$ proves \eqref{eq:proof-lem-A-scaling-value}.

Moreover, using concavity and the same reasoning as in the proof of Lemma
\ref{lem:joint-continuity-(a,r)}, it is easy to see that 
$$
\partial_xV_\ell^n(x;r)\longrightarrow \partial_xV_0(x;r)\qquad \text{locally uniformly on} \ (0,\infty)\times K.
$$
Since $\partial_xV_\ell^n(x;r)=\lambda_n p_\ell^n(\eta_n x;r)$ and $ \partial_xV_0(x;r)=\kappa(r)^{-\gamma}x^{-\gamma}$,
we obtain
\begin{equation}\label{eq:proof-lem-A-scaling-I}
\lambda_n p_\ell^n(\eta_n x;r)
\longrightarrow \frac{1}{(\kappa(r)x)^{\gamma}}
\qquad \text{locally uniformly on} \ (x,r)\in (0,\infty)\times K.
\end{equation}

\emph{Step 2. Target wealth asymptotic.} Next, for $n\ge1$, let $\bar a_n(r)$ be the largest zero of the high-state saving rule $s_h^n(\bar a_n(r);r)=0$. We claim that
\be\label{eq:proof-lem-A-scaling-II}
x_n(r):=\frac{\bar a_n(r)}{\eta_n}\longrightarrow
L_\gamma(r),
\qquad\text{uniformly on }K.
\ee

First, note that $\bar a_n(r)$ is well-defined in $(0,\infty)$ since $s_h^n(a;r)<0$ for all sufficiently large $a$ by Lemma \ref{lem:properties-saving-rules} \eqref{lem:properties-saving-rules-III}, while
$s_h^n(0;r)>0$ for all sufficiently large $n$. To see the latter, suppose $s_h^n(0;r)=0$. Then $p_h^n(0;r)=(z_h^n)^{-\gamma}$ and, by Lemma \ref{lem:properties-saving-rules} \eqref{lem:properties-saving-rules-I}, $p_\ell^n(0;r)=(z_\ell^n)^{-\gamma}$. The boundary inequality \eqref{eq:Euler-bdry-zero} in Lemma \ref{lem:Euler} gives
$$
(\rho-r+\lambda_n)p^n_h(0;r)
\ge
\lambda_n (z^n_\ell)^{-\gamma} 
=
\left(\frac{\eta_n}{z_\ell^n}\right)^\gamma
\longrightarrow\infty,
$$
contradicting boundedness of $p_h^n(0;r)$ on $K$. Hence $s_h^n(0;r)>0$ eventually.

Next, we show that $x_n(r)$ remains in a compact subset of $(0,\infty)$, for all $r\in K$. At the
high-state cutoff $\bar a_n(r)$, we have $s_h^n(\bar a_n(r);r)=0$, so $
p_h^n(\bar a_n(r);r)
= \bigl(r\eta_n x_n(r)+z_h^n\bigr)^{-\gamma}$. Combined with the Euler equation, 
\be\label{eq:proof-lem-A-scaling-III}
\lambda_n p_\ell^n(\eta_n x_n(r);r)
=
(\rho-r+\lambda_n)
\bigl(r\eta_n x_n(r)+z_h^n\bigr)^{-\gamma}.
\ee
We first prove 
$$
\underset{n\to\infty}{\liminf}\, \inf_{r\in K}\, x_n(r)>0.
$$
Suppose not. Then, after passing to a subsequence, there exist
$r_n\in K$ such that $x_n(r_n)\to0$. For any $x_0>0$,
$x_n(r_n)\le x_0$ eventually. Since $p_\ell^n(\cdot;r)$ is
non-increasing, $\lambda_n p_\ell^n(\eta_n x_n(r_n);r_n)
\ge \lambda_n p_\ell^n(\eta_n x_0;r_n)$.
By the locally uniform convergence in \eqref{eq:proof-lem-A-scaling-I},
$$
\lambda_n p_\ell^n(\eta_n x_0;r_n)
=
\kappa(r_n)^{-\gamma}x_0^{-\gamma}+o(1),
$$
On the other hand,
\eqref{eq:proof-lem-A-scaling-III} and $x_n(r_n)\to0$ imply
$$
\lambda_n p_\ell^n(\eta_n x_n(r_n);r_n)
=
(\rho-r_n+\lambda_n)
(r_n\eta_nx_n(r_n)+z_h^n)^{-\gamma}
=
\rho-r_n+o(1),
$$
uniformly along the sequence. Since $x_0>0$ can be chosen arbitrarily small, the preceding two displays contradict each other. 

Next, we claim
\be\label{eq:proof-lem-A-scaling-upper-bound-xn}
\underset{n\to\infty}{\limsup}\,\sup_{r\in K}\, x_n(r)<\infty.
\ee
Since $\gamma>1$, we have $u\le0$. Choosing a deterministic zero-income policy implies the existence of $C_K'>0$ independent of $r\in K$, such that
$$
-C_K'a^{1-\gamma}\le v_\ell^n(a;r)\le0,\qquad \forall r\in K.
$$
Hence, by concavity, for $C_K:=2^\gamma C_K'$,
$$
p_\ell^n(a;r)
\le
\frac{v_\ell^n(a;r)-v_\ell^n(a/2;r)}{a/2}
\le
C_Ka^{-\gamma},\qquad \forall r\in K.
$$
Taking $a=\eta_nx_n(r)$ and using $\lambda_n=\eta_n^\gamma$ gives
$$
\lambda_np_\ell^n(\eta_nx_n(r);r)\le C_Kx_n(r)^{-\gamma}.
$$
Combining this with \eqref{eq:proof-lem-A-scaling-III}, for all large $n$,
$$
c_K x_n(r)\le r\eta_nx_n(r)+z_h^n
\le |r|\eta_nx_n(r)+z_h^n
$$
for some $c_K>0$ independent of $n$ and $r$. Since $K$ is compact,
$\eta_n\to0$, and $z_h^n\to1$, the first term on the right can be absorbed
into the left, establishing \eqref{eq:proof-lem-A-scaling-upper-bound-xn}. Hence, there are
$0<\underline x<\bar x<\infty$ such that
$$
x_n(r)\in[\underline x,\bar x],
\qquad \forall r\in K,\ n\ge1.
$$
By \eqref{eq:proof-lem-A-scaling-I} and \eqref{eq:proof-lem-A-scaling-III}, using
$\lambda_n\to0$, $z_h^n\to1$, and $\eta_nx_n(r)\to0$ uniformly on $K$, we see that
$$
\frac{1}{(\kappa(r)x_n(r))^{\gamma}}
=
\rho-r+o(1)
$$
uniformly for $r\in K$. Since $\kappa(r)$ and $\rho-r$ are bounded away from zero
on $K$, this is equivalent to
$$
x_n(r)\to
\frac{1}{\kappa(r)(\rho-r)^{1/\gamma}}
= L_\gamma(r)\qquad \text{uniformly in} \ r\in K.
$$

\emph{Step 3. Concentration around $\bar a_n(r)$.}
We show that the invariant wealth distribution concentrates at the high-state cutoff $\bar a_n(r)$. Choose $0<\delta<1$ so small that $\inf_{r\in K}L_\gamma(r)>2\delta$ and set 
$$
\bar a_\pm^n(r):=\eta_n(L_\gamma(r)\pm\delta).
$$
By Step 2., for $n\ge1$,
\be\label{eq:proof-lem-A-scaling-IV}
\bar a_-^n(r)
\le \bar a_n(r)
\le \bar a_+^n(r)
\qquad r\in K.
\ee
In particular, setting $ M:=1+\sup_{r\in K}L_\gamma(r)$ yields $\sup_{r\in K}\bar a_n(r)\le M\eta_n$.

\emph{Step 3a.} First choose $m_\delta>0$ such that
$$
\frac{1}{[\kappa(r)(L_\gamma(r)-\delta)]^{\gamma}}
\ge \rho-r+ 3m_\delta,
\qquad \forall r\in K,
$$
using $\kappa(r)^{-\gamma}L_\gamma(r)^{-\gamma}=\rho-r$ and compactness of $K$.
Combined with the low-income scaling \eqref{eq:proof-lem-A-scaling-I}, for sufficiently large $n$,
$$
\lambda_n p_\ell^n(\bar a^n_-(r);r)
\ge \rho-r+2m_\delta,
\qquad \forall  r\in K.
$$
Since $p_\ell^n(\cdot;r)$ is decreasing, it follows that
\be\label{eq:proof-lem-A-scaling-V}
\lambda_n p_\ell^n(\eta_n x;r)
\ge \rho-r+2m_\delta,
\qquad \forall x\in[0, L_\gamma(r)-\delta],\  r\in K.
\ee

\emph{Step 3b.} We claim that, for all sufficiently large $n$,
\begin{equation}\label{eq:proof-lem-A-scaling-VI}
s_h^n(a;r)>0,
\qquad \forall a\in [0,\bar a^n_-(r)],\ r\in K.
\end{equation}
Towards a contradiction, suppose that there exist $r_n\in K$ and $a_n\in[0,\bar a_-^n(r_n)]$ with $s_h^n(a_n;r_n)\le0$ for all $n$. Write $x_n:=a_n/\eta_n$. Then
$$
0\le x_n\le L_\gamma(r_n)-\delta\le M.
$$
By \eqref{eq:proof-lem-A-scaling-V},
\be \label{eq:proof-lem-A-scaling-V-3b-I}
\lambda_n p_\ell^n(a_n;r_n)
=
\lambda_n p_\ell^n(\eta_n x_n;r_n)
\ge \rho-r_n+2m_\delta .
\ee
On the other hand, $s_h^n(a_n;r_n)\le0$ implies
$$
c_h^n(a_n;r_n)\ge r_na_n+z_h^n,
\qquad
p_h^n(a_n;r_n)\le (r_na_n+z_h^n)^{-\gamma}.
$$
Since $a_n=\eta_n x_n$ with $x_n\le M$, $r_n\in K$, $\eta_n\to0$, and
$z_h^n\to1$, we have
$$
r_na_n+z_h^n=r_n\eta_n x_n+z_h^n\to1.
$$
Hence, $(\rho-r_n+\lambda_n)(r_na_n+z_h^n)^{-\gamma} = \rho-r_n+o(1)$.
Therefore, for sufficiently large $n$,
\be \label{eq:proof-lem-A-scaling-V-3b-II}
(\rho-r_n+\lambda_n)p_h^n(a_n;r_n)
\le \rho-r_n+m_\delta .
\ee
Combining \eqref{eq:proof-lem-A-scaling-V-3b-I} and \eqref{eq:proof-lem-A-scaling-V-3b-II},
$$
R_h^n(a_n;r_n)
\le
(\rho-r_n+m_\delta)-(\rho-r_n+2m_\delta)
=
-m_\delta<0,
$$
where
$$
R_h^n(a;r):=
(\rho-r+\lambda_n)p_h^n(a;r)-\lambda_n p_\ell^n(a;r).
$$
But Euler's equation implies $R_h^n(a;r)\ge0 $ whenever $s_h^n(a;r)\le0$, a contradiction. This proves
\eqref{eq:proof-lem-A-scaling-VI}.

Finally, since $\bar a_n(r)$ is the largest zero of $s_h^n(\cdot;r)$ and
$s_h^n(a;r)<0$ for all sufficiently large $a$, continuity implies $s_h^n(a;r)<0$ for all $a>\bar a_n(r)$. Together with $s_\ell^n(a;r)<0$ for $a>0$, this shows that
\begin{equation}\label{eq:proof-lem-A-scaling-VII}
\mathrm{supp}\, G_n^*(\cdot;r)
\subset
[0,\bar a_n(r)]\times\cZ_n
\subset
[0,M\eta_n]\times\cZ_n .
\end{equation}

\emph{Step 3c.} Let $\Phi_t^{n,r}(a)$ be the high-income flow,
$$
\frac{\d}{\dt}\Phi_t^{n,r}(a)
=
s_h^n(\Phi_t^{n,r}(a);r),
\qquad
\Phi_0^{n,r}(a)=a.
$$
We show that there exists $T_\delta<\infty$, independent of $n$ and $r\in K$,
such that
$$
\Phi_t^{n,r}(a)\in [\bar a_-^n(r), \bar a^n(r)],
\qquad t\ge T_\delta,
$$
for all sufficiently large $n$, all $r\in K$, and all
$0\le a\le \bar a_n(r)$.

By Step 3b, $s_h^n(\cdot;r)>0$ on $[0,\bar a_-^n(r)]$. Hence the time needed to
reach $\bar a_-^n(r)$ from any $a\le \bar a_-^n(r)$ is bounded by
$$
\int_0^{\bar a_-^n(r)}\frac{\dx}{s_h^n(x;r)}.
$$
We prove that this integral is uniformly bounded. Choose $\theta>0$ small enough, such that, for sufficiently large $n$,
\begin{equation}\label{eq:proof-lem-A-scaling-VIII}
(\rho-r+\lambda_n)(ra+z_h^n-\theta)^{-\gamma}
\le \rho-r+m_\delta,
\qquad \forall a\in [0,M\eta_n],\ r\in K.
\end{equation}
This is possible because $ra+z_h^n\to1$ and $\lambda_n\to0$. Let
$$
I_n(r):=\{a\in[0,\bar a_-^n(r)]:0<s_h^n(a;r)\le\theta\}.
$$
For $a\in I_n(r)$,
$$
c_h^n(a;r)=ra+z_h^n-s_h^n(a;r)\ge ra+z_h^n-\theta,
$$
and, by
\eqref{eq:proof-lem-A-scaling-V} and \eqref{eq:proof-lem-A-scaling-VIII},
$$
R_h^n(a;r)
=
(\rho-r+\lambda_n)p_h^n(a;r)-\lambda_n p_\ell^n(a;r)
\le -m_\delta .
$$
Euler's equation gives $R_h^n(a;r)=\partial_a p_h^n(a;r)s_h^n(a;r)$. Since $s_h^n(a;r)>0$ and $R_h^n(a;r)\le -m_\delta$, we have
$$
\partial_a p_h^n(a;r)
=
\frac{R_h^n(a;r)}{s_h^n(a;r)}
\le
-\frac{m_\delta}{s_h^n(a;r)}<0\qquad \Rightarrow\qquad 
\frac1{s_h^n(a;r)}
\le
-\frac{\partial_a p_h^n(a;r)}{m_\delta}.
$$
Moreover, since $z^n_h\to1$, $c^n_h(a;r)\geq 1/2$, for all $(a,r)\in I_n(r)\times K$, upon making $\theta$ smaller if necessary. This implies $p^n_h(a;r)\leq 2^\gamma$ on $(a,r)\in I_n(r)\times K$. Since $p_h^n(\cdot;r)$ is decreasing,
$$
\int_{I_n(r)}\frac{\da}{s_h^n(a;r)}
\le
\frac1{m_\delta}\int_{I_n(r)}-\partial_a p_h^n(a;r)\,\da
\le
\frac{2^\gamma}{m_\delta}.
$$
On the complement of $I_n(r)$ inside $[0,\bar a_-^n(r)]$, we have
$s_h^n(a;r)>\theta$, and therefore
$$
\int_{[0,\bar a_-^n(r)]\setminus I_n(r)}
\frac{\da}{s_h^n(a;r)}
\le
\frac{\bar a_-^n(r)}{\theta}
\le
\frac{M\eta_n}{\theta}
\le 1
$$
for sufficiently large $n$. Hence,
$$
\int_0^{\bar a_-^n(r)}\frac{\da}{s_h^n(a;r)}
\le
\frac{2^\gamma}{m_\delta}+1
=:T_\delta .
$$
This proves that the high-state flow reaches $\bar a_-^n(r)$ within time
$T_\delta$, uniformly in $n$ and $r\in K$.

Once the flow has reached $\bar a_-^n(r)$, it cannot go below it, because $s_h^n>0$ on $[0,\bar a_-^n(r)]$. It also cannot go above $\bar a_n(r)$ since $s_h^n(\bar a_n(r);r)=0$. Hence, for
$t\ge T_\delta$,
$$
\Phi_t^{n,r}(a)\in [\bar a_-^n(r),\bar a_n(r)].
$$
Using Step 2,
$$
\bar a_-^n(r) = \eta_n(L_\gamma(r)-\delta),
\qquad
\bar a_n(r)\le \bar a_+^n(r)= \eta_n(L_\gamma(r)+\delta),
$$
we conclude that
\begin{equation}\label{eq:proof-lem-A-scaling-IX}
|\Phi_t^{n,r}(a)-\bar a_n(r)|
\le
2\delta\eta_n,
\qquad t\ge T_\delta,
\end{equation}
for all $0\le a\le \bar a_n(r)$ and $r\in K$.

\emph{Step 3d.}
Let $(a_t,z_t)$ be the optimally controlled process started from its invariant
law $G_n^*(\cdot;r)$. By stationarity, $(a_t,z_t)\sim G_n^*(\cdot;r)$ for all
$t\ge0$. Define
$$
E_n:=\{z_t=z_h^n\ \text{for all }0\le t\le T_\delta\}.
$$
On $E_n$, the wealth path follows the high-state flow for time $T_\delta$.
By \eqref{eq:proof-lem-A-scaling-VII}, $a_0\in[0,\bar a_n(r)]$ almost surely, so that Step 3c implies
$$
|a_{T_\delta}-\bar a_n(r)|\le 2\delta\eta_n
\qquad\text{on }E_n.
$$
On $E_n^c$, we use only the crude bound $|a_{T_\delta}-\bar a_n(r)|\le M\eta_n $ from
\eqref{eq:proof-lem-A-scaling-VII}.
Therefore, using stationarity,
$$
\int_{\sX}|a-\bar a_n(r)|\,G_n^*(\da,\dz;r)
=
\E_{G^*_n}|a_{T_\delta}-\bar a_n(r)|
\le
2\delta\eta_n+M\eta_n\P(E_n^c).
$$
By stationarity,
$$
\P(E_n^c)
=
1-\pi_h^n e^{-\lambda_nT_\delta} 
\le
\pi_\ell^n+\lambda_nT_\delta.
$$
Consequently,
$$
\sup_{r\in K}
\frac1{\eta_n}
\int_{\sX}|a-\bar a_n(r)|\,G_n^*(\da,\dz;r)
\le
2\delta+M(\pi_\ell^n+\lambda_nT_\delta).
$$
Letting $n\to\infty$ and then $\delta\downarrow0$ yields, using $\pi^n_\ell,\lambda_n\to0$,
$$
\sup_{r\in K}
\left|
\frac{A_n(r)}{\eta_n}
-
\frac{\bar a_n(r)}{\eta_n}
\right|\leq 
\sup_{r\in K}
\frac1{\eta_n}
\int_{\sX}|a-\bar a_n(r)|\,G_n^*(\da,\dz;r)
\to0.
$$
Combining this with Step 2 proves \eqref{eq:A-scaling}.

\section{Proof of the comparison theorem} \label{app:comparison}
We now prove Theorem \ref{thm:comparison}, and we use an equivalent definition of viscosity solutions in terms of sub- and superdifferentials. A related comparison theorem can be found in \cite{gassiat2014investment}. We define
the superdifferential of a function $w$ at $x$ by
$$
D^+ w(x) := \big\{p\in\R:w(y)\leq w(x)+p(y-x)+o(|y-x|)\ \text{as} \ y\to x\big\}.
$$
Similarly, the subdifferential of $w$ at $x$ is defined by
$$
D^- w(x) := \big\{p\in\R:w(y)\geq w(x)+p(y-x)+o(|y-x|)\ \text{as} \ y\to x\big\}
$$
It is classical (see \cite{crandall1983viscosity}) that an equivalent characterization of a constrained viscosity solution $v$ of \eqref{eq:hh-problem-HJB} is the following:
\begin{enumerate}[(i)]
\item \emph{Subsolution property.} For every $(a,z)\in\bar\cO\times\cZ$ and $p\in D^+ v(\cdot,z)(a)$, $F(a,z,v(a,\cdot),p)\leq0$.
\item \emph{Supersolution property.} For every $(a,z)\in\cO\times\cZ$ and $p\in D^- v(\cdot,z)(a)$, $F(a,z,v(a,\cdot),p)\geq0$.
\end{enumerate}

\vspace{1em}

\noindent \emph{Step 1.} Recall that $\bar z\in\cZ$ denotes the maximal income level. Define the auxiliary function
$$
\psi:\bar\cO\longmapsto [0,\infty),\quad a\longmapsto
\begin{cases}
\exp(\eta (a-\lba)), & r\leq 0,\\
(ra+\bar z)^{\rho/r}, & r>0,
\end{cases}
$$
where $\eta:=\rho/(1+(r\lba+\overline{z})_+)$.
We claim that 
$$
w^\varepsilon (a,z) := w(a,z)+\varepsilon \psi(a),\quad (a,z)\in\bar\cO\times\cZ,
$$ 
is again a viscosity supersolution, for $\varepsilon>0$. Indeed, first observe that $D^-_a(w^\varepsilon ) = \{p+\varepsilon \psi'(a)\,:\,p\in D_a^-(w)\}$. Then, for any $(a,z)\in\cO\times \cZ$, $p\in D^-_a w(a,z)$,
$$
F(a,z,w^\varepsilon (a,\cdot),p+\varepsilon \psi'(a)) = F(a,z,w(a,\cdot),p)+ \varepsilon [\rho\psi(a)-\psi'(a)(ra+z)] + H(p)-H(p+\varepsilon \psi'(a)).
$$
We justify that each term is non-negative. By assumption, the first term satisfies $F(a,z,w(a,\cdot),p)\geq0$. The second term is non-negative by choice of $\psi$.
Finally, $H(p)-H(p+\varepsilon \psi'(a))\ge0$ since $H(\cdot)$ is decreasing, $\varepsilon \psi'(a)\ge0$, and $H(p)<\infty$. This establishes
$$
F(a,z,w^{\varepsilon }(a,\cdot),p+\varepsilon \psi'(a))\geq0,
$$
proving that $w^{\varepsilon }$ is a supersolution.\\

Clearly, if $v\leq w^\varepsilon$ for all $\varepsilon>0$, then $v\leq w$. Towards a contradiction, assume there exists $\varepsilon_*>0$ such that 
\begin{equation}\label{eq:proof-comparison-I}
M:=\sup_{\bar \cO\times\cZ}\ (v-w^{\varepsilon_*})>0.
\end{equation}
By the linear growth condition \eqref{eq:growth-condition} (in the case  $r=\rho$, the sublinear growth condition, respectively), 
$$
\lim_{a\to\infty } \ (v(a,z) - \varepsilon_*\psi(a)) =-\infty.
$$
Hence, for all $z\in\cZ$,
$$
v(a,z)-w^{\varepsilon_*}(a,z) \leq  v(a,z)-\inf w -\varepsilon_*\psi(a) \to -\infty,\quad \text{as}\ a\to\infty,
$$
so that there exist maximizers
$$
(a_*,z_*) \in \underset{\bar\cO\times\cZ}{\text{arg\,max}}\ (v-w^{\varepsilon_*}).
$$

\emph{Step 2.} For $k\ge1$, $a,a'\in\bar\cO$, and $z\in\cZ$, define 
$$
\Phi_k(a,a',z)= v(a,z)-w^{\varepsilon_*}(a',z) -\varphi_k(a,a'),\qquad  \varphi_k(a,a'):= \frac{k}{2} (a'-a)^2 +\frac12 [(k(a'-a) -1)^-]^2 
$$
Using the linear growth condition of $v$ and the definition of $w^{\varepsilon_*}$, we see that $\Phi_k$ attains its maximum over $\bar\cO\times\bar \cO\times \cZ$, and we denote
$$
(a_k,a_k',z_k)\in \underset{\bar\cO\times\bar \cO\times \cZ}{\text{arg\,max}} \ \Phi_k,\qquad k\geq1.
$$
Comparing this to $(a_*,a_*+1/k,z_*)$ yields
\begin{align*}
\Phi_k(a_k,a_k',z_k) &\geq \Phi_k(a_*,a_*+1/k,z_*)\\
&=v(a_*,z_*) - w^{\varepsilon_*}(a_*+1/k, z_*)-\frac{1}{2k}\\
&= M - \omega_k
\end{align*}
for $\omega_k:=w^{\varepsilon_*}(a_*+1/k, z_*)-w^{\varepsilon_*}(a_*,z_*)+1/(2k)$, which tends to $0$ by continuity of $w^{\varepsilon_*}$. This implies 
\be\label{eq:proof-comparison-II}
\underset{k\to\infty}{\liminf} \Phi_k(a_k,a_k',z_k)\geq M>0.
\ee
This implies that $(a_k,a_k')_{k\ge1}$ is a bounded sequence. Using the construction, we may assume that, after passing to a subsequence,
\be\label{eq:proof-comparison-III}
(a_k,a_k')\to(\hat a,\hat a),\qquad z_k=\hat z,\qquad \varphi_k(a_k,a_k')\to \ell,
\ee
for some $(\hat a,\hat z)\in\bar\cO\times \cZ$ and $\ell\ge0$. By continuity and \eqref{eq:proof-comparison-II}, 
$$
M\leq v(\hat a,\hat z)-w^{\varepsilon_*}(\hat a,\hat z)-\ell\leq M-\ell,
$$
so that $\ell=0$ and $v(\hat a,\hat z)-w^{\varepsilon_*}(\hat a,\hat z)=M$. In particular,
$$
q_k := [k(a_k'-a_k)-1]^- \to 0.
$$
For sufficiently large $k\ge1$, we have $q_k\le1/2$ so that
$$
k(a_k'-a_k)-1\ge-\frac12 \qquad \Rightarrow\qquad a_k'\geq a_k+\frac{1}{2k}\geq \lba + \frac{1}{2k} >\lba.
$$
Therefore, one may use the supersolution property of $w^{\varepsilon_*}$ at $a_k'$. For such $k$, use $a\mapsto \varphi_k(a,a_k')$ as a test function for the subsolution $v$ at $(a_k,z_k)$, and use $a'\mapsto-\varphi_k(a_k,a')$ as a test function for the supersolution $w^{\varepsilon_*}$ at $(a_k',z_k)$. Set 
$$
p_k := \partial_a \varphi_k(a_k,a_k'),\qquad p_k' := - \partial_{a'} \varphi_k(a_k,a_k').
$$
A direct calculation shows the identity 
\be\label{eq:proof-comparison-IV}
p_k=p_k'=k(a_k-a_k')+kq_k.
\ee
Using the viscosity properties,
\be\label{eq:proof-comparison-V}
F(a_k,z_k,v(a_k,\cdot),p_k)\le 0,\qquad F(a_k',z_k,w^{\varepsilon_*}(a_k',\cdot),p_k)\ge 0.
\ee
By the definition of $(a_k,a_k',z_k)$ as maximizers of $\Phi_k$, 
$$
v(a_k,y) - w^{\varepsilon_*}(a_k',y)\leq v(a_k,z_k)-w^{\varepsilon_*}(a_k',z_k),\qquad \forall y\in\cZ,
$$
so that 
$$
\sL v(a_k,\cdot )(z_k)-\sL w^{\varepsilon_*}(a_k',\cdot )(z_k)\le0.
$$
Using this, the fact that $H(p_k)<\infty$, and subtracting the inequalities in \eqref{eq:proof-comparison-V} leads to
$$
\rho(v(a_k,z_k)-w^{\varepsilon_*}(a_k',z_k))\leq p_k(r a_k+z_k) - p_k(ra_k'+z_k)= rp_k(a_k-a_k').
$$
Since $\Phi_k(a_k,a_k',z_k)\geq M-\omega_k$ and $\varphi_k\ge0$,
$$
v(a_k,z_k)-w^{\varepsilon_*}(a_k',z_k)\geq M-\omega_k,\qquad \Rightarrow\qquad \rho(M-\omega_k)\leq rp_k(a_k-a_k').
$$
It remains to show that $p_k(a_k-a_k')\to0$ as $k\to\infty$ as this would yield $\rho M\le0$, contradicting \eqref{eq:proof-comparison-I}.
Indeed, a direct computation yields, with $d_k:=a_k-a_k'$,
$$
p_k\,d_k=k\,d_k^2 + k\,q_k\,d_k.
$$
The first term goes to zero since $\varphi_k(a_k,a_k')\to0$ by \eqref{eq:proof-comparison-III}. For the second term, observe that $q_k\to0$ by \eqref{eq:proof-comparison-III}. Now, if $q_k>0$, then $q_k=1+kd_k$, so that $kq_kd_k=q_k(q_k-1)$, which tends to zero. This establishes $p_kd_k\to0$ as desired.\hfill $\Box$

\section{Proof of Lemma \ref{lem:properties-saving-rules} for $r=\rho$} \label{app:r-equal-rho}

We continue to establish \eqref{lem:properties-saving-rules-I} and \eqref{lem:properties-saving-rules-II} of Lemma \ref{lem:properties-saving-rules} in the case $r=\rho$. By the continuity property of $v_a^*$ in the interest rate recorded in Lemma \ref{lem:joint-continuity-(a,r)}, we have
\be\label{eq:appendix-savings-I}
s^*(a,\lbz)\leq 0,\qquad s^*(\lba,\lbz)=0,\qquad a\geq\lba.
\ee

\emph{Proof of Lemma \ref{lem:properties-saving-rules} \eqref{lem:properties-saving-rules-I}.} Towards a contradiction, suppose that $s^*(a_0,\lbz)=0$ for some $a_0>\lba$ and let $z\in \mathrm{arg\,max}\{v^*_a(a_0,z):z\in\cZ\}$. Following the reasoning in the proof of Lemma \ref{lem:properties-saving-rules}, we see that this necessitates $\sL v^*_a(a_0,\cdot)(z)=0$. Using irreducibility of the income process, this implies that $v^*_a(a_0,\cdot)$ is constant for all $z\in\cZ$, and we denote the common value by $p\ge0$. In particular, this implies $c^*(a_0,z)=c^*(a_0,\lbz)$, so that $s^*(a_0,z)= ra_0+z-c(a_0,\lbz)=z-\lbz$. Hence, there exist $\varepsilon,\eta>0$ such that 
\be\label{eq:appendix-savings-II}
s^*(a,z)\geq\eta \qquad \forall a\in[a_0,a_0+\varepsilon],\qquad \forall z\neq\lbz.
\ee
Define 
$$
q(a,z):=p-v^*_a(a,z),\qquad Q(a):=\max_{z\neq \lbz} q(a,z),\qquad a\in [a_0,a_0+\varepsilon].
$$
By concavity of $v^*(\cdot,z)$, each $v^*_a(\cdot,z)$ is decreasing so that $q(a,z)\geq0$ and $Q(a_0)=0$. 

Next, for $a\in(a_0,a_0+\varepsilon]$, Lemma \ref{lem:Euler} and \eqref{eq:appendix-savings-I} imply that 
\be\label{eq:appendix-savings-III}
\sum_{y\neq \lbz} \lambda(\lbz,y) (q(a,\lbz)-q(a,y)) = \sL v^*_a(a,\cdot)(\lbz)= - v_{aa}^*(a,\lbz) s^*(a,\lbz)\leq0\quad \Rightarrow\quad q(a,\lbz)\leq Q(a).
\ee

By Lemma \ref{lem:Euler} and \eqref{eq:appendix-savings-II}, for $z\neq \lbz$ and $a\in[a_0,a_0+\varepsilon]$,
$$
\partial_a q(a,z) = \frac{\sL v^*_a(a,\cdot)(z)}{s^*(a,z)} = \frac{\sum_{y\neq z} \lambda(z,y) (q(a,z)-q(a,y))}{s^*(a,z)}\leq c_* Q(a)
$$
for some $c_*<\infty$. Therefore, the right-sided derivative of $Q$ satisfies $D^+ Q(a)\leq c_* Q(a)$,  $Q(a_0)=0$, and by Gr\"{o}nwall's inequality, $Q\equiv0$ on $[a_0,a_0+\varepsilon]$. By \eqref{eq:appendix-savings-III}, $q(a,\lbz)=0$ as well. This shows that the optimal consumption rule $c^*(a,z)$ is constant and equal to $\rho a_0+\lbz$ on $[a_0,a_0+\varepsilon]\times\cZ$, so that $s^*(a,\lbz)=\rho(a-a_0)>0$. This contradicts \eqref{eq:appendix-savings-I}, and hence establishes $s^*(a,\lbz)<0$ for every $a>\lba$ and $s^*(\lba,\lbz)=0$.\hfill $\Box$

\vspace{1em}

\emph{Proof of Lemma \ref{lem:properties-saving-rules} \eqref{lem:properties-saving-rules-II}.} Since $s^*(\lba,\lbz)=0$, Lemma \ref{lem:Euler} implies $\sL v^*_a(\lba,\cdot)(\lbz)\leq 0$. Towards a contradiction, assume that equality holds. Following the proof of Lemma \ref{lem:properties-saving-rules} \eqref{lem:properties-saving-rules-I}, this implies that $c^*(a,z)$ is constant in a right-neighborhood of $\lba$. This implies that $s^*(a,\lbz)=\rho(a-\lba)>0$ for $a>\lba$ close to $\lba$, contradicting \eqref{lem:properties-saving-rules-I}. This shows that $\sL v^*_a(\lba,\cdot)(\lbz)<0$.\hfill $\Box$

\bibliographystyle{abbrvnat}
\bibliography{bibliography}

\end{document}